%% file: main.tex
\documentclass[11pt]{article}

\usepackage[utf8]{inputenc} % allow utf-8 input
\usepackage{url}
\usepackage{wrapfig}
\usepackage[authoryear]{natbib}
\usepackage{bigints}
\usepackage{tikz}
\usetikzlibrary{arrows.meta,arrows, shapes, positioning}
\usepackage{siunitx}
\usepackage[colorlinks,linkcolor=magenta,citecolor=blue]{hyperref}
\usepackage{graphicx}
\usepackage{amsmath,amssymb,amsthm,fullpage,graphicx}
\usepackage{comment}
\usepackage{color}
\usepackage{setspace}
\usepackage{pdflscape}
\usepackage{enumerate}
\usepackage{enumitem}
\usepackage{multirow}
\usepackage{multicol}

\usepackage{placeins}

\usepackage{bbm}
\usepackage[english]{babel}

\usepackage{macros}
\usepackage{cancel}

\usepackage{booktabs}

\makeatletter
\renewcommand\tableofcontents{%
    \@starttoc{toc}%
}
\makeatother

    \usepackage{thmtools}

\newlist{thmlist}{enumerate}{1}
\setlist[thmlist]{label=(\roman{thmlisti}), ref=\thetheorem(\roman{thmlisti})}

\definecolor{larrypurple}{RGB}{144, 3, 252}

\newtheorem{theorem}{Theorem}

\usepackage{algorithm2e}
\usepackage[noend]{algpseudocode}
\usepackage{nameref,hyperref}
\usepackage[capitalize]{cleveref}

\Crefname{theorem}{Theorem}{Theorems}
\Crefname{lemma}{Lemma}{Lemmas}

\addtotheorempostheadhook[theorem]{\crefalias{thmlisti}{theorem}}
\addtotheorempostheadhook[lemma]{\crefalias{thmlisti}{lemma}}

\usepackage{amsmath}
\makeatletter
\let\save@mathaccent\mathaccent
\newcommand*\if@single[3]{%
  \setbox0\hbox{${\mathaccent"0362{#1}}^H$}%
  \setbox2\hbox{${\mathaccent"0362{\kern0pt#1}}^H$}%
  \ifdim\ht0=\ht2 #3\else #2\fi
  }
\newcommand*\rel@kern[1]{\kern#1\dimexpr\macc@kerna}
\newcommand*\widebar[1]{\@ifnextchar^{{\wide@bar{#1}{0}}}{\wide@bar{#1}{1}}}
\newcommand*\wide@bar[2]{\if@single{#1}{\wide@bar@{#1}{#2}{1}}{\wide@bar@{#1}{#2}{2}}}
\newcommand*\wide@bar@[3]{%
  \begingroup
  \def\mathaccent##1##2{%
    \let\mathaccent\save@mathaccent
    \if#32 \let\macc@nucleus\first@char \fi
    \setbox\z@\hbox{$\macc@style{\macc@nucleus}_{}$}%
    \setbox\tw@\hbox{$\macc@style{\macc@nucleus}{}_{}$}%
    \dimen@\wd\tw@
    \advance\dimen@-\wd\z@
    \divide\dimen@ 3
    \@tempdima\wd\tw@
    \advance\@tempdima-\scriptspace
    \divide\@tempdima 10
    \advance\dimen@-\@tempdima
    \ifdim\dimen@>\z@ \dimen@0pt\fi
    \rel@kern{0.6}\kern-\dimen@
    \if#31
      \overline{\rel@kern{-0.6}\kern\dimen@\macc@nucleus\rel@kern{0.4}\kern\dimen@}%
      \advance\dimen@0.4\dimexpr\macc@kerna
      \let\final@kern#2%
      \ifdim\dimen@<\z@ \let\final@kern1\fi
      \if\final@kern1 \kern-\dimen@\fi
    \else
      \overline{\rel@kern{-0.6}\kern\dimen@#1}%
    \fi
  }%
  \macc@depth\@ne
  \let\math@bgroup\@empty \let\math@egroup\macc@set@skewchar
  \mathsurround\z@ \frozen@everymath{\mathgroup\macc@group\relax}%
  \macc@set@skewchar\relax
  \let\mathaccentV\macc@nested@a
  \if#31
    \macc@nested@a\relax111{#1}%
  \else
    \def\gobble@till@marker##1\endmarker{}%
    \futurelet\first@char\gobble@till@marker#1\endmarker
    \ifcat\noexpand\first@char A\else
      \def\first@char{}%
    \fi
    \macc@nested@a\relax111{\first@char}%
  \fi
  \endgroup
}
\makeatother
\def\indep{\perp\!\!\!\perp}

\def\E{\mathbb{E}}
\def\Var{\text{Var}}

\def\one{\mathbbm{1}}
\def\d{\,d}

\def\calR{\mathcal{R}}

\def\za2{z_{\alpha/2}}

\renewcommand{\leq}{\leqslant}
\renewcommand{\geq}{\geqslant}

\def\opt{\text{opt}}

\def\uia{u_{i, \mathrm{A}}}
\def\uit{u_{i, \mathrm{T}}}
\def\uiw{u_{i, \mathrm{W}}}

\def\Uia{U_{i, \mathrm{A}}}
\def\Uit{U_{i, \mathrm{T}}}
\def\Uiw{U_{i, \mathrm{W}}}
\def\Uix{U_{i, \mathrm{X}}}
\def\Uiy{U_{i, \mathrm{Y}}}

\def\ua{u_{\mathrm{A}}}
\def\ut{u_{\mathrm{T}}}
\def\uw{u_{\mathrm{W}}}

\def\Ua{U_{\mathrm{A}}}
\def\Ut{U_{\mathrm{T}}}
\def\Uw{U_{\mathrm{W}}}

\def\Uy{U_{\mathrm{Y}}}

\def\phia{\phi_{\mathrm{A}}}

\def\phit{\phi_{\mathrm{T}}}
\def\phiw{\phi_{\mathrm{W}}}
\def\phix{\phi_{\mathrm{X}}}
\def\phiy{\phi_{\mathrm{Y}}}
\def\phiu{\phi_{\mathrm{U}}}

\usepackage{titlesec}

\titleformat{\paragraph}
{\normalfont\normalsize\bfseries}{\theparagraph}{1em}{}
\titlespacing*{\paragraph}
{0pt}{3.25ex plus 1ex minus .2ex}{1.5ex plus .2ex}

\RequirePackage[authoryear]{natbib}

\newcolumntype{C}[1]{>{\centering\arraybackslash}p{#1}}

\begin{document}

\large{

\title{\bf Estimating Optimal Dynamic Treatment Regimes Using Irregularly Observed Data: A Target Trial Emulation and Bayesian Joint Modeling Approach}
  \author{
      Larry Dong$^{1, 2}$, Eleanor Pullenayegum$^{1, 2}$, Rodolphe Thi{\'e}baut$^{3, 4}$ and Olli Saarela$^{1}$
  }
  \date{
      \small
      $^1$Dalla Lana School of Public Health, University of Toronto, Toronto, Canada\\
      $^2$Child Health Evaluation Services, The Hospital for Sick Children, Toronto, Canada\\
      $^3$Department of Public Health, University of Bordeaux, Inserm U1219 Bordeaux Population Health Research Center, Inria Statistics in Systems Biology and Translational Medicine (SISTM), Bordeaux, France\\
      $^4$Department of Medical information, CHU Bordeaux, Bordeaux, France
    }
  \maketitle
}

\begin{abstract}
    An optimal dynamic treatment regime (DTR) is a sequence of decision rules aimed at providing the best course of treatments individualized to patients. While conventional DTR estimation uses longitudinal data, such data can also be irregular, where patient-level variables can affect visit times, treatment assignments and outcomes. In this work, we first extend the target trial framework -- a paradigm to estimate statistical estimands specified under hypothetical randomized trials using observational data -- to the DTR context; this extension allows treatment regimes to be defined with intervenable visit times. We propose an adapted version of G-computation marginalizing over random effects for rewards that encapsulate a treatment strategy's value. To estimate components of the G-computation formula, we then articulate a Bayesian joint model to handle correlated random effects between the outcome, visit and treatment processes. We show via simulation studies that, in the estimation of regime rewards, failure to account for the observational treatment and visit processes produces bias which can be removed through joint modeling. We also apply our proposed method on data from INSPIRE 2 and 3 studies to estimate optimal injection cycles of Interleukin 7 to treat HIV-infected individuals.
\end{abstract}

\noindent%

\newpage

\tableofcontents

\newpage

\input{chapters/1-introduction}

\input{chapters/2.1-notational-framework}

\input{chapters/2.2-relate-observational-experimental}

\input{chapters/2.3-estimation}

\input{chapters/3-simulation}

\input{chapters/4-data-application}

\input{chapters/5-discussion}

\bibliographystyle{Chicago}

\bibliography{Bibliography-MM-MC}

\appendix

\clearpage

\bigskip
\begin{center}
{\Huge\bf Supplementary Materials to ``Estimating Optimal Dynamic Treatment Regimes Using Irregularly Observed Data: A Target Trial Emulation and Bayesian Joint Modeling Approach''}
\end{center}

\input{appendix/appendix-C-optimal-decision}

\input{appendix/appendix-A-reward-derivation}

\input{appendix/appendix-B-g-computation}

\input{appendix/appendix-D-posterior-rewards}

\input{appendix/appendix-E-simulation-algorithm}

\input{appendix/appendix-F-simulation-results}

\input{appendix/appendix-G-data-analysis}

\end{document}

%% file: chapters/1-introduction.tex
\section{Introduction}

A dynamic treatment regime (DTR), also known as an adaptive treatment strategy, is a sequence of deterministic decision rules which, at every stage, assigns a treatment given an individual's available history \citep{murphy2003optimal, moodie2007demystifying}. In particular, interest lies in estimating the optimal DTR, which maximizes a quantity that can comprise of long-term and short-term outcomes \citep{murphy2003optimal}. In the context of precision medicine, these regimes allow decision-makers to tailor treatments with respect to a patient's evolving health status and response to their previous treatments. %For instance, when considering chronic diseases or conditions with progressive stages, a ``one-size-fits-all'' approach often yields sub-optimal patient-centric outcomes.

Sequential Multiple Assignment Randomized Trials (SMART) are study designs in which participants are randomized at every stage of a treatment sequence \citep{lavori2000design, lei2012smart}. SMARTs are ideal for developing DTRs but such trials are often impractical in practice due to logistical, ethical, and financial constraints \citep{wallace2016smart}. As a result, researchers often rely on observational data to infer causal estimands and optimal DTRs \citep{lavori2000design, lei2012smart}. G-computation and marginal structural models (MSMs) have been used to estimate causal quantities with observational data, accounting for time-varying confounders \citep{hernan2020causal, robins2000marginal}. However, observational longitudinal data can also be subject to irregular observation times, which can introduce bias if not properly accounted for \citep{zeger1986longitudinal, pullenayegum2016longitudinal}. Methods to account for irregularly observed data are generally built on generalized estimating equations (GEE) or joint modeling \citep{pullenayegum2016longitudinal, gasparini2020mixed}. In the former, bias induced by informative visit times can be corrected using inverse intensity weights (IIW). In the Bayesian paradigm, but not limited to it, methods for accounting for irregularly observed data have relied on shared or correlated random effects by jointly modeling the outcome and visit processes \citep{farewell2017ignorability, ryu2007longitudinal}.

Bayesian methods have increasingly been applied to the estimation of DTRs, offering flexible frameworks to account for time-varying confounders and the complexities of longitudinal data. These approaches include two-step quasi-Bayesian methods, Bayesian extensions of MSMs, and techniques that utilize inverse probability treatment weighting (IPTW) within a Bayesian framework \citep{hoshino2008bayesian, keil2018bayesian, saarela2015bayesian}. Recent advances have focused on estimating optimal DTRs, mainly using Bayesian methods, tailored for irregularly observed data \citep{coulombe2023estimating, guan2020bayesian, hua2022personalized, oganisian2022bayesian}. However, despite such recent methodological developments, there is limited work on using Bayesian joint models with correlated random effects with irregularly observed data to estimate optimal DTRs with intervenable visit times. One way to address this methodological problem is to build on the framework of the target trial, a hypothetical randomized trial whose purpose is to define causal quantities of interest such that they can be estimated using observational data \citep{hernan2011great, hernan2016using}. The DTR estimation problem has usually been formulated in the potential outcomes framework. Alternatively, the DTR estimation problem based on observational, possibly irregular, longitudinal data could be framed with respect to an ideal SMART by extending the target trial emulation framework. However, it appears that this has not been fully explored in the literature. In our work, we address these gaps by introducing a new yet parsimonious target trial notation and we show how optimal DTRs can be estimated using Bayesian joint modeling when sub-processes affect each other.

%% file: chapters/2.1-notational-framework.tex
\section{Methods}
\label{section:methods}

\subsection{Notation}
\label{subsection:notation}

Akin to \cite{yiu2022causal}, we begin by simultaneously using the target trial paradigm from \cite{hernan2016using} and the mathematical notation from \cite{dawid2010identifying}'s decision-theoretic framework. The target trial concept provides a formal causal framework that allowing explicit definition of each component of the trial under which we wish to define our causal estimand of interest \citep{hernan2016using, hernan2022target}. We then extend these frameworks by incorporating visit times information into both the target trial concept and the decision-theoretic mathematical notation.

% While conventional notation to define DTRs have been developed for sequential treatments, but requires extending to handle visit time information \citep{dawid2010identifying, yiu2022causal}. 

We denote outcomes $Y_{ij}$, covariates $X_{ij}$ with $W_{ij} = (Y_{ij}, X_{ij})$, treatments $A_{ij}$ and visit times $T_{ij}$ where $j \in \mathbb{N}$ indexes the observation number of an individual $i \in \mathbb{N}$. We denote $U_i = (\Uiy, \Uix, \Uia, \Uit)$ to be individual-specific random effects corresponding to each process and $\zeta_i$ to be the right-censoring time. Hereon, unless we talk about estimation procedures, we omit that subscript $i$ for notational simplicity, with all defined variables understood to be specific to any individual $i \in \mathbb{N}$. The overbar notation, e.g. $\overline{t}_{j} = (t_{1}, \dots, t_{j})$, denotes a set of accrued variables and the underline notation, e.g. $\underline{t}_j = (t_j, \dots, t_{K+1})$, represents accrued future variables; without loss of generality, we assume that $\overline{t}_0 = \emptyset$ and the same for other variables. DTRs consist of decision functions $d_j$ where the corresponding treatment $d_j(h_j) = a_j \in \Psi_j(h_j)$ is assigned according to the history $h_j = (\overline{t}_j, \overline{w}_{j}, \overline{a}_{j-1})$ among possible options $\Psi_j(h_j)$, also referred as the feasible set \citep{tsiatis2019dynamic}. Furthermore, we define the optimal midstream DTR $\underline{d}_k^{\opt} = \big(d^{\opt}_{k}, \dots, d_{K}^{\opt}\big)$ maximizing a reward conditional on past history $H_{k} = h_k$. Optimal stage $j$ decisions $d_j^{\opt}$ are defined recursively assuming that future decisions $d_{j+1}^\opt, \dots, d_K^\opt$ are optimal. % denoted by $\calR_{k}(h_k, \underline{t}_{k+1}, \underline{d}_k, \phi_\calR)$, which is conditional on history $H_{ik}=h_k$ accrued from $k$ visits, future visit times $\underline{t}_{k+1} = (t_{k+1}, t_{K+1})$

\subsection{Target Trial Definition}
\label{subsection:target-trial-definition}

We define $\calR_{k}(h_k, \underline{t}_{k+1}, \underline{d}_k, \phi_\calR)$ as the reward which, conditional on past history $H_{k} = h_k$, is a function of expected outcomes from stages $k+1$ to $K+1$; we use $\phi_\calR$ to represent reward parameters which we will expand in Section~\ref{subsection:relate-e-o}:
\begin{align*}
    % \calR^\mathrm{M}(\overline{t}_{K+1}, \overline{d}_{K}, \phi_\calR) &= \sum_{j=2}^{K+1} \gamma_j \E_\calE\left[Y_{ij} \mid \overline{T}_{ij} = \overline{t}_j, \overline{A}_{i(j-1)} = \overline{d}_{j-1}, \phi_\calR \right]\\
    \calR_{k}(h_k, \underline{t}_{k+1}, \underline{d}_{k}, \phi_\calR) &= \sum_{j=k+1}^{K+1} \gamma_j \E_\calE\left[Y_{j} \mid \overline{T}_{j} = \overline{t}_j, \overline{A}_{j-1} = \overline{d}_{j-1}, \overline{W}_{k}=\overline{w}_{k}, \phi_\calR \right] \,,
\end{align*}
for any $i \in \mathbb{N}$. The conditional statement $A_j = d_j$ means that the random variable $A_j$ is assigned the treatment value of $d_j(h_j)$ and $\overline{A}_{j-1} = \overline{d}_{j-1}$ means that the first $k-1$ treatments $\overline{a}_{k-1} \subset h_k$ are conditioned on and future treatments $a_{k}, \dots, a_j$ are assigned according to the decision rules $d_{k}, \dots, d_j$. Notice that $\calR_{k}(h_k, \underline{t}_{k+1}, \underline{d}_{k}, \phi_\calR)$ is marginal over $W_{k+1}, \dots, W_{K+1}$ but are conditional on $\overline{W}_{k} = \overline{w}_k \subset h_k$. With the reward defined under $\calE$ given $\phi_\calR$, $d_j^\opt$ can be defined recursively through backwards induction where we are primarily interested in the optimal decision function under true parameters $\phi_{\calR 0}$, i.e. $d_j^{\opt}(h_j, \phi_{\calR 0})$:
\begin{align*}
    d_j^{\opt}(h_j, \phi_{\calR}) &= \argmax_{a_j \in \Psi_j(h_j)} \, \calR_{j}\Big(h_j, \underline{t}_{j+1}, \big(a_j, \underline{d}_{j+1}^{\opt}\big), \phi_{\calR} \Big) \quad \text{for } j = k, \dots, K - 1 \, ,
\end{align*}
with $d_K^{\opt}(h_K, \phi_{\calR}) = \argmax_{a_K \in \Psi_K(h_K)} \, \calR_{K}\Big(h_K, t_{K+1}, \overline{a}_K, \phi_{\calR} \Big)$ being the terminal case. We expand on the recursive nature of $d_j^\opt$'s in Appendix~\ref{appendix-1:optimal-decision-functions}.

\input{chapters/dag-1}

% Without loss of generality and for notational simplicity, we present our work hereon using $W_j, A_j$ as continuous variables with all integrals and causal assumptions written with respect to corresponding densities. Integrals and such densities can be respectively interpreted as general Riemann-Stieltjes integrals and probability distributions.

% Considering $\calE$ as SMART, the trial we wish to emulate to estimate optimal DTRs, we posit that treatments $A_{j+1}$ and observation times $T_{j+1}$ do not depend on individual-specific information under $\calE$ \cite{lavori2000design}. This randomization procedure can be generalized to depend on accrued observable individual-level characteristics, but, for simplicity, we consider $\calE$ as a completely randomized SMART hereon.
We model $\calE$ as a SMART trial for estimating optimal DTRs, assuming that treatments $A_{j+1}$ and observation times $T_{j+1}$ are randomized independently of individual-specific information \cite{lavori2000design}. Although randomization could incorporate individual characteristics and even be combined with deterministic decision rules, we focus here on a completely randomized SMART as our $\calE$ for simplicity:
\begin{align*}
    \overline{W}_{j+1}, U \indep_\calE T_{j+1} \mid \overline{T}_{j}, \overline{A}_{j}, \phit^* \,\,\, \text{and} \,\,\, \overline{W}_{j}, U \indep_\calE A_{j} \mid \overline{T}_{j}, \overline{A}_{j-1}, \phia^*\, , \quad \text{for } j = 1, \dots, K\,,
\end{align*}
where $\phit^*, \phia^*$ are parameters respectively driving the target trial visit and treatment assignment processes. Equivalently, these conditions can be specified in $\calE-$densities:
\begin{align*}
    f_\calE(a_{j+1}, t_{j+1} \mid  \overline{w}_{j+1}, \overline{a}_{j}, \overline{t}_{j}, u, \phia^*, \phit^*) = f_\calE(a_{j+1}, t_{j+1} \mid \overline{a}_{j}, \overline{t}_{j}, \phia^*, \phit^*) \, ,
\end{align*}
for all $j = 1, \dots, K - 1$ and $f_\calE(t_{K+1} \mid \overline{a}_{K}, \overline{w}_{K}, \overline{t}_{K}, u, \phit^*)$ = $f_\calE(t_{K+1} \mid \overline{a}_{K}, \overline{t}_{K}, \phit^*)$. Under these trial specifications, $\calR_{k}(h_k, \underline{t}_{k+1}, \underline{d}_k, \phi_\calR)$ can be written as follows:
\begin{align*}
    \calR_{k}(h_k, \underline{t}_{k+1}, \underline{d}_k, \phi_\calR)
    &= \bigintsss_{\uw} \left\{ \sum_{j=k+1}^{K+1} \gamma_j \E_\calE\left[Y_{j} \mid \overline{T}_{j} = \overline{t}_j, \overline{A}_{j-1} = \overline{d}_{j-1}, \overline{W}_k=\overline{w}_{k}, \uw, \phiw \right] \right\}\\
    &\qquad \quad \times f_\calE(\uw \mid \overline{t}_{k}, \overline{w}_k, \overline{d}_{k-1}, \phi_\calR) \d \uw \\
    &= \bigintsss_{\uw} \bigintsss_{w_{j+1}} \dots \bigintsss_{w_{K+1}} \sum_{j=k+1}^{K+1} \gamma_j y_j\\
    &\quad \times \frac{\left\{\prod_{j=k+1}^{K+1} f_\calE(w_j \mid \overline{t}_{j}, \overline{d}_{j-1}, \overline{w}_{j-1}, \uw, \phiw)\right\} f_\calE(\uw \mid \phiu)}{\bigintss_{\uw} \left\{\prod_{j=1}^{k} f_\calE(w_j \mid \overline{t}_{j}, \overline{d}_{j-1}, \overline{w}_{j-1}, \uw, \phiw) \right\} f_\calE(\uw \mid \phiu) d \uw} \\
    &\quad \times d w_{K+1} \dots d w_{k+1} d \uw \,.
\end{align*}
The full derivation is available in the Appendix~\ref{appendix-1:definition-rewards}. While we consider $W_j$ as continuous variables, all integrals can be generalized by defining them as Riemann-Stieltjes integrals and cumulative distribution functions. By factorizing the rewards as above, we find that optimal rewards and regimes only depend on two sets of parameters, hence $\phi_\calR = (\phiw, \phiu)$. The dotted arrows in Figure~\ref{fig:dag-project-1} linking random effects indicates their possible dependence. The absence of correlation between the random effects, reflecting unobserved individual-level variation, would render certain model components ignorable for outcomes. This would allow the likelihood to factorize into separate components and allowing inference to be carried out without modeling the treatment and visit processes \citep{farewell2017ignorability}. In particular, in our work, we take interest in the scenario where this is not the case.

%% file: chapters/dag-1.tex
\begin{figure}[!ht]
    \centering
    \begin{minipage}{.5\textwidth}
        \begin{center}
            \begin{tikzpicture}[scale=0.6, every node/.style={scale=0.6}, >=stealth, shorten >=1pt, main node/.style={circle,draw,font=\large\bfseries}]
                \node[main node] at (-4, 0) (W1) {$W_{1}$};
                \node[main node] at (-3, -3) (A1) {$A_{1}$};
                \node[main node] at (-4.5, -5) (T1) {$T_{1}$};
                \node[main node] at (1, 0) (W2) {$W_{2}$};
                \node[main node] at (1.5, -3) (A2) {$A_{2}$};
                \node[main node] at (0, -5) (T2) {$T_{2}$};
                \node[main node] at (5, 0) (W3) {$W_{3}$};
                \node[main node] at (4, -5) (T3) {$T_{3}$};
                \node[main node] at (0.5, 2) (bW) {$\Uw$};
                \node[main node] at (-0.5, -8) (bA) {$\Ua$};
                \node[main node] at (1.5, -8) (bT) {$\Ut$};
                \node[main node] at (-6, -3) (b) {$U^*$};
        
                \draw[->]
                    (W1) edge (W2)
                         edge[out=25, in=155] (W3)
                    (A1) edge (W2)
                         edge (A2)
                         edge[out=330, in=250] (W3)
                         edge (T2)
                         edge (T2)
                    (W2) edge (W3)
                    (A2) edge (W3)
                         edge (T3);
                \draw[->, teal]
                    (T1) edge (W1)
                         edge (A1)
                         edge[out=30, in=240] (W2)
                         edge (A2)
                         edge (T2)
                         edge[out=347, in=270, looseness=2.2] (W3)
                         edge[out=353, in=220] (T3)
                    (T2) edge[out=30, in=260] (W3)
                         edge (A2)
                         edge (W2)
                         edge (T3)
                    (T3) edge (W3);
                \draw[->, orange]
                    (bW) edge[out=180, in=70] (W1)
                         edge (W2)
                         edge[out=360, in=110] (W3);
                \draw[->, dashed, orange]
                    (b) edge[out=280, in=170] (bA)
                        edge[out=270, in=220] (bT)
                        edge[out=90, in=170, looseness=1.2] (bW);
        
                \draw[dashed]
                    (-1.5, -7) -- (-1.5, 1.5)
                    (3, -7) -- (3, 1.5);

                \node[text width=3cm] at (4.5, -10) {\Large $\calE$ graph};
            \end{tikzpicture}
        \end{center}
    \end{minipage}
    \vrule
    \hspace{0.5cm}
    \begin{minipage}{.45\textwidth}
        \begin{center}
            \begin{tikzpicture}[scale=0.6, every node/.style={scale=0.6}, >=stealth, shorten >=1pt, main node/.style={circle,draw,font=\large\bfseries}]
                \node[main node] at (-4, 0) (W1) {$W_{1}$};
                \node[main node] at (-3, -3) (A1) {$A_{1}$};
                \node[main node] at (-4.5, -5) (T1) {$T_{1}$};
                \node[main node] at (1, 0) (W2) {$W_{2}$};
                \node[main node] at (1.5, -3) (A2) {$A_{2}$};
                \node[main node] at (0, -5) (T2) {$T_{2}$};
                \node[main node] at (5, 0) (W3) {$W_{3}$};
                \node[main node] at (4, -5) (T3) {$T_{3}$};
                \node[main node] at (0.5, 2) (bW) {$\Uw$};
                \node[main node] at (-0.5, -8) (bA) {$\Ua$};
                \node[main node] at (1.5, -8) (bT) {$\Ut$};
                \node[main node] at (-6, -3) (b) {$U^*$};
                \draw[->]
                    (W1) edge (A1)
                         edge (W2)
                         edge (A2)
                         edge[out=25, in=155] (W3)
                         edge (T2)
                         edge[out=320, in=160] (T3)
                    (A1) edge (W2)
                         edge (A2)
                         edge (T2)
                         edge[out=330, in=250] (W3)
                         edge (T3)
                    (W2) edge (W3)
                         edge (A2)
                         edge (T3)
                    (A2) edge (W3)
                         edge (T3);
                \draw[->, teal]
                    (T1) edge (W1)
                         edge (A1)
                         edge[out=30, in=240] (W2)
                         edge (A2)
                         edge (T2)
                         edge[out=347, in=270, looseness=2.2] (W3)
                         edge[out=353, in=220] (T3)
                    (T2) edge[out=30, in=260] (W3)
                         edge (A2)
                         edge (W2)
                         edge (T3)
                    (T3) edge (W3);
                \draw[->, orange]
                    (bW) edge[out=180, in=70] (W1)
                         edge (W2)
                         edge[out=360, in=110] (W3)
                    (bA) edge[out=130, in=280] (A1)
                         edge[out=50, in=270] (A2)
                    (bT) edge (T1)
                         edge (T2)
                         edge (T3);
                \draw[->, dashed, orange]
                    (b) edge[out=280, in=170] (bA)
                        edge[out=270, in=220] (bT)
                        edge[out=90, in=170, looseness=1.2] (bW);
                    
                \draw[dashed]
                    (-1.5, -7) -- (-1.5, 1.5)
                    (3, -7) -- (3, 1.5);

                \node[text width=3cm] at (4.5, -10) {\Large $\calO$ graph};
            \end{tikzpicture}
        \end{center}
        \end{minipage}
    \caption{Left: DAG for the observational ($\calO$) data-generating process where three visits occur before censoring time. Right: DAG representing the target trial model ($\calE$) of a two-stage regime in which treatments and observation times are no longer driven by past covariates, outcomes and latent random effects. Both DAGs encode the assumption that the random effects $U = (\Uw, \Ua, \Ut)$ are independent conditional on $U^*$, a latent factor that affects them.}
    \label{fig:dag-project-1}
\end{figure}
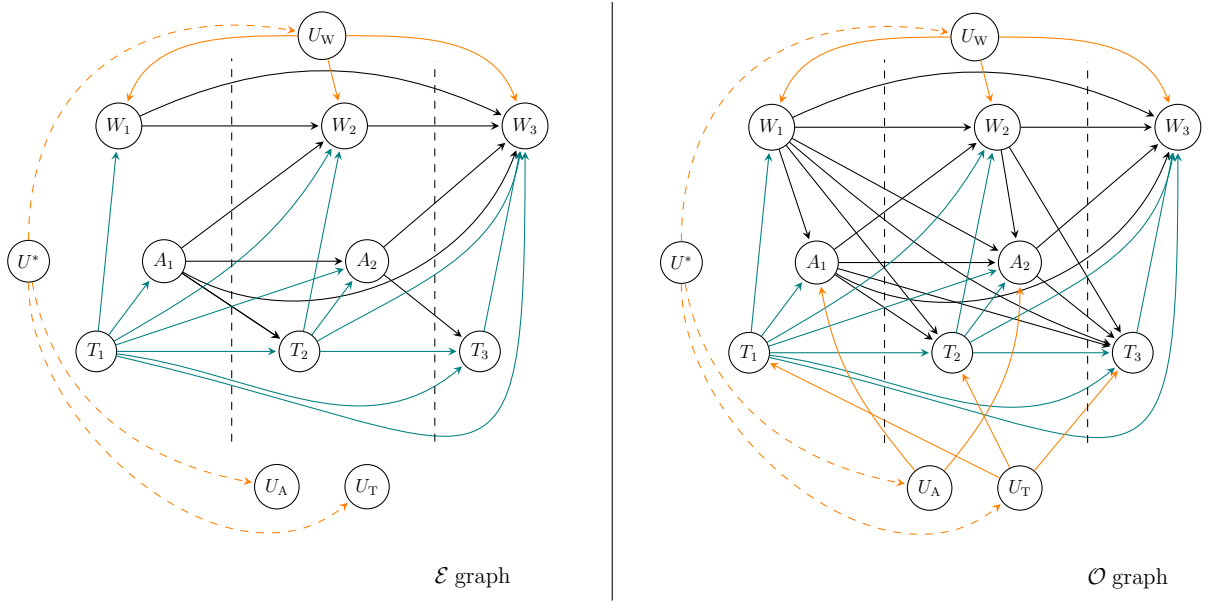

%% file: chapters/2.2-relate-observational-experimental.tex
\subsection{Relating the Target Trial and the Observational Setting}
\label{subsection:relate-e-o}

While data under $\calE$ are unavailable, we can still infer on estimands defined under $\calE$ using $\calO$ data through several causal assumptions outlined in~\ref{project-1-assm:1}--\ref{project-1-assm:4}.
\begin{enumerate}[label=\textbf{A1}]
    \item \label{project-1-assm:1} \textbf{Stability of random effects}: the prior distribution of random effects $\Uw$ given $\phiu$ is the same under $\calE$ and $\calO$:
    \begin{align*}
        f_\calE (\uw \mid \phiu) = f_\calO(\uw \mid \phiu) \,,
    \end{align*}
    for any $\phiu$.
\end{enumerate}
\begin{enumerate}[label=\textbf{A2}]
    \item \label{project-1-assm:2} \textbf{Extended stability of observables}:
    \begin{align*}
        f_\calE(w_{j} \mid \overline{t}_{j}, \overline{a}_{j-1}, \overline{w}_{j-1}, \uw, \phiw) = f_\calO(w_{j} \mid \overline{t}_{j}, \overline{a}_{j-1}, \overline{w}_{j-1}, \uw, \phiw) \, ,
    \end{align*}
    for $j = 1, \dots, K+1$, for all $i \in \mathbb{N}$, for any $\phiw$.
\end{enumerate}
\begin{enumerate}[label=\textbf{A3}]
    \item \label{project-1-assm:3} \textbf{Impact of random effects on longitudinal observations}: Under $\calO$, distributions of observables only depend on their corresponding random effects and parameters:
    \begin{align*}
        T_{j} \indep_{\calO} \Uw, \Ua &\mid \overline{A}_{j-1}, \overline{W}_{j-1}, \overline{T}_{j-1}, \Ut, \phit\\
        W_{j} \indep_{\calE, \calO} \Ut, \Ua &\mid \overline{T}_{j}, \overline{A}_{j-1}, \overline{W}_{j-1}, \Uw, \phiw\\
        A_{j} \indep_{\calO} \Ut, \Uw &\mid \overline{W}_{j-1}, \overline{T}_{j-1}, \overline{A}_{j-1}, \Ua, \phia \, ,
    \end{align*}
    for $j = 1, \dots, K$. Recall that, from Section~\ref{subsection:target-trial-definition}, the distributions of $A_j$ and $T_j$'s are defined as desired under $\calE$. These assumptions complement the definition of the target trial from the previous subsection.
\end{enumerate}
\begin{enumerate}[label=\textbf{A4}]
    \item \label{project-1-assm:4}
    \textbf{Positivity}: conditional distributions under $\calE$ of $A_{j}$ and $T_{j}$, given all past variables that could be observed, are absolutely continuous with respect to their $\calO$ analogues:
    \begin{align*}
        f_\calE(a_j \mid \overline{t}_{j}, \overline{a}_{j-1}, \phia^*) > 0 &\Rightarrow f_\calO \big(a_j \mid \overline{w}_{j}, \overline{t}_{j}, \overline{a}_{j-1}, \ua, \phia) > 0 \,, \hspace{1.05cm} \text{for all $j = 1, \dots, K$}\, ;\\
        f_\calE(t_j \mid \overline{a}_{j-1}, \overline{t}_{j-1}, \phit^*) > 0 &\Rightarrow f_\calO \big(t_j \mid \overline{a}_{j-1}, \overline{w}_{j-1}, \overline{t}_{j-1}, \uit, \phit) > 0 \,, \quad \text{for all $j = 1, \dots, K + 1$}\,,
    \end{align*}
    for all possible values of $\ut, \ua, \phit, \phia, \phit^*, \phia^*$ and for all $i \in \mathbb{N}$.
\end{enumerate}
Because stability assumptions~\ref{project-1-assm:1} and~\ref{project-1-assm:2} posit equivalencies in conditional distributions under $\calE$ and $\calO$, $\phiw$ parametrizes the conditional distribution of $W_{ij}$ under both $\calE$ and $\calO$ in~\ref{project-1-assm:3}. Moreover, these stability assumptions imply that the target trial and observational study have the same target population, i.e. same eligibility, inclusion and exclusion criteria. We provide an extended version of G-computation that builds upon \cite{dawid2010identifying} by tailoring their proposed causal assumptions to infer on our rewards as per Figure~\ref{fig:dag-project-1}.
\begin{theorem}[Extended G-Computation]\label{theorem:extended-g-computation}
    The conditional reward $\calR_{k}(h_k, \underline{t}_{k+1}, \underline{d}_k, \phi_\calR)$ can be expressed as an integral with respect to probability measures under $\calO$ under \ref{project-1-assm:1}--\ref{project-1-assm:4}:
    \begin{align*}
        \calR_{k}(h_k, \underline{t}_{k+1}, \underline{d}_k, \phi_\calR) &= \bigintsss_{\uw} \bigintsss_{w_{k+1}} \dots \bigintsss_{w_{K+1}} \sum_{j=k+1}^{K+1} \gamma_j y_j\\
        &\quad \times \frac{\left\{\prod_{j=1}^{K+1} f_\calO(w_j \mid \overline{t}_{j}, \overline{d}_{j-1}, \overline{w}_{j-1}, \uiw, \phiw)\right\} f_\calO(\uw \mid \phiu)}{\bigintss_{\uw} \left\{\prod_{j=1}^{k} f_\calO(w_j \mid \overline{t}_{j}, \overline{d}_{j-1}, \overline{w}_{j-1}, \uw, \phiw) \right\} f_\calO(\uw \mid \phiu) d\uw}\\
        &\quad \times \d w_{K+1} \dots \d w_{j+1} \d \uw \, .
    \end{align*}
\end{theorem}
\begin{proof}
    The proof is available in the Appendix~\ref{appendix-1:g-computation-proof}.
\end{proof}

%% file: chapters/2.3-estimation.tex
\subsection{Bayesian Estimation: Parameters, Rewards and Optimal Decisions}
\label{subsection:estimation}

% Using $\Phi = \big(\phi, u\big) = \big(\phi, \{u_i\}_{i=1}^n\big)$ to denote the set of all unobserved variables which include data-generating parameters and random effects, w
We let $\pi_0(\phi)$ denote the prior on $\phi$ and $\pi_n(\phi)$ to denote the posterior distribution of $\phi$ given the data of $n$ observed individuals under $\calO$. We are primarily interested in $\pi_n(\phi_\calR) = \int_{\phit, \phia} \pi_n(d \phi)$, the marginal posterior distribution of $\phi_\calR$, because they parametrize our rewards. Following \cite{saarela2015predictive}, the proof that the reward evaluated at $\phi_{\calR 0}$ can be expressed as a limiting case of an integral with respect to the posterior $\pi_n(\phi_\calR)$ is available in Appendix~\ref{appendix-1:de-finetti}.

While the characterization of regimes, $\calE, \calO$ processes and causal assumptions have been made entirely non-parametrically so far, model specifications are required to carry out posterior sampling. A pivotal point of this work stems from the requirement of $\pi_n(\phi_\calR)$ needing a joint model for fitting due to the presence of correlated random effects $U = (\Uw, \Ua, \Ut)$ through its prior $f_\calO(u \mid \phiu)$. In other words, no components can be taken out of the integral to solely infer on the posterior $\pi_n(\phi_\calR) \stackrel{\phi_\calR}{\propto} \pi_n(\phi)$ which, in product form, is:
\begin{align*}
    \pi_n(\phi) &\propto \bigintsss_{u_1, \dots, u_n} \Bigg\{\prod_{i=1}^n \Bigg\{ \prod_{j=1}^{m_i} f_\calO\big(w_{ij} \mid \overline{t}_{ij}, \overline{a}_{i(j-1)}, \overline{w}_{i(j-1)}, \uiw, \phiw\big)\\
    &\hspace{2cm} \times f_\calO\big(t_{ij} \mid \overline{a}_{i(j-1)}, 
    \overline{w}_{i(j-1)}, \overline{t}_{i(j-1)}, \uit, \phit\big)\\
    &\hspace{2cm} \times \one_{j > 1} \bigg[f_\calO\big(a_{i(j-1)} \mid \overline{w}_{i(j-1)}, \overline{t}_{i(j-1)}, \overline{a}_{i(j-2)}, \uia, \phia\big)\bigg] \Bigg\}\\
    &\hspace{2cm} \times \left(1 - \bigintsss_{t_{im_i}}^{\zeta_i} f_\calO(s \mid \overline{w}_{im_i}, \overline{t}_{im_i}, \overline{a}_{i(m_i - 1)}, \uit, \phit) \d s\right)f_\calO(u_i \mid \phiu) \Bigg\}\pi_0(\phi) \, .
\end{align*}
The visit process contribution to likelihood formula above assumes an absence of competing risks, such as death, which is presumed to be a sufficiently rare outcome to be disregarded; this assumption is reasonable in the context of our data analysis in Section~\ref{section:data-analysis}.

To estimate the optimal reward given $\phi_{\calR}$, we simulate each treatment arm in a forward fashion. At each decision point $j \geq k$, the treatment arm yielding the maximum reward conditional the new simulated history $h_j$ serves as the optimal reward for the specific $\phi_{\calR}$ value. To estimate the true conditional expectations $\E_\calE\left[Y_{j} \mid \overline{T}_{j} = \overline{t}_j, \overline{A}_{j-1} = \overline{d}_{j-1}, \overline{W}_{k}=\overline{w}_{k}, \phi_\calR \right]$ in rewards, we carry out Monte Carlo integration for each posterior draw $\phi_{\calR}^{(b)}$. With random effects $\uw$ influenced by $h_j$ which is being conditioned on, a Metropolis algorithm is needed to sample from $f_\calE(\uw \mid \overline{w}_k, \overline{t}_k, \overline{d}_{k-1}, \phiu)$, which is the same distribution under $\calO$ by Bayes' Theorem and stability assumptions~\ref{project-1-assm:1}--\ref{project-1-assm:2}. We provide details about the computational procedure in Appendix~\ref{appendix-1:monte-carlo-approximation}. Given sufficiently large amounts of posterior draws via MCMC and Monte Carlo samples, the Monte Carlo error of posterior estimates should be negligible, i.e. converging to 0; see Appendix~\ref{appendix-1:mc-error-reward} for more details.

%% file: chapters/3-simulation.tex
\section{Simulation Study}

In our simulation study, we demonstrate that joint modeling of all model components ensures unbiased reward estimates, enabling more accurate identification and  estimation of optimal treatment regimes.

\subsection{Setup}
\label{subsection:setup}

We posit a data-generating process under $\calO$ to estimate $\calE$ quantities of interest, here rewards and optimal regimes. We begin by generating each individual's unmeasured confounders $U_i = \log(\Uit), \Uiw, \Uia \mid \Sigma_0 \sim \operatorname{MVN}\big(0, \Sigma_0\big)$ where $\Sigma_0$ specifies the degree of correlation between model-specific random effects. We define $\Sigma_0$ here to be a $3 \times 3$ matrix with diagonal entries to be 0.36 and other correlation values to be 0.216. We simulate individual visits $T_{i1}, \dots$ until the censoring time $\zeta_i \sim \operatorname{Unif}(3.5, 4)$ is attained; each simulated individual has the observed data $\big\{(t_{ij}, a_{ij}, w_{ij})\big\}_{i=1}^{m_i} \cup \left\{t_{i(m_i+1)}, w_{i(m_i + 1)}, m_i\right\}$ under $\calO$. We define model-specific histories $H_{ij, \mathrm{A}} = \big[1, Y_{ij}, X_{ij}, A_{i(j-1)}\big]$, $H_{ij, \mathrm{Y}} = \big[1, T_{ij}, X_{ij}, A_{i(j-1)}, A_{i(j-1)} X_{ij}, A_{i(j-1)} T_{ij}\big]$ and $H_{ij, \mathrm{X}} = H_{ij, \mathrm{T}} = \big[X_{i(j-1)}, A_{i(j-1)}\big]$. Under $\calO$, at every time point $T_{ij}$, we simulate a corresponding ``mark'' $(W_{ij}, A_{ij})$ given historic data according to the equations below and the following data-generating parameters: $\phi_{\mathrm{T}0} = (-0.3, 0.5)$, $\lambda_0 = 0.2$, $\alpha_0 = 3.5$, $\phi_{\mathrm{Y}0} = (-0.3, 0.3, 0.3, 0.55, -0.5, -0.5)$, $\phi_{\mathrm{A}0} = (0.2, -0.2, 0.2)$,
$\phi_{\mathrm{X}0} = (0.4, 0.4)$ and $\tau_{\mathrm{X}0}^2 = 0.3$.
\begin{align*}
    A_{ij} \mid \overline{W}_{ij}, \overline{T}_{ij}, \overline{A}_{i(j-1)}, \Uia, \phia &\sim \operatorname{Bern}\left(\Phi^{-1}\left(H_{ij, \mathrm{A}}^\top \phia + \Uia\right)\right)\\
    X_{ij} \mid \overline{T}_{ij}, \overline{A}_{i(j-1)}, \overline{W}_{i(j-1)}, \phix &\sim \mathcal{N}\left(H_{ij, \mathrm{X}}^\top\phix\scalebox{1.2}{,}\, \tau_{\mathrm{X}}^2\right)\\
    Y_{ij} \mid 
    \overline{X}_{ij}, \overline{T}_{ij}, \overline{A}_{i(j-1)}, \Uiy, \phiy &\sim \operatorname{Bern}\left(\Phi^{-1}\left(H_{ij, \mathrm{Y}}^\top\phiy + \Uiy\right)\right)\\
    T_{ij} - T_{i(j-1)} \mid \overline{W}_{i(j-1)}, \overline{A}_{i(j-1)}, \overline{T}_{i(j-1)}, \Uit, \phit &\sim \operatorname{Weibull}\left(\lambda_0 \log(\Uit) \exp \left\{H_{ij, \mathrm{T}}^\top \phit \right\} \scalebox{1.2}{,}\, \alpha\right)
\end{align*}
We provide the data-generating mechanism in Appendix~\ref{appendix-1:simulation-study}. Conditional on the posterior draw $\phi^{(b)}$, a posterior distribution of regime rewards can be obtained via Theorem~\ref{theorem:extended-g-computation}. Using Monte Carlo integration for each $\phi^{(b)}$ as per Appendix~\ref{appendix-1:monte-carlo-approximation}, we can estimate conditional rewards using an arbitrary large Monte Carlo sample size, i.e. $10\,000\,000$ here. Under $\calO$, we simulate ``training" sets of $n=100$ and $n=300$ participants whose data are used for posterior inference of $\phi_\calR$. For each training set (or simulation replication), we carry out the Bayesian estimation procedure under four model specifications: $(Y, A, T)$, $(Y, A)$, $(Y, T)$ and $(Y)$. We refer to the latter three as ``partial models'' as at least one of the treatment and visit processes are not modeled. If included, they are all correctly specified; for instance, the outcome, time-varying covariate and treatment assignment processes are correctly specified in $(Y, A)$ but the visit process is not modeled at all.

Posterior draws of these parameters under each of the four models are used to estimate 1) the scalar-valued $\calR_{k}(h_k, \underline{t}_{k+1}, \underline{d}_k^\opt, \phi_{\calR 0})$ and 2) the individualized treatments at visit times $t_{i2} = 2$ and $t_{i3} = 3$ for a ``test'' set of 100 participants not utilized in the posterior inference procedure for $\phi_\calR$. Here, we define our feasible sets to be $\{0, 1\}$ for all stages under $\calE$. We carry out simulations for two sets of history profiles consisting of one observed visit: for $i = 1$, $\overline{x}_{11} = x_{11} = -0.35$, $\overline{y}_{11} = y_{11} = 1$ and, for $i = 2$, $\overline{x}_{21} = x_{21} = 0.35$, $\overline{y}_{21} = y_{21} = 0$. We refer to these two patient histories as $h_{11} = (x_{11}, y_{11})$ and $h_{21} = (x_{21}, y_{21})$ respectively. For patient 1 with history $h_{11}$, $(d_2^\opt, d_3^\opt) = (1, d_3^\opt)$ with a true optimal reward value of 1.132 and, for patient 2 with history $h_{21}$, $(d_2^\opt, d_3^\opt) = (0, d_3^\opt)$ with a true optimal reward value of 0.952. The third treatment is left ``to be decided'' with respect to the decision function $d_3^\opt$ according to their evolving status after their second visit:
\begin{align*}
    \text{Patient 1:} \quad \max_{a \in \{0, 1\}} \calR_1\big(h_{11}, (2, 3), (a, d_3^\opt), \phi_{\calR 0}\big) &= \calR_1\big(h_{11}, (2, 3), (1, d_3^\opt), \phi_{\calR 0}\big) \approx 1.132 \, ;\\
    \text{Patient 2:} \quad \max_{a \in \{0, 1\}} \calR_1\big(h_{11}, (2, 3), (a, d_3^\opt), \phi_{\calR 0}\big) &= \calR_1\big(h_{21}, (2, 3), (0, d_3^\opt), \phi_{\calR 0}\big) \approx 0.952 \, .
\end{align*}
We also simulate a ``test'' set of $n_{\text{test}} = 100$ individuals from $\calE$ and identify their optimal individualized treatments at times $2$ and $3$. We report the mean agreement rate (AR), the average probability of identifying the correct optimal treatment under all four model specifications across all $n_{\text{test}}$ individuals. Finally, generating $h_{i1}$ under $\calE$ with $t_1 = 1$ and assuming future $t_2 = 2$ and $t_3 = 3$, we define individual-specific estimated benefit, bias and ARs for each test subject $i \in \{n + 1, \dots, n + n_{\text{test}}\}$ as follows:
\begin{align*}
    \operatorname{Benefit}_i &= \calR_1\big(h_{i1}, (2, 3), (1, d_3^\opt), \phi_{\calR 0}\big) - \calR_1\big(h_{i1}, (2, 3), (0, d_3^\opt), \phi_{\calR 0}\big) \\
    \operatorname{Bias}_i &= \frac{1}{S} \sum_{s = 1}^S  \E_{\phi_{\calR} \mid \calF_n^{(s)}}\left[\calR_1\big(h_{i1}, (2, 3), (\widehat{d}_2^\opt, \widehat{d}_3^\opt), \phi_{\calR }\big)\right] - \calR_1\big(h_{i1}, (2, 3), (d_2^\opt, d_3^\opt), \phi_{\calR 0} \big) \\
    \operatorname{AR}_i &= \frac{1}{S} \sum_{s=1}^S \one \left\{ \calR_1\big(h_{i1}, (2, 3), (d_2^\opt, d_3^\opt), \phi_{\calR 0}\big) =  \E_{\phi_{\calR} \mid \calF_n^{(s)}}\left[\calR_1\big(h_{i1}, (2, 3), (\widehat{d}_2^\opt, \widehat{d}_3^\opt), \phi_{\calR }\big)\right]\right\} \, .
\end{align*}
We also elaborate on Monte Carlo error estimates across simulation replications in  Appendix~\ref{appendix-1:mc-error-reward-simulation-study}.

\subsection{Simulation Results}

\begin{figure}[!ht]
    \centering
    \begin{minipage}[b]{0.49\textwidth}
        \includegraphics[width=1.1\textwidth]{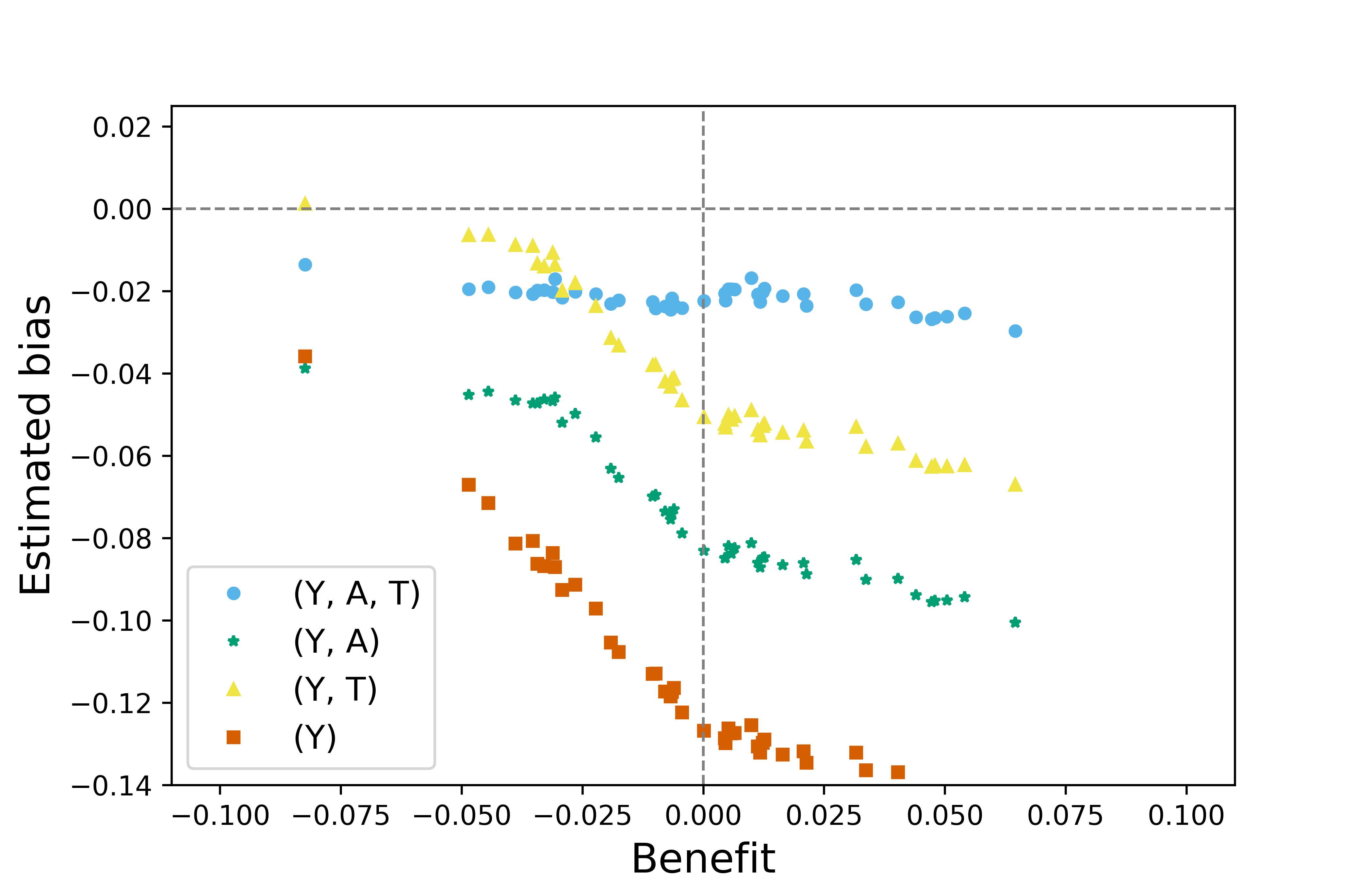}
    \end{minipage}
    \hfill
    \begin{minipage}[b]{0.49\textwidth}
        \includegraphics[width=1.1\textwidth]{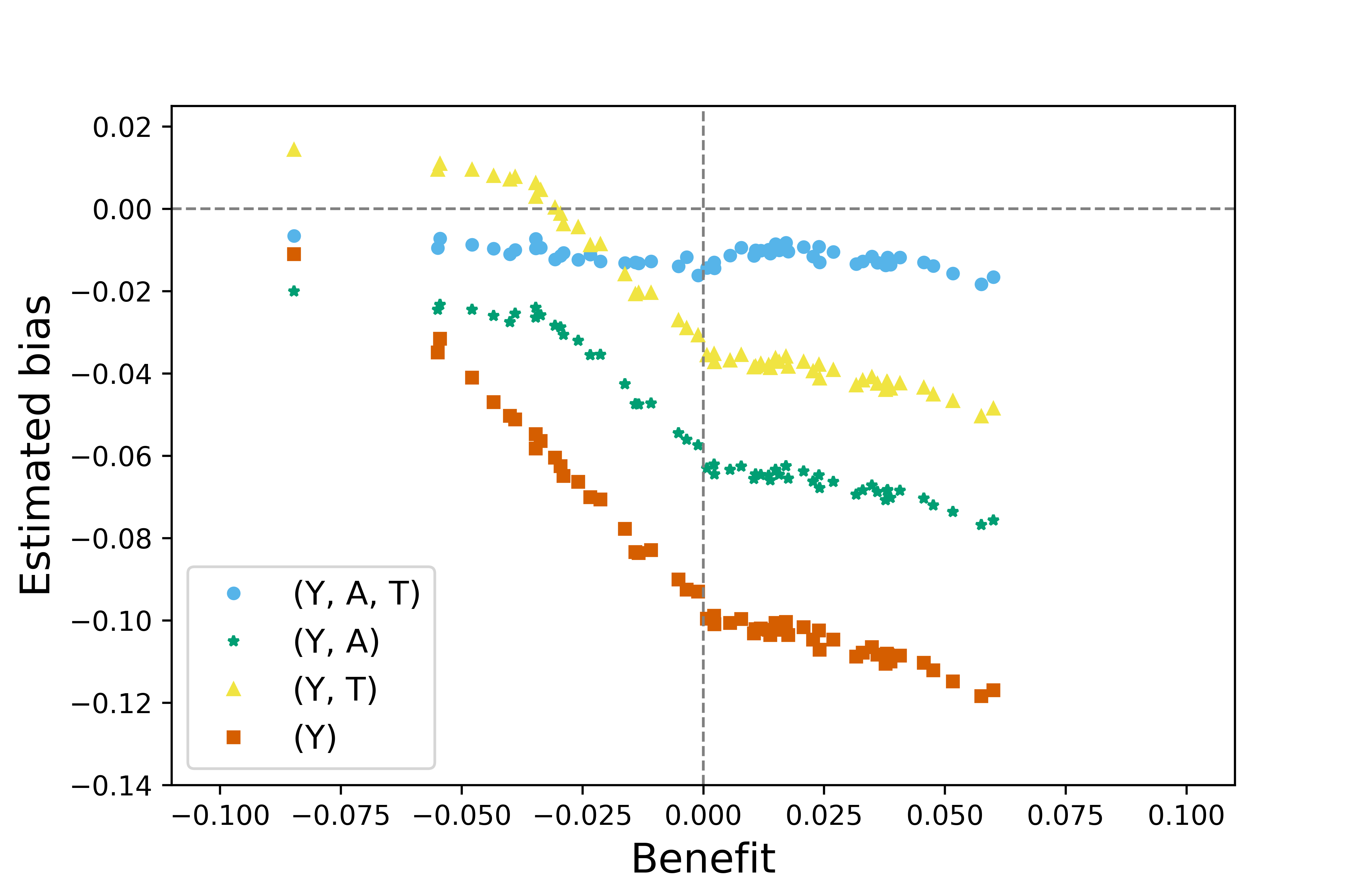}
    \end{minipage}
    \caption{Benefit vs. estimated bias for all four model specifications given previous outcome $y_{i1} = 0$ (left) and $y_{i1} = 1$ (right).}
    \label{fig:benefit-vs-estimated-loss}
\end{figure}
\begin{figure}[!ht]
    \begin{minipage}[b]{0.49\textwidth}
        \includegraphics[width=1.1\textwidth]{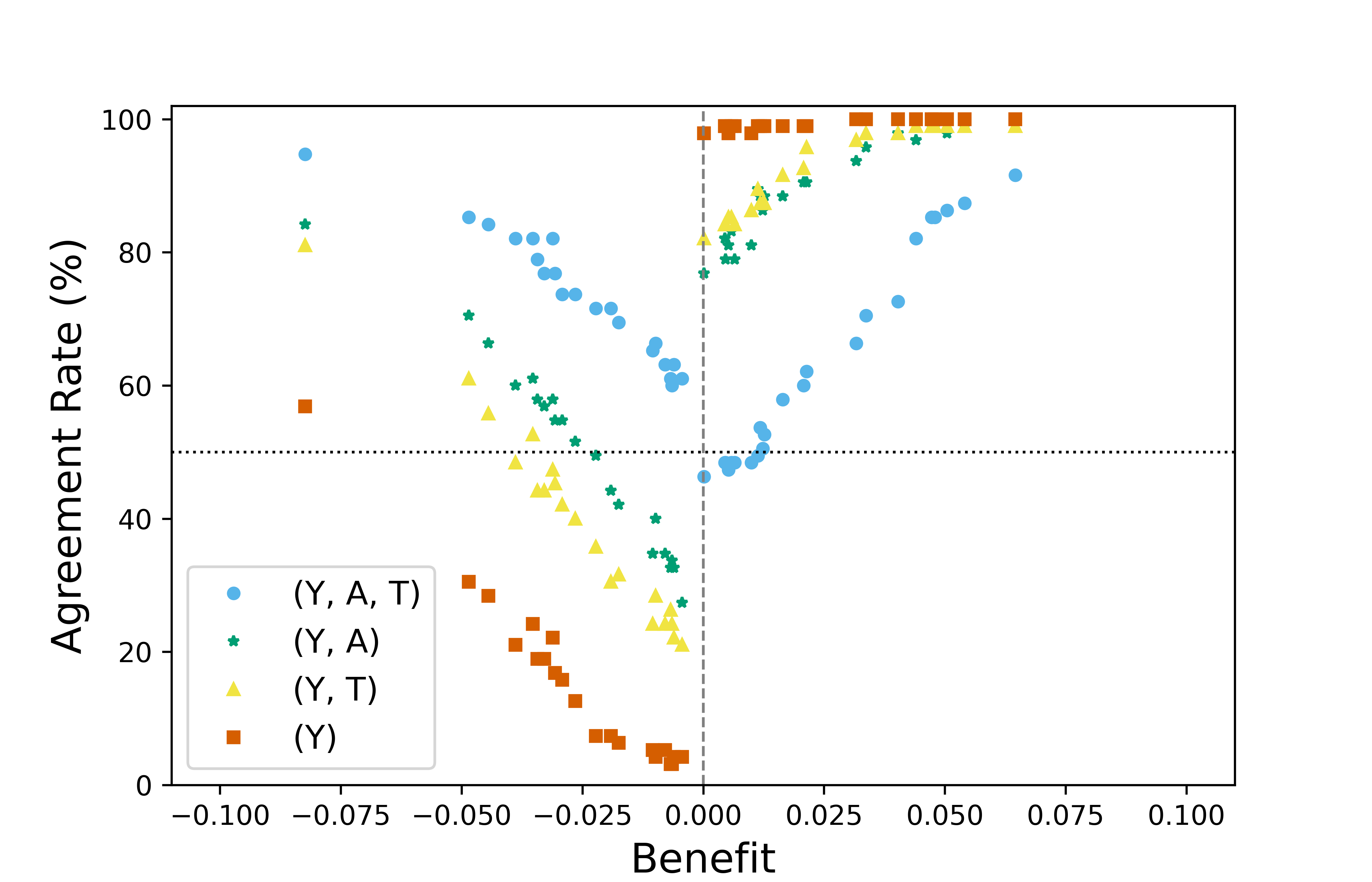}
        % \label{fig:benefit-vs-agreement-rate_y1_0}
    \end{minipage}
    \hfill
    \begin{minipage}[b]{0.49\textwidth}
        \includegraphics[width=1.1\textwidth]{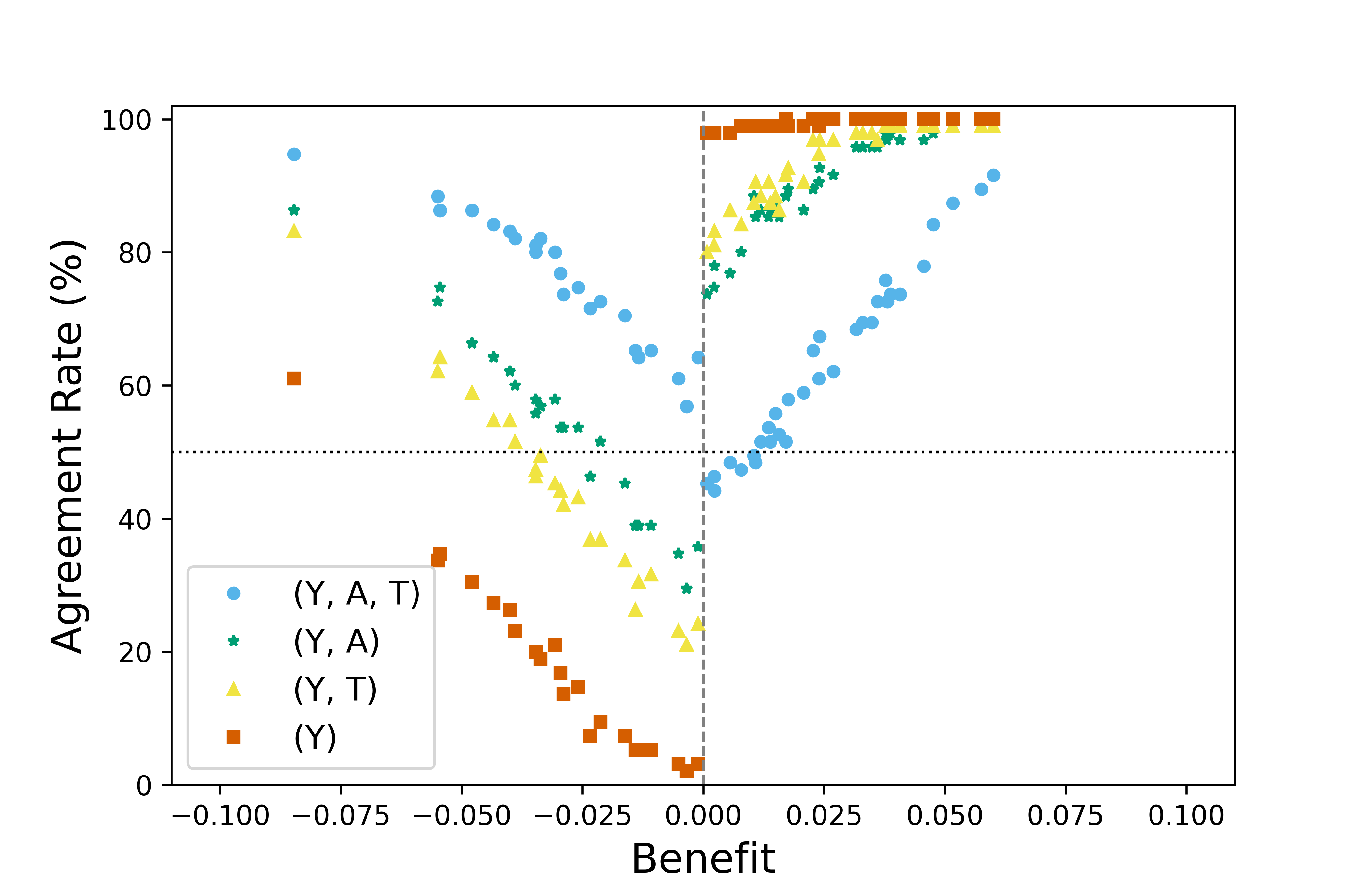}
        % \label{fig:benefit-vs-agreement-rate_y1_1}
    \end{minipage}
    \caption{Benefit vs. agreement rate for all four model specifications given previous outcome $y_{i1} = 0$ (left) and $y_{i1} = 1$ (right).}
    \label{fig:benefit-vs-agreement-rate}
\end{figure}

In Figure~\ref{fig:benefit-vs-estimated-loss}, we observe the true benefit -- reward difference under $(1, d_3^\opt)$ and $(0, d_3^\opt)$ -- against the estimated loss, the difference in the estimated reward under estimated optimal regime $(\widehat{d}_2^\opt, \widehat{d}_3^\opt)$ and true optimal reward. We highlight that the bias estimates using $(Y, A, T)$ where all three processes are modeled are closest to 0 on average. For select individuals whose optimal regime is $(0, d_3^\opt)$, bias estimates from the $(Y, A, T)$ seem comparable to the ones from the $(Y, T)$ model and hover around the 0 mark, representing unbiasedness. The increase in the magnitude of estimated biases for the three partial models indicates in individuals whose optimal regime is $(1, d_3^\opt)$ that, in our simulation setting, treatment effectiveness seems over-estimated, with the $(Y)$ model performing the worse of all models. We provide bias and standard error estimates under optimal regime and fixed regimes in Appendix~\ref{appendix-1:accuracy-opt-reward} and~\ref{appendix-1:accuracy-fixed-reward} respectively.

Similarly, in Figure~\ref{fig:benefit-vs-agreement-rate}, for individuals with a near-trivial benefit hovering around the 0 mark on the x-axis, we expect decision-makers to be ambivalent, i.e. having an AR around the 50\% mark. As benefit increases, we expect ARs to increase and depart from the 50\% mark and slowly converge to 100\%. We observe this phenomenon for the full joint model $(Y, A, T)$, aligning with the finding from Figure~\ref{fig:benefit-vs-estimated-loss} in that rewards estimates from $(Y, A, T)$ are best. We observe that the partial models are more polarized and biased, sometimes outright preferring one treatment regime over the other. Treatment effects seem to be overestimated in $(Y, A)$, $(Y, T)$ and $(Y)$ models. Given that $(1, d_3^\opt)$ is preferable over $(0, d_3^\opt)$ for around 57\% of randomly generated individuals according to our simulation mechanism, this finding aligns well with the higher agreement rate of the partial models in individuals with positive benefit. Figures~\ref{fig:benefit-vs-estimated-loss} and~\ref{fig:benefit-vs-agreement-rate} highlight the importance of accurately estimating reward values under all treatment arms as the regime individualizing is a comparative process. We also report that, if random effects are generated under $(\Uiw, \Uia) \indep \Uit$ or $(\Uiw, \Uit) \indep \Uia$, estimation under partial models $(Y, A)$ and $(Y, T)$ is respectively equivalent to estimation under the full joint model $(Y, A, T)$; these results are also available in Appendix~\ref{appendix-1:simulation-results-under-partial-models}. Lastly, we provide estimates of various standard errors and study the statistical significance of estimated bias in rewards in Appendix~\ref{appendix-1:simulation-results-under-partial-models}.

%% file: chapters/4-data-application.tex
\section{Individualizing IL-7 Injection Cycles}
\label{section:data-analysis}

\subsection{INSPIRE 2 \& 3 Trials}

Human immunodeficiency virus (HIV) is a chronic disease primarily characterized by a depletion of CD4 cells leading to reduced immune system functionality \citep{douek2009emerging}. Highly active anti-retroviral therapy (HAART) replenishes CD4 cell count to a healthy threshold, i.e. above 500 cells/$\mu$L of blood, in HIV-infected individuals, enabling them to exhibit similar life expectancy as the non-infected \citep{douek2009emerging}. However, 15\% to 30\% of HIV-infected individuals fail to increase their CD4 count under HAART; such patients are referred to as low immunological responders \citep{levy2009enhanced}. INSPIRE trials 2 and 3 were randomized trials which aimed to assess the benefits of repeated Interleukin-7 (IL-7) injections in low immunological responders as a means of increasing their CD4 cell count \citep{thiebaut2014quantifying, thiebaut2016repeated}. Data from the INSPIRE trials demonstrated that patients sustained an increase in CD4 cell count when given IL-7 injections \citep{thiebaut2014quantifying}.

\begin{figure}[!b]
    \centering
    \includegraphics[width=0.8\linewidth]{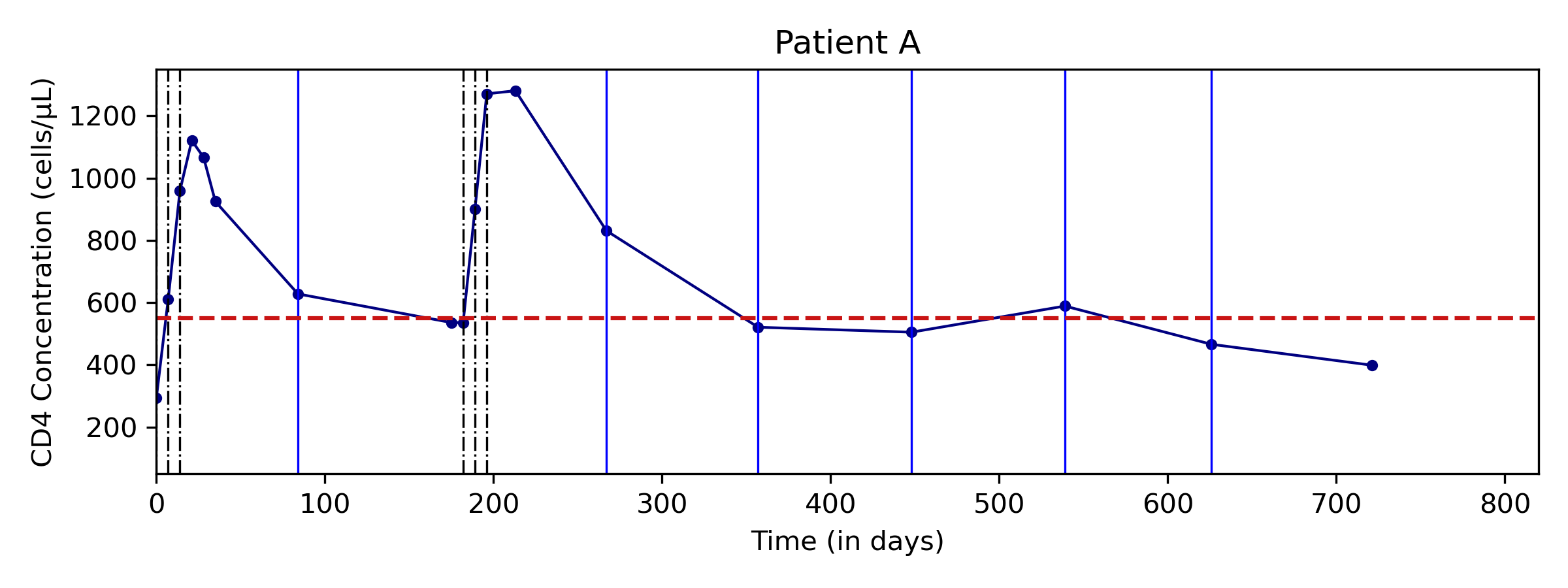}
    \includegraphics[width=0.8\linewidth]{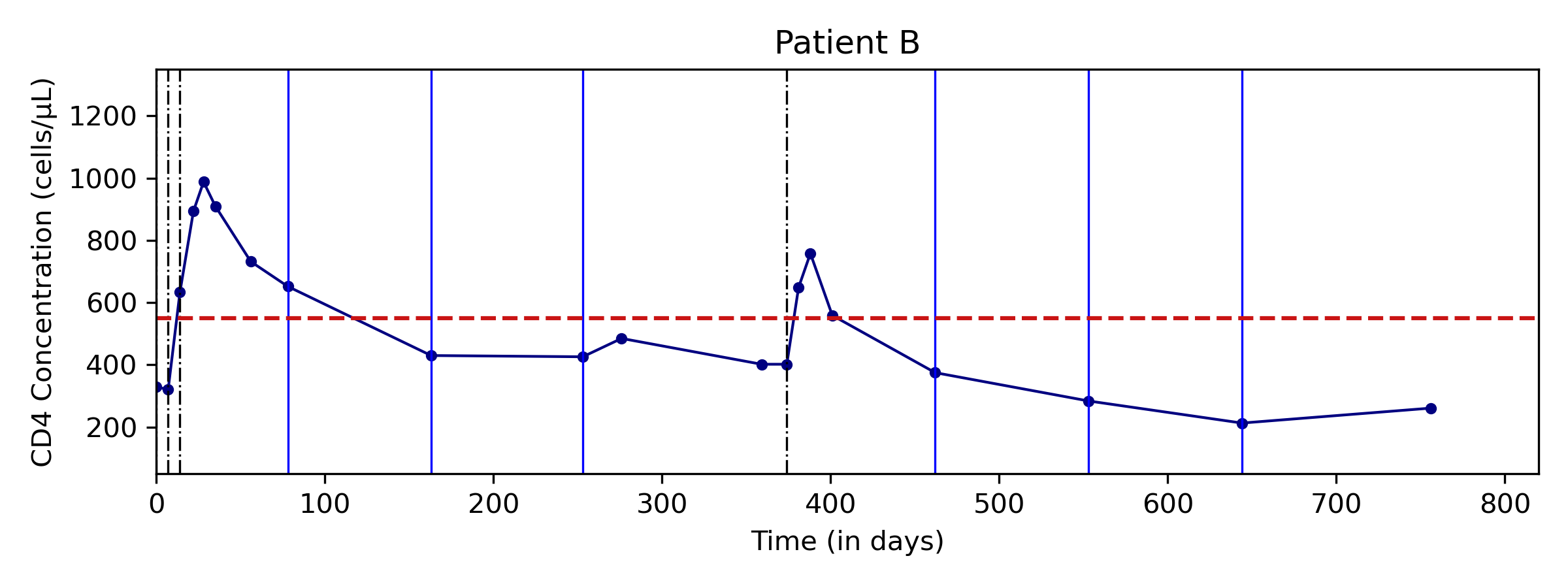}
    \caption{An example of the CD4 trajectory of two patients who has been observed for a total of around two years is displayed in Figure~\ref{fig:inspire-cd4-example}. Each dot represents a CD4 measurement. The vertical dashed-dot lines indicate administered injections and the solid vertical lines delineate approximate 90-windows around which investigators decided if injections were necessary or not.}
    \label{fig:inspire-cd4-example}
\end{figure}
The INSPIRE trials' protocols calls for IL-7 injections at 20$\mu$g/Kg for each injection cycle administered at 90 day intervals if participants' CD4 cell count falls below 550 cells/$\mu$L of blood \citep{thiebaut2016repeated}. However, this time interval varied between patients and within patient visits and participants also did not always receive all prescribed injections in a cycle when eligible.

% For instance, in Figure~\ref{fig:inspire-cd4-example}, two patients' CD4 measurements taken at clinical visits are displayed with variable interarrival times and, in Patient B's chart, they receive one injection instead of three as INSPIRE's protocols in their second cycle.

\subsection{Analysis Plan}

\subsubsection{Definition of Target Trial and Estimands}

With $\operatorname{CD4}_{ij}$ being patient $i$'s CD4 measurement at visit $j \geq 1$, we define outcomes to be $Y_{ij} = \one \left\{\operatorname{CD4}_{ij} \geq 500 \right\}$ and $A_{ij} = 1$ if an injection cycle was initiated at visit $j$, even if it was incomplete, and $0$ otherwise, akin to an intention-to-treat definition. We delineate the study designs of the INSPIRE trials and our target trial of interest as well as the observational setting in Appendix~\ref{appendix-1:inspire-protocol}. The applied problem is the following: given the 111 participant data from INSPIRE 2 and 3, we aim to optimize the injection cycles for the next year according to the target trial's protocol, i.e. for the next four visits. We perform an in-depth analysis for two patients, whom we refer as ``Patient A'' and ``Patient B'', whose trajectories are illustrative of two different patterns of sustained CD4 increase following IL-7 injections; their trajectories are displayed in Figure~\ref{fig:inspire-cd4-example}. %We assume that we see Patients A and B after around one year of follow-up and optimize their treatment strategy for the next year, considering four possible injection time points conditional on data from their first year.

The nature of sustained levels of CD4 concentration, high or low, may not be well-explained by measured variables and could be captured through random effect modeling. Two baseline variables used for treatment tailoring are denoted by $X_{i0}$: BMI and age. We define a patient's eligibility to receive injections at visit $j$ through a feasible set $\Psi_j(h_{j}) = \{0, 1\}$ if the total number of prior injections is two or less, otherwise $\Psi_j(h_j) = \{0\}$.
% \begin{align*}
%     \Psi_{j}(h_{j}) &=
%     \begin{cases}
%         \{0, 1\}& \, \, \, \text{if } \big\{\text{CD4}_{j} \leq 550\big\} \land \Big\{\sum_{\ell=1}^j a_{\ell} \leq 3 \Big\}\\
%         \{0\}& \, \, \, \text{otherwise}
%     \end{cases}\, .
% \end{align*}
% Two conditions must be satisfied: 1) their CD4 cell count must be below 550 cells/$\mu$L of blood and 2) they received fewer than 3 injection cycles in total.

Finally, employing the conditional reward defined in Section~\ref{subsection:target-trial-definition}, we condition on one year's worth of data for Patients A and B, i.e.  $k = 9$ for both individuals. Notably, patient A has received two injection cycles before day 365 whereas Patient B has only received one, since their second cycle starts at day 374. As such, given the feasible set definitions, Patient A would be eligible for a maximum of one more injections cycle whereas Patient B would eligible for up to two cycles. We are interested in optimizing treatment regimes for the next year, considering administering injections at 90-day intervals:
\begin{align*}
    \calR_k(h_k, \underline{t}_{k}, \underline{d}_k^\opt, \phi_{\calR 0}) = \frac{1}{4} \sum_{j=k+1}^{k+4} \E_\calE\big[Y_{j} \mid \overline{T}_{j} = \overline{t}_{j}, \overline{A}_{j-1} = (\overline{a}_{k}, d_k^\opt, \dots, d_{j-1}^\opt), \overline{W}_{k} = \overline{w}_k, \phi_{\calR 0}\big] \, .
\end{align*}
Given that outcomes are binary, $\gamma_j = \frac{1}{4}$ for all $j$ allows the reward to be interpreted as an average probability across future outcomes. As a comparison, we also estimate rewards under a fixed regime ``treat as early and as feasibly often as possible'', which we refer to as the ``always treated'' regime.

\subsubsection{Bayesian Joint modeling of INSPIRE Data}
The immediate rebounding of CD4 concentration after receiving IL-7 injections followed by a decaying effect with time can be seen in Figure~\ref{fig:inspire-cd4-example}; as such, we employ a dose-response (DR) relationship model to depict this phenomenon \citep{davidian2003nonlinear}. As above, we define a variable $Q_{ij}$ to be the indicator if patients have received \textit{any} treatment before visit $j$ and $T_{ij}^{\text{last}}$ to be the time since the last injection if there were any and 0 otherwise.
\begin{align*}
    Q_{ij} =
    \begin{cases}
        0 &\quad \text{if } j = 1\\
        \one \left[\sum_{k=1}^{j-1} A_{ik} > 0\right] &\quad \text{if } j > 1
    \end{cases} \, ;
    \qquad 
    T_{ij}^{\text{last}} =
    \begin{cases}
        0 &\quad \text{if } Q_{ij} = 0\\
        T_{ij} - \max \big\{T_{ik} A_{ik}\big\}_{k=1}^{j-1} &\quad \text{if } Q_{ij} = 1
    \end{cases} \, .
\end{align*}
Moreover, we use an exponential function to mimic the rapid diminishing effects of IL-7 with increasing time without injections in defining  $A_{ij}^{\text{DR}} = Q_{ij} \exp\left(- T_{ij}^{\text{last}} \eta_i\right)$ used in the outcome model, converging to 0 with increasing $T_{ij}^{\text{last}}$, and $\eta_i$ representing the individualized post-injection CD4 decay rate. Our model-specific regressors are: $H_{ij, \mathrm{A}} = \big[1, T_{ij}^{\text{last}}, T_{ij} - T_{i(j-1)}\big]$, $H_{ij, \mathrm{Y}} = \big[1, A_{ij}^{\text{DR}}, A_{ij}^{\text{DR}} * \operatorname{Age}_i, A_{ij}^{\text{DR}} * \operatorname{BMI}_i, Y_{i(j-1)}\big]$ and $H_{ij, \mathrm{T}} = \big[1, T_{ij}^{\text{last}}, \operatorname{Phase}_{ij}, \operatorname{Age}_i, \operatorname{BMI}_i\big]$. In the following equations, we specify three model components: visit process, treatment assignment process and outcome process.
\begin{align*}
    A_{ij} \mid \overline{T}_{ij}, \overline{A}_{i(j-1)}, \overline{Y}_{i(j-1)}, X_{i0}, \eta_i, \Uia, \phia &\sim \operatorname{Bern}\left(\Phi^{-1}\left(H_{ij, \mathrm{A}}^\top \phia + \Uia\right)\right)\\
    Y_{ij} \mid  \overline{T}_{ij}, \overline{A}_{i(j-1)}, \overline{Y}_{i(j-1)}, X_{i0}, \Uiy, \phiy &\sim \operatorname{Bern}\left(\Phi^{-1}\left(H_{ij, \mathrm{Y}}^\top\phiy + \Uiy\right)\right)\\
    T_{ij} - T_{i(j-1)} \mid \overline{A}_{i(j-1)}, \overline{Y}_{i(j-1)}, \overline{T}_{i(j-1)}, \operatorname{Phase}_{ij}, X_{i0}, \Uit, \phi_T, \sigma_{\mathrm{T}}^2 &\sim \mathcal{N}\left(H_{ij, \mathrm{T}}^\top \phi_{T} + \Uit, \sigma_{\mathrm{T}}^2\right)
\end{align*}
The interarrival times $T_{ij} - T_{i(j-1)}$ were modeled using a Normal distribution, incorporating the protocol phase of the patient as a categorical variable. Given the ITT model for modeling the uptake of injection method, we employ a binomial GLM model with probit link. We exclude patient information from visits where they are receiving their second or third injection in a cycle, defining the first post-treatment observation as the initial ``injection-free'' visit following an IL-7 cycle. We consider correlated random intercepts $\Uit, \Uiw, \Uia \mid \Sigma \sim \operatorname{MVN}\big(0, \Sigma\big)$ in the full joint model and similar conditional independence approaches in partial models $(Y, A)$, $(Y, T)$ and $(Y)$ as in the simulation study. We model the decay rate $\eta_i \mid \mu, \tau^2 \sim \mathcal{N}(\mu, \tau^2)$ and posit the following uninformative priors for our model parameters: $\phit \sim \operatorname{MVN}\left(0, 25 \cdot \mathbf{I}_{\operatorname{dim}(\phi_{\mathrm{T}})}\right)$, $\sigma_{\mathrm{T}}^2 \sim \operatorname{InvGamma}(0.1, 0.1)$, $\phia \sim \mathcal{N}(0, 25 \cdot \mathbf{I}_{\operatorname{dim}(\phia)})$, $\phiy \sim \mathcal{N}(0, 25 \cdot \mathbf{I}_{\operatorname{dim}(\phiy)})$, $\mu \sim \mathcal{N}(0, 25)$, $\tau^2 \sim \operatorname{InvGamma}(0.1, 0.1)$ and $\Sigma \sim \operatorname{InverseWishart}\big(\mathbf{I}_{3}, 3\big)$.

Four MCMC chains were run, each generating 100\,000 samples, discarding the first 50\,000 burn-in samples; traceplots for reward parameters $\phiy$, $\phiu$, $\mu$ and $\tau^2$ are available in Appendix~\ref{appendix-1:mcmc-traceplots}. After applying a thinning interval of 10, G-computation was run on the 10\,000 retained posterior samples to estimate the optimal and ``always treated`` reward.

\subsection{Data Analysis Results}

\begin{figure}[!ht]
    \centering   
    \includegraphics[width=\linewidth]{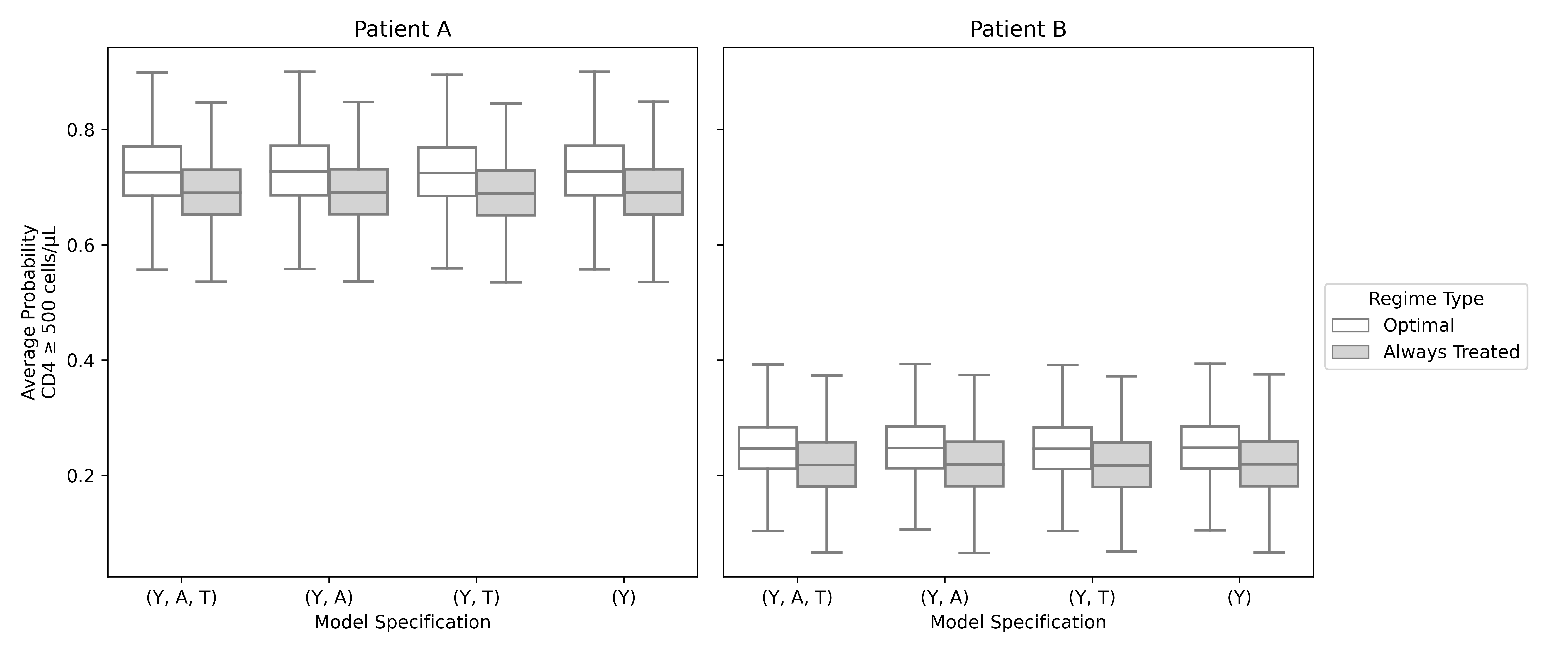}
    \caption{Estimated optimal reward under estimated optimal regime and estimated reward under an ``always treated'' regime for patients A and B under all four joint model specifications.}
    \label{fig:inspire-boxplots}
\end{figure}
In Figure~\ref{fig:inspire-boxplots}, we illustrate the posterior predictive rewards as boxplots if patient were seen for four more visits at 90 day intervals with treatment assigned under estimated optimal regime (in white) and under a fixed regime (in gray). From Figure~\ref{fig:inspire-boxplots}, we can see that all four model specifications yield similar reward estimates under both regimes, indicating that little subject-specific heterogeneity underpinned the observational treatment and visit processes. As expected, the reward under an optimal rewards yields a larger value than under the ``always treated'' regime.

Under the full joint model $(Y, A, T)$ and estimated optimal regime, which we assume hereon, Patient A was recommended $\widehat{\underline{d}}_k^\opt = (0, 1, 0, 0)$ at a frequency of 64.9\% followed by $(0, 1, 0, 0)$ 29.0\% of times whereas Patient B was recommended $(1, 1, 0, 0)$ and $(1, 0, 1, 0)$ respectively $54.5\%$ and $33.8\%$ of the time; more details of the frequency of treatment course recommendation is available in Appendix~\ref{appendix-individualized-treatment-frequency.}. This result can be understood as the model recommending to treat as often as permitted, with the preferred timing to depend on patient's history. For instance, from Figure~\ref{fig:inspire-cd4-example}, Patient A's longitudinal CD4 measurements are on the higher side whereas Patient B's CD4 count falls lower than Patient A. This explains why Patient A is more often recommended to receive an injection cycle after the first quarter than immediately whereas Patient B's most recommended course of treatment to receive injections as soon as possible, if eligible. Moreover, given Patient's A history, the average probability of having their CD4 count greater than 500 cells/$\mu$L is higher than Patient B's. This reward estimation which is explainable through patient history stems from random effects $\Uiy$ and $\eta_i$ being sampled from a distribution conditional on observed history. Lastly, injection recommendation is not heavily affected by patient characteristics given the statistically insignificant estimates of tailoring variable coefficients based on 95\% credible intervals; a summary of parameters from the fitted joint model is available in Table~\ref{table:inspire-summary-table} in Appendix~\ref{subsection:inspire-summary-table}. The similarity of results between model specifications can be attributed to the low correlation between model-specific random effects. While our proposed method and framework allows for correlated random effects, in this application the unobserved individual-level heterogeneity did not seem to be a major driver of the results. %While our proposed method and framework allows studying latent factors influencing observed heterogeneity, the discrepancies between observed and scheduled treatments, visits are likely due to data missing at random rather than being driven by historical factors.

% While the data analysis results aligns with the initial biological hypothesis that CD4 count rebounding immediately after IL-7 injections and decaying with time, 

% The data analysis results aligns with the overarching initial belief that, with injections rebounding CD4 count, optimizing a reward consisting CD4 information will yields a regime that suggests injecting as often as the restrictions would enable a medical provider. Interestingly, given the multiple decision points, our optimal regimes for both Patients A and B require that injections meet the maximum number of allotted cycles and are as spaced out as possible. This aligns with the biological hypothesis that CD4 counts declines over time, hence temporal spacing could maximize benefit in low immunological responders.

% A summary of parameters from the Bayesian joint model is available in Table~\ref{table:inspire-summary-table} in Appendix~\ref{subsection:inspire-summary-table}. In particular, we take interest in interaction terms, also referred as ``blip'' parameters in DTR literature \citep{moodie2007demystifying}. While none of the blip parameter estimates are statistically significant at the 5\% level, optimal treatment was not significantly affected by tailoring variable information due to the idea of injection as often as permitted. However, one area where this may have an impact is the timing of treatment in optimal regimes. For instance, some subpopulations seem to respond better to treatment and it may be the case that the timings of injections is less relevant than for other individuals.

%% file: chapters/5-discussion.tex
\section{Discussion and Conclusion}

% Summary of paper contribution + Relation to prior work
In this article, we first build upon the target trial paradigm using the decision theory framework from~\citep{saarela2015predictive} and \citep{yiu2022causal} to both articulate optimal multi-stage DTRs and incorporate visit time information. In the same spirit as the potential outcome framework, assumptions relating $\calE$ and $\calO$ settings handle well the continuous nature of visit times in the definition of the causal estimand, here optimal rewards and treatment regimes \citep{dawid2010identifying}. Secondly, we propose using Bayesian joint modeling with irregular data to account for correlated random effects between different model components in estimating optimal treatment regimes and rewards. In particular, while \cite{hua2022personalized}, \cite{guan2020bayesian} and \cite{oganisian2022bayesian} also explore Bayesian estimation of optimal DTRs with irregular data, our work studies the induced bias by ignoring the correlated structure of random effects and how this bias translates in the estimation procedure of optimal rewards.

%A minor limitation is that ground truth conditional rewards are estimated via Monte Carlo sampling because of the potentially intractable distribution with respect to which integration is taken. However, we argue that 10\,000\,000 Monte Carlo draws are sufficient to assume precision up to five decimals in our simulation study, which is enough to investigate bias of both estimated rewards and estimated optimal DTRs.% In many cases, despite model misspecification of random effects, there may not be bias translating in the estimation of optimal rewards and DTRs. Further research delineating mathematical conditions that requires fitting a complete joint model, a computationally expensive procedure particularly due to MCMC sampling of latent variables.

% Limitations
Our notational framework assumes outcomes, covariates and treatments can only be measured at visits, akin to a marked point process. For instance, observational data collected from different data sources could have treatment assignments prescribed and clinical measurements taken on different days, requiring further refinement of our proposed framework. Moreover, our Bayesian joint model also makes strong parametric assumptions for the distribution of random effects, a restriction that can be relaxed in future work akin to~\cite{orihara2023addressing}. In the definition of our optimal rewards $\calR_{k}(h_k, \underline{t}_{k+1}, \underline{d}_k^\opt, \phi_{\calR 0})$, we assume that future decision times $\underline{t}_{k+1}$ are chosen when defining the target trial of interest. Akin to~\cite{guan2020bayesian} and~\cite{hua2022personalized}, they could be part of the decision-making process and included in the definition of optimal DTRs. In principle, the proposed framework allows for optimization of visit times, but studying this is left as future work. 
% the optimization procedure to produce deterministic decisions optimized to the individual rather
% rather to stem from stochastic distributions

In the data analysis, we highlight several limitations. Akin to~\cite{villain2019adaptive} and \cite{dong2020evaluating}, the number of IL-7 injections can be included in the optimization procedure with a clearer articulation of the target trial. Secondly, in the spirit of retaining a parsimonious outcome model, only effect modifiers for the magnitude of injection benefits were looked into, in that we could also look into individualized effects of treatment decay. We also assume that the visit process random effect does not depend on the protocol phase, an assumption which can be relaxed since visit times in irregular data can vary more as a study progresses. Lastly, injection schedules and their visit times can themselves be optimized, motivating further methodological extension of this work.

% The expensive nature of clinical trials cannot keep up with the increasing complexity of clinical research questions. While observational data becomes more readily available with technological advancements, our work addresses the overarching theme of developing techniques grounded in conceptual [...], statistical theory and computational feasibility.
% The rising complexity of clinical research questions outpaces the feasibility of conducting costly clinical trials. As observational data becomes increasingly accessible through technological advancements, our work contributes to this need of developing methodologies that are not only conceptually grounded but also rooted in robust statistical theory and computational feasibility.

% cherry on top?
% Finally, the rising complexity of clinical research questions outpaces the feasibility of conducting costly clinical trials. As observational data becomes increasingly accessible through technological advancements, our work contributes to this shift by developing methodologies that are conceptually grounded, statistically rigorous, and computationally feasible.
To conclude, longitudinal observational data often feature irregular and informative observation times; we have shown that treating these times as non-informative can hinder identification of the optimal DTR and proposed a joint modeling approach to account for their informative nature. This, coupled with our target trial approach to identifying the causal estimand can be used to improve the rigor of DTR estimation.

%% file: appendix/appendix-C-optimal-decision.tex
\section{Characterizing Optimal Regimes via Optimal Decision Functions}
\label{appendix-1:optimal-decision-functions}

% We primarily take interest in estimating optimal reward under optimal future treatment allocation, $\calR_{k}(h_k, \underline{t}_{k+1}, \underline{d}_{k}^\opt, \phi_\calR)$. 
A recursive definition of the reward $\calR_{k}(h_k, \underline{t}_{k+1}, \underline{d}_{k}^\opt, \phi_\calR)$ can be used for any non-terminal stage $k$:
\begin{align*}
    &\quad \calR_{k}(h_k, \underline{t}_{k+1}, \underline{d}_{k}^\opt, \phi_\calR)\\
    &= \gamma_{k+1} \E_\calE\left[Y_{k+1} \mid \overline{T}_{k+1} = \overline{t}_{k+1}, \overline{A}_{k} = \big(\overline{a}_{k-1}, d_{k}^\opt\big), \overline{W}_{k}=\overline{w}_{k}, \phi_\calR \right] \\
    &\quad + \sum_{j=k+2}^{K+1} \gamma_j \E_\calE\left[Y_{j} \mid \overline{T}_{j} = \overline{t}_j, \overline{A}_{j-1} = \big(\overline{a}_{k-1}, d_{k}^\opt, \dots, d_{j-1}^\opt\big), \overline{W}_{k}=\overline{w}_{k}, \phi_\calR \right]\\
    &= \gamma_{k+1} \E_\calE\left[Y_{k+1} \mid \overline{T}_{k+1} = \overline{t}_{k+1}, \overline{A}_{k} = \big(\overline{a}_{k-1}, d_{k}^\opt\big), \overline{W}_{k}=\overline{w}_{k}, \phi_\calR \right]
\end{align*}
\begin{align*}
    &\qquad + \bigintsss_{w_{k+1}} \left\{\sum_{j=k+2}^{K+1} \gamma_j \E_\calE\left[Y_{j} \mid \overline{T}_{j} = \overline{t}_j, \overline{A}_{j-1} = \big(\overline{a}_{k-1}, d_{k}^\opt, \dots, d_{j-1}^\opt\big), \overline{W}_{k+1}=\overline{w}_{k+1}, \phi_\calR \right]\right\}\\
    &\hspace{2cm} \times f_\calE\big(w_{k+1} \mid \overline{t}_{k+1}, \big(\overline{a}_{k-1}, d^\opt_k(h_k)\big), \overline{w}_{k}, \phi_{\calR 0}\big) \, d w_{k+1}\\
    &= \gamma_{k+1} \E_\calE\left[Y_{k+1} \mid \overline{T}_{k+1} = \overline{t}_{k+1}, \overline{A}_{k} = \big(\overline{a}_{k-1}, d_{k}^\opt\big), \overline{W}_{k}=\overline{w}_{k}, \phi_\calR \right]\\
    &\qquad +
    \bigintsss_{w_{k+1}} \calR_{k+1}(h_{k+1}, \underline{t}_{k+2}, \underline{d}^\opt_{k+1}, \phi_{\calR}) \, f_\calE\big(w_{k+1} \mid \overline{t}_{k+1}, \big(\overline{a}_{k-1}, d^\opt_k(h_k)\big), \overline{w}_{k}, \phi_{\calR 0}\big) \, d w_{k+1} \,.
\end{align*}
As such, $d^\opt_j$ for $j = k, \dots, K$ can be characterized recursively with $d_K^\opt$ being the base case:
\begin{align*}
    d_{K}^\opt(h_K, \phi_{\calR}) &= \argmax_{d_K} \,\gamma_{K+1} \E_\calE\left[Y_{k+1} \mid \overline{T}_{k+1} = \overline{t}_{K+1}, \overline{A}_{K} = \big(\overline{a}_{K-1}, d_{K}\big), \overline{W}_{K}=\overline{w}_{K}, \phi_\calR \right]\\
    &= \argmax_{a_K \in \Psi_K(h_K)} \, \gamma_{K+1} \E_\calE\left[Y_{k+1} \mid \overline{T}_{k+1} = \overline{t}_{K+1}, \overline{A}_{K} = \overline{a}_{K}, \overline{W}_{K}=\overline{w}_{K}, \phi_\calR \right]\\
    d_j^\opt(h_j, \phi_{\calR}) &= \argmax_{d_j} \, \Bigg\{ \gamma_{j+1} \E_\calE\left[Y_{j+1} \mid \overline{T}_{j+1} = \overline{t}_{j+1}, \overline{A}_{j} = \big(\overline{a}_{j-1}, d_{j}\big), \overline{W}_{j}=\overline{w}_{j}, \phi_\calR \right]\\
    &\qquad +
    \bigintsss_{w_{j+1}} \calR_{j+1}(h_{j+1}, \underline{t}_{j+2}, \underline{d}^\opt_{j+1}, \phi_{\calR}) \, f_\calE\big(w_{j+1} \mid \overline{t}_{j+1}, \big(\overline{a}_{j-1}, d_j(h_j)\big), \overline{w}_{j}, \phi_{\calR 0}\big) \, d w_{j+1} \Bigg\}\\
    &= \argmax_{a_j \in \Psi_j(h_j)} \, \Bigg\{ \gamma_{j+1} \E_\calE\left[Y_{j+1} \mid \overline{T}_{j+1} = \overline{t}_{j+1}, \overline{A}_{j} = \overline{a}_{j}, \overline{W}_{j}=\overline{w}_{j}, \phi_\calR \right]\\
    &\qquad +
    \bigintsss_{w_{j+1}} \calR_{j+1}(h_{j+1}, \underline{t}_{j+2}, \underline{d}^\opt_{j+1}, \phi_{\calR}) \, f_\calE(w_{j+1} \mid \overline{t}_{j+1}, \overline{a}_j, \overline{w}_{j}, \phi_{\calR 0}) \, d w_{j+1} \Bigg\} \, .
\end{align*}

%% file: appendix/appendix-A-reward-derivation.tex
\section{Derivation of Rewards}
\label{appendix-1:definition-rewards}

First, we begin with factorizing $f_\calE(\overline{w}_{j+1}, \overline{t}_{j+1}, \overline{a}_{j} \mid u, \phiw, \phit^*, \phia^*)$ for any $j = 1, \dots, J$:
\begin{align*}
    &\quad f_\calE(\overline{w}_{j+1}, \overline{t}_{j+1}, \overline{a}_{j} \mid u, \phiw, \phit^*, \phia^*)\\
    &= f_\calE(w_{j+1}, t_{j+1} \mid \overline{w}_{j}, \overline{t}_{j}, \overline{a}_{j}, \uw, \ut, \phiw, \phit^*) \left\{\prod_{\ell=1}^j f_\calE(a_\ell, w_{\ell}, t_{\ell} \mid \overline{a}_{\ell-1}, \overline{w}_{\ell-1}, \overline{t}_{\ell-1}, u, \phiw, \phit^*, \phia^*) \right\} \\
    &= f_\calE(w_{j+1} \mid \overline{t}_{j+1}, \overline{w}_j, \overline{a}_j, \uw, \phiw) \, f_\calE(t_{j+1} \mid \overline{w}_j, \overline{t}_j, \overline{a}_j, \ut, \phit^*) \\
    &\quad \times \left\{\prod_{\ell=1}^j f_\calE(a_\ell \mid  \overline{w}_{\ell}, \overline{t}_{\ell}, \overline{a}_{\ell-1}, \ua, \phia^*) f_\calE(w_\ell \mid  \overline{t}_{\ell}, \overline{a}_{\ell-1}, \overline{w}_{\ell-1}, \uw, \phiw) f_\calE(t_\ell \mid \overline{a}_{\ell-1}, \overline{w}_{\ell-1}, \overline{t}_{\ell-1}, \ut, \phit^*) \right\}\\
    &= \left\{\prod_{\ell=1}^{j+1} f_\calE(w_\ell \mid  \overline{t}_{\ell}, \overline{a}_{\ell-1}, \overline{w}_{\ell-1}, \uw, \phiw) \right\}
    \left\{\prod_{\ell=1}^{j+1} f_\calE(t_\ell \mid  \overline{a}_{\ell-1}, \overline{w}_{\ell-1}, \overline{t}_{\ell-1}, \ut, \phit^*) \right\}\\ 
    &\quad \times\left\{\prod_{\ell=1}^{j} f_\calE(a_\ell \mid  \overline{w}_{\ell}, \overline{t}_{\ell}, \overline{a}_{\ell-1}, \ua, \phia^*) \right\}\\
    &= \left\{\prod_{\ell=1}^{j+1} f_\calE(w_\ell \mid  \overline{t}_{\ell}, \overline{a}_{\ell-1}, \overline{w}_{\ell-1}, \uw, \phiw) \right\}
    \left\{\prod_{\ell=1}^{j+1} f_\calE(t_\ell \mid \overline{a}_{\ell-1}, \overline{t}_{\ell-1}, \ut, \phit^*) \right\} \\ 
    &\quad \times\left\{\prod_{\ell=1}^{j} f_\calE(a_\ell \mid  \overline{t}_{\ell-1}, \overline{a}_{\ell-1}, \ua, \phia^*) \right\}\\
    &= \left\{\prod_{\ell=1}^{j+1} f_\calE(w_\ell \mid  \overline{t}_{\ell}, \overline{a}_{\ell-1}, \overline{w}_{\ell-1}, \uw, \phiw) \right\} f_\calE (\overline{t}_{j+1}, \overline{a}_j \mid \ua, \ut, \phia^*, \phit^*) \\
    &= \left\{\prod_{\ell=1}^{j+1} f_\calE(w_\ell \mid  \overline{t}_{\ell}, \overline{a}_{\ell-1}, \overline{w}_{\ell-1}, \uw, \phiw) \right\} f_\calE (\overline{t}_{j+1}, \overline{a}_j \mid \phia^*, \phit^*) \, .
\end{align*}
The third to last equation is obtained from using the specifications of treatment and observation processes under $\calE$ as per~\ref{project-1-assm:3}. Using the derivation above, we can show that $f_\calE(\overline{w}_j \mid \overline{t}_j, \overline{a}_{j-1}, \uw, \phiw)$ factorizes into a product of conditional distributions:
\begin{align*}
    &\qquad f_\calE(\overline{w}_j \mid \overline{t}_j, \overline{a}_{j-1}, \uw, \phiw, \phia^*, \phit^*)\\
    &= \frac{f_\calE(\overline{w}_j, \overline{t}_j, \overline{a}_{j-1} \mid u, \phiw, \phia^*, \phit^*)}{f_\calE(\overline{t}_j, \overline{a}_{j-1} \mid u, \phia^*, \phit^*)}\\
    &= \frac{\left\{\prod_{\ell=1}^{j} f_\calE(w_\ell \mid  \overline{t}_{\ell}, \overline{a}_{\ell-1}, \overline{w}_{\ell-1}, \uw, \phiw) \right\} f_\calE (\overline{t}_{j}, \overline{a}_{j-1} \mid \phia^*, \phit^*)}{f_\calE(\overline{t}_{j}, \overline{a}_{j-1} \mid \phia^*, \phit^*)} && \text{using the previous derivation}\\
    &= \frac{\left\{\prod_{\ell=1}^{j} f_\calE(w_\ell \mid  \overline{t}_{\ell-1}, \overline{a}_{\ell-1}, \overline{w}_{\ell-1}, \uw, \phiw) \right\} \cancel{f_\calE (\overline{t}_{j}, \overline{a}_{j-1} \mid \phia^*, \phit^*)}}{\cancel{f_\calE(\overline{t}_{j}, \overline{a}_{j-1} \mid \phia^*, \phit^*)}}\\
    &= \prod_{\ell=1}^{j} f_\calE(w_\ell \mid  \overline{t}_{\ell-1}, \overline{a}_{\ell-1}, \overline{w}_{\ell-1}, \uw, \phiw)
\end{align*}
The density from the conditional reward $f_\calE(w_j \mid \overline{t}_j, \overline{a}_{j-1}, \overline{w}_{k}, \uw, \phiw)$ only conditions on partial observations $\overline{W}_k = \overline{w}_k$. For $j > k+1$, we express it as:
\begin{align*}
    &\qquad f_\calE(w_j \mid \overline{t}_j, \overline{a}_{j-1}, \overline{w}_{k}, \uw, \phiw)\\
    &= \bigintsss_{w_{k+1}} \dots \bigintsss_{w_{j+1}} f_\calE(w_{k+1}, \dots, w_{j+1}, w_j \mid \overline{t}_{j}, \overline{a}_{j-1}, \overline{w}_{j-1}, \uw, \phiw) \d w_{j+1} \dots \d w_{k+1} \d \uw \, .
\end{align*}
The stage $j > k$ expected outcome conditional can then be rewritten:
\begin{align*}
    &\qquad \E_\calE\left[Y_{j} \mid \overline{T}_{j} = \overline{t}_j, \overline{A}_{j-1} = \overline{d}_{j-1}, \overline{W}_{k}=\overline{w}_{k}, \phi_\calR \right]\\
    &= \bigintsss_{\uw} \bigintsss_{w_j} y_j \, f_\calE(w_j \mid \overline{t}_j, \overline{a}_{j-1}, \overline{w}_{k}, \uw, \phiw) \, f_\calE(\uw \mid \overline{t}_j, \overline{a}_{j-1}, \overline{w}_{j-1} \phiu) \, \d w_j \d \uw\\
    &= \bigintsss_{\uw} \bigintsss_{w_{k+1}} \dots \bigintsss_{w_j} y_k \, \prod_{\ell=k+1}^{j} f_\calE(w_\ell \mid  \overline{t}_{\ell-1}, \overline{a}_{\ell-1}, \overline{w}_{\ell-1}, \uw, \phiw) \, f_\calE(\uw \mid \overline{t}_j, \overline{a}_{j-1}, \overline{w}_{j-1} \phiu) \d w_j \dots \d w_{k+1} \d \uw \, .
\end{align*}
Using this, we can then rewrite our reward $\calR_{k}(h_k, \underline{t}_{k+1}, \underline{d}_{k}, \phi_\calR)$ using conditional distributions of the two following forms: $f_\calE(w_j \mid \overline{t}_j, \overline{a}_{j-1}, \overline{w}_{j-1}, \uw, \phiw)$ for some $j \leq K + 1$ and $f_\calE(\uw \mid \phiu)$. Grouping several alike terms and using Bayes' formula from the third to fourth step yields the following derivation:
\begin{align*}
    &\qquad \calR_{k}(h_k, \underline{t}_{k+1}, \underline{d}_k, \phi_\calR)\\
    &= \sum_{j=k+1}^{K+1} \gamma_j \E_\calE\left[Y_{j} \mid \overline{T}_{j} = \overline{t}_j, \overline{A}_{j-1} = \overline{d}_{j-1}, \overline{W}_k=\overline{w}_{k}, \phi_\calR \right]\\
    &= \bigintsss_{\uw} \left\{ \sum_{j=k+1}^{K+1} \gamma_j \E_\calE\left[Y_{j} \mid \overline{T}_{j} = \overline{t}_j, \overline{A}_{j-1} = \overline{d}_{j-1}, \overline{W}_k=\overline{w}_{k}, \uw, \phiw \right] \right\} f_\calE(\uw \mid \overline{t}_{k}, \overline{w}_k, \overline{d}_{k-1}, \phi_\calR) \d \uw \\
    &= \bigintsss_{\uw} \Bigg\{ \bigintsss_{w_{k+1}} \dots \bigintsss_{w_{K+1}} \sum_{j=k+1}^{K+1} \gamma_j y_j \prod_{j=k+1}^{K+1} f_\calE(w_j \mid \overline{t}_j, \overline{d}_{j-1}, \overline{w}_{j-1}, \uw, \phiw) \, \d w_{K+1} \dots \d w_{j+1}\Bigg\}\\
    &\qquad \times f_\calE(\uw \mid \overline{w}_j, \overline{t}_{j}, \overline{d}_{j-1}, \phi_\calR) \d \uw\\
    &= \bigintsss_{\uw} \Bigg\{ \bigintsss_{w_{k+1}} \dots \bigintsss_{w_{K+1}} \sum_{j=k+1}^{K+1} \gamma_j y_j \prod_{j=k+1}^{K+1} f_\calE(w_j \mid \overline{t}_j, \overline{d}_{j-1}, \overline{w}_{j-1}, \uw, \phiw) \, \d w_{K+1} \dots \d w_{k+1}\Bigg\}\\
    &\qquad \times \frac{\left\{\prod_{j=1}^{k} f_\calE(w_j \mid \overline{t}_{j}, \overline{d}_{j-1}, \overline{w}_{j-1}, \uw, \phiw)\right\} f_\calE(\uw \mid \phiu)}{\bigintss_{\uw} \left\{\prod_{j=1}^{k} f_\calE(w_k \mid \overline{t}_{k}, \overline{d}_{k-1}, \overline{w}_{k-1}, \uw, \phiw) \right\} f_\calE(\uw \mid \phiu) \d \uw} \d \uw \\
    &= \bigintsss_{\uw} \bigintsss_{w_{k+1}} \dots \bigintsss_{w_{K+1}} \left\{\sum_{j=k+1}^{K+1} \gamma_j y_j \right\}\\
    &\qquad \times \frac{\left\{\prod_{j=1}^{K+1} f_\calE(w_j \mid \overline{t}_{j}, \overline{d}_{j-1}, \overline{w}_{j-1}, \uw, \phiw)\right\} f_\calE(\uw \mid \phiu)}{\bigintss_{\uw} \left\{\prod_{j=1}^{k} f_\calE(w_j \mid \overline{t}_{j}, \overline{d}_{j-1}, \overline{w}_{j-1}, \uw, \phiw) \right\} f_\calE(\uw \mid \phiu) \d \uw} \d w_{K+1} \dots \d w_{k+1} \d \uw \,,
\end{align*}
where the last step in the derivation above combines products of $f_\calE(w_j \mid \overline{t}_{j}, \overline{d}_{j-1}, \overline{w}_{j-1}, \uw, \phiw)$ for indices $j = 1, \dots, k$ and $j = k+1, \dots, K+1$.

%% file: appendix/appendix-B-g-computation.tex
\section{Extended G-Computation Proof}
\label{appendix-1:g-computation-proof}

\begin{proof}
    We make use of assumptions~\ref{project-1-assm:1},~\ref{project-1-assm:2} and~\ref{project-1-assm:4} to convert $\calE-$distributions into $\calO-$distributions. We start with the ratio of $\calE$ and $\calO$ random effects $\uw$ to be equal to 1:
    \begin{align*}
        \frac{f_\calO(\uw \mid \phiu)}{f_\calE(\uw \mid \phiu)} = 1 \, ,
    \end{align*}
    for all possible $\uw$ due to~\ref{project-1-assm:3}. We also consider the following importance sampling ratio:
    \begin{align*}
        \frac{f_\calO(w_k \mid \overline{t}_{k}, \overline{d}_{k-1}, \overline{w}_{k-1}, \uw, \phiw)}{f_\calE(w_k \mid \overline{t}_{k}, \overline{d}_{k-1}, \overline{w}_{k-1}, \uw, \phiw)} = 1
    \end{align*}
    to be well-defined, i.e. have a non-zero denominator, due to~\ref{project-1-assm:4} and to be equal to $1$ for any $k = 1, \dots, K + 1$ due to~\ref{project-1-assm:2}. With these stability assumptions, we directly convert $\calE$ to $\calO$ densities through direct substitions
    \begin{align*}
        &\qquad \calR_k(h_k, \underline{t}_{k+1}, \underline{d}_{k}, \phi_\calR) \\
        &= \bigintsss_{\uw} \bigintsss_{w_{k+1}} \dots \bigintsss_{w_{K+1}}  \sum_{j=k+1}^{K+1} \gamma_j y_j \\
        &\quad \times \frac{\left\{\prod_{j=k+1}^{K+1} f_\calE(w_j \mid \overline{t}_{j}, \overline{d}_{j-1}, \overline{w}_{j-1}, \uw, \phiw)\right\} f_\calE(\uw \mid \phiu)}{\bigintss_{\uw} \left\{\prod_{j=1}^{k} f_\calE(w_j \mid \overline{t}_{j}, \overline{d}_{j-1}, \overline{w}_{j-1}, \uw, \phiw) \right\} f_\calE(\uw \mid \phiu) d \uw} \, d w_{K+1} \dots d w_{k+1} d \uw \\
        &= \bigintsss_{\uw} \bigintsss_{w_{k+1}} \dots \bigintsss_{w_{K+1}}  \sum_{j=k+1}^{K+1} \gamma_j y_j \\
        &\quad \times \frac{\left\{\prod_{j=k+1}^{K+1} f_\calO(w_j \mid \overline{t}_{j}, \overline{d}_{j-1}, \overline{w}_{j-1}, \uw, \phiw)\right\} f_\calO(\uw \mid \phiu)}{\bigintss_{\uw} \left\{\prod_{j=1}^{k} f_\calO(w_j \mid \overline{t}_{j}, \overline{d}_{j-1}, \overline{w}_{j-1}, \uw, \phiw) \right\} f_\calO(\uw \mid \phiu) d \uw} \, d w_{K+1} \dots d w_{k+1} d \uw  \, .
\end{align*}
The positivity assumption~\ref{project-1-assm:4} is implicitly used when evoking $f_\calE(w_j \mid \overline{t}_j, \overline{a}_{j-1}, \overline{w}_{j-1}, \uw, \phiw)$ where the conditional set must have non-zero density. Given that rewards are defined under $\calE$, we require that the measure of the conditional set $(\overline{t}_j, \overline{a}_{j-1}, \overline{w}_{j-1})$ under $\calE$ be absolutely continuous respect to its $\calO$ analogue, i.e.:
\begin{align*}
    f_\calE(\overline{t}_j, \overline{a}_{j-1}, \overline{w}_{j-1}) > 0 &\Rightarrow f_\calO(\overline{t}_k, \overline{a}_{j-1}, \overline{w}_{j-1}) > 0\\
    \Leftrightarrow f_\calE(\overline{t}_j, \overline{a}_{j-1}, \overline{w}_{j-1} \mid u, \phiw, \phia^*, \phit^*) > 0 &\Rightarrow f_\calO(\overline{t}_j, \overline{a}_{j-1}, \overline{w}_{j-1} \mid u, \phi) > 0 \qquad (*)
\end{align*}
for all $j=1, \dots, K+1$ and all possible $u, \phi, \phia^*, \phit^*$ values. Akin to the factorization in~\ref{appendix-1:definition-rewards}, the densities can be expressed as:
\begin{align*}
    f_\calE(\overline{t}_j, \overline{a}_{j-1}, \overline{w}_{j-1} \mid u, \phiw, \phia^*, \phit^*) &= \left\{\prod_{\ell=1}^{j} f_\calE(t_\ell \mid \overline{a}_{\ell-1}, \overline{t}_{\ell-1}, \ut, \phit^*) \right\}\left\{\prod_{\ell=1}^{j-1} f_\calE(a_\ell \mid  \overline{t}_{\ell-1}, \overline{a}_{\ell-1}, \ua, \phia^*) \right\}\\ 
    &\quad \times\ \left\{\prod_{\ell=1}^{j-1} f_\calE(w_\ell \mid  \overline{t}_{\ell}, \overline{a}_{\ell-1}, \overline{w}_{\ell-1}, \uw, \phiw) \right\}\\
    f_\calO(\overline{t}_j, \overline{a}_{j-1}, \overline{w}_{j-1} \mid u, \phi) &= \left\{\prod_{\ell=1}^{j} f_\calO(t_\ell \mid \overline{a}_{\ell-1}, \overline{w}_{\ell-1}, \overline{t}_{\ell-1}, \ut, \phit) \right\}\\
    &\quad \times \left\{\prod_{\ell=1}^{j-1} f_\calO(w_\ell \mid  \overline{t}_{\ell}, \overline{a}_{\ell-1}, \overline{w}_{\ell-1}, \uw, \phiw) \right\}\\ 
    &\quad \times\left\{\prod_{\ell=1}^{j-1} f_\calO(a_\ell \mid  \overline{t}_{\ell-1}, \overline{a}_{\ell-1}, \overline{w}_{\ell-1}, \ua, \phia) \right\} \, .
\end{align*}
% A strong induction argument can be applied such that $f_\calE(w_j \mid  \overline{t}_{j}, \overline{a}_{j-1}, \overline{w}_{j-1}, \uw, \phiw)$ is well-defined for arbitrary $j > 1$. By~\ref{project-1-assm:4}, the $\calO$ analogue is also well-defined. The inductive step then becomes to show that:
% \begin{align*}
%     f_\calE(\overline{t}_{j+1}, \overline{a}_{j}, \overline{w}_{j} \mid u, \phiw, \phia^*, \phit^*) > 0 &\Rightarrow f_\calO(\overline{t}_{j+1}, \overline{a}_{j}, \overline{w}_{j} \mid u, \phi) > 0 \,,
% \end{align*}
The factorization of joint densities can be incorporated into the condition $(*)$ above. Under~\ref{project-1-assm:2}, it simplifies to $f_\calE(t_{j+1} \mid \overline{a}_{j}, \overline{t}_j, \phit^*) > 0 \Rightarrow f_\calO(t_{j+1} \mid \overline{a}_{j}, \overline{t}_j, \overline{w}_j, \ut, \phit) > 0$ and $f_\calE(a_{j} \mid \overline{t}_{j}, \overline{a}_{j-1}, \phia^*) > 0 \Rightarrow f_\calO(a_{j} \mid \overline{w}_{j}, \overline{t}_j, \overline{a}_j, \ua, \phia) > 0$ for any $j$, which is the positivity assumption posited in~\ref{project-1-assm:4}.
\end{proof}

%% file: appendix/appendix-D-posterior-rewards.tex
\section{Posterior Estimation of Rewards}

\subsection{Justification of Bayesian Estimator via de Finetti's Representation Theorem}
\label{appendix-1:de-finetti}

Following \cite{saarela2023role}, we show that the rewards for a yet-to-be-observed individual $i > n$ evaluated at the true $\phi_{\mathrm{W}0}$ value is the limiting case of an expression that would serve as a natural Bayesian estimator using de Finetti's Representation Theorem.
\begin{align*}
    &\quad \lim_{n \rightarrow \infty} \E_\calE\left[Y_{j} \mid \overline{T}_{j} = \overline{t}_{j}, \overline{A}_{j-1} = \overline{d}_{j-1}, \overline{W}_{j-1} = \overline{w}_{j-1}, \calF_n\right]\\
    &= \lim_{n \rightarrow \infty} \bigintsss_{\phi_\calR} \bigintsss_{\uw} \E_\calE\left[Y_{j} \mid \overline{T}_{j} = \overline{t}_{j}, \overline{A}_{j-1} = \overline{d}_{j-1}, \overline{W}_{j-1} = \overline{w}_{j-1}, \uw, \phiw\right] f_\calO(\uw \mid \phiu) \,\pi_n(d\phi_\calR)\\
    &= \bigintsss_{\phi_{\calR}} \bigintsss_{\uw} \E_\calE\left[Y_{j} \mid \overline{T}_{j} = \overline{t}_{j}, \overline{A}_{j-1} = \overline{d}_{j-1}, \overline{W}_{j-1} = \overline{w}_{j-1}, \uw\right] f_\calO(\uw \mid \phiu) \, \delta_{\phi_{\calR0}}(d \phi_{\calR})\\
    &= \E_\calE \left[Y_{j} \mid \overline{T}_{j} = \overline{t}_j, \overline{A}_{j-1} = \overline{d}_{j-1}, \overline{W}_{j-1} = \overline{w}_{j-1}, \phi_{\calR0} \right] \, ,
\end{align*}
which is the $j$th stage intermediate reward under visit times $\overline{t}_j$, treatments $\overline{a}_{j-1}$ and the \textit{true} parameter values $\phi_{\calR0}$, assuming that the posterior converges to a degenerate distribution at $\phi_{\calR0}$. The last equality uses $Y_{j} \indep \calF_n \mid \phi$ for $1 \leq j \leq K$, $i \notin \{1, \dots, n\}$ from the representation theorem. With finite $n$, the above can be approximated by the posterior predictive expectation: %The latter also implies the existence of a (possibly infinite-dimensional) parameter $\phi$ such that, in the case of finite $n$
\begin{align*}
    &\qquad \E_\calE\left[Y_{j} \mid \overline{T}_{j} = \overline{t}_{j}, \overline{A}_{j-1} = \overline{d}_{j-1}, \overline{W}_{j-1} = \overline{w}_{j-1}, \phi_{\calR0}\right]\\
    &\approx \E_\calE\left[Y_{j} \mid \overline{T}_{j} = \overline{t}_{j}, \overline{A}_{j-1} = \overline{a}_{j-1}, \overline{W}_{j-1} = \overline{w}_{j-1}, \calF_n\right]\\
    &= \bigintsss_{\phi_\calR} \E_\calE\left[Y_{j} \mid \overline{T}_{j} = \overline{t}_{j}, \overline{A}_{j-1} = \overline{a}_{j-1}, \overline{W}_{j-1} = \overline{w}_{j-1}, \phi_\calR\right]  \, \pi_n(d \phi_\calR)\\
    &\approx \frac{1}{B} \sum_{b=1}^B \E_\calE\left[Y_{j} \mid \overline{T}_{j} = \overline{t}_{j}, \overline{A}_{j-1} = \overline{a}_{j-1}, \overline{W}_{j-1} = \overline{w}_{j-1}, \phi_\calR^{(b)}\right] \, ,
\end{align*}
where $\big\{\phi_\calR^{(b)}\big\}_{b=1}^B$ are drawn from $\pi_n(\phi_\calR)$ following an MCMC procedure.

\subsection{Monte Carlo Approximation to Optimal and Non-Optimal Rewards}
\label{appendix-1:monte-carlo-approximation}

We also introduce a Monte Carlo estimation procedure to estimate reward values for a fixed value of $\phi_\calR$. Rewards involve nested integration which can be estimated via repeated sampling from $f_\calO(\cdot)$ densities. Recall that, from Appendix~\ref{appendix-1:definition-rewards}, that the conditional reward can be written as:
\begin{align*}
    &\qquad \calR_{k}(h_k, \underline{t}_{k+1}, \underline{d}_k, \phi_\calR)\\
    &= \bigintsss_{\uw} \Bigg\{ \bigintsss_{w_{k+1}} \dots \bigintsss_{w_{K+1}} \sum_{j=k+1}^{K+1} \gamma_k y_k \prod_{j=k+1}^{K+1} f_\calE(w_j \mid \overline{t}_j, \overline{d}_{j-1}, \overline{w}_{j-1}, \uw, \phiw) \, \d w_{K+1} \dots \d w_{k+1}\Bigg\}\\
    &\qquad \times f_\calE(\uw \mid \overline{w}_k, \overline{t}_{k}, \overline{d}_{k-1}, \phi_\calR) \d \uw \, .
\end{align*}
To sample random effects conditional on observed history $\uw \sim f_\calO(\uw \mid \overline{t}_j, \overline{w}_j, \overline{a}_{j-1}, \phiu)$, we resort to a Metropolis argument. This can be done because we are able to evaluate the ratio of its density, the latter which can be expanded via Bayes' Theorem:
\begin{align*}
    f_\calE(\uw \mid \overline{w}_k, \overline{t}_{k}, \overline{d}_{k-1}, \phi_\calR) = \frac{\left\{\prod_{j=1}^{k} f_\calE(w_j \mid \overline{t}_{j}, \overline{d}_{j-1}, \overline{w}_{j-1}, \uw, \phiw)\right\} f_\calE(\uw \mid \phiu)}{\bigintss_{\uw} \left\{\prod_{j=1}^{k} f_\calE(w_j \mid \overline{t}_{j}, \overline{d}_{j-1}, \overline{w}_{j-1}, \uw, \phiw) \right\} f_\calE(\uw \mid \phiu) \d \uw} \, .
\end{align*}
The key is that the normalizing constant on the denominator does not depend on $\uw$ as it is being marginalized out; otherwise, the denominator would also be difficult to evaluate. As a result, the ratio of densities in the Metropolis algorithm, say under random effect values of $\uw'$ and $\uw^{(r-1)}$, can be computed as a ratio of densities:
\begin{align*}
    &\quad \frac{f_\calE(\uw^{'} \mid \overline{w}_k, \overline{t}_{k}, \overline{d}_{k-1}, \phi_\calR)}{f_\calE(\uw^{(r - 1)} \mid \overline{w}_k, \overline{t}_{k}, \overline{d}_{k-1}, \phi_\calR)}\\
    &= \frac{\frac{\left\{\prod_{j=1}^{k} f_\calE(w_j \mid \overline{t}_{j}, \overline{d}_{j-1}, \overline{w}_{j-1}, \uw^{'}, \phiw)\right\} f_\calE(\uw^{'} \mid \phiu)}{\bigintsss_{\uw} \left\{\prod_{j=1}^{k} f_\calE(w_j \mid \overline{t}_{j}, \overline{d}_{j-1}, \overline{w}_{j-1}, \uw, \phiw) \right\} f_\calE(\uw \mid \phiu) \d \uw}}{\frac{\left\{\prod_{j=1}^{k} f_\calE(w_j \mid \overline{t}_{j}, \overline{d}_{j-1}, \overline{w}_{j-1}, \uw^{(r-1)}, \phiw)\right\} f_\calE(\uw^{(r-1)} \mid \phiu)}{\bigintsss_{\uw} \left\{\prod_{j=1}^{k} f_\calE(w_j \mid \overline{t}_{j}, \overline{d}_{j-1}, \overline{w}_{j-1}, \uw, \phiw) \right\} f_\calE(\uw \mid \phiu) \d \uw}}\\
    &= \frac{\left\{\prod_{j=1}^{k} f_\calE(w_j \mid \overline{t}_{j}, \overline{d}_{j-1}, \overline{w}_{j-1}, \uw^{'}, \phiw)\right\} f_\calE(\uw^{'} \mid \phiu)}{\left\{\prod_{j=1}^{k} f_\calE(w_j \mid \overline{t}_{j}, \overline{d}_{j-1}, \overline{w}_{j-1}, \uw^{(r-1)}, \phiw)\right\} f_\calE(\uw^{(r-1)} \mid \phiu)} \, .
\end{align*}
We use a proposal distribution of $f_\calO(\uw \mid \phiu)$, a normal distribution from which we can easily sample. With this, we can estimate the reward using a Monte Carlo sample of size $R$ through the following steps:
\begin{enumerate}
    \item For iteration $1 \leq r \leq R$, sample a proposal $\uw^{'}$ from the prior $f_\calE(\uw \mid \phiu)$.
    \item Sample a dummy $\mathsf{U} \sim \operatorname{Unif}(0, 1)$, assign $\uw^{(r)}$ via:
    \begin{align*}
        \uw^{(r)} =
        \begin{cases}
            \uw^{'} & \text{if } \mathsf{U} \leq \min\left(\frac{f_\calE(\uw^{'} \mid \overline{w}_k, \overline{t}_{k}, \overline{d}_{k-1}, \phi_\calR)}{f_\calE(\uw^{(r - 1)} \mid \overline{w}_k, \overline{t}_{k}, \overline{d}_{k-1}, \phi_\calR)}, 1\right) \,,\\
            \uw^{(r - 1)}&\text{otherwise.}
        \end{cases}
    \end{align*}
    \item Sample $w_{j+1}^{(r)} \sim f_\calO(w_{j+1} \mid \overline{t}_{j+1}, \overline{w}_j, \overline{a}_{j}, \uw^{(r)}, \phiw)$ starting from $j = k$.
    \item Considering all valid $a_{j+1} \in \Psi_{j+1}(h_{j+1})$ possible arms where $h_{j+1} = \left(h_j, t_{j+1}, w_{j+1}^{(r)}\right)$ and $w_{j+1}^{(r)}$ obtained from the previous step, sample $w_{j+2}^{(r)}$.
    \item When estimating fixed regimes, repeat steps 2 and 3 for stages $j+2, \dots, K+1$ and for all $r = 1, \dots, R$ since future treatments are fixed. However, when estimating the optimal regime, $a_k^\opt$ can be estimated using:
    \begin{align*}
        \widehat{a}_j^\opt = \argmax_{a_j \in \Psi_j(h_j)} \, \left\{\frac{1}{R}\sum_{r=1}^R w_{j+1}^{(r)} + \calR_{k+1}\left(h_{j+1}, \underline{t}_{j+2}, \underline{d}_{j+1}^\opt, \phi_\calR\right) \right\} \, ,
    \end{align*}
    with the newly updated $h_{j+1} = \left(h_j, t_{j+1}, \frac{1}{R}\sum_{r=1}^R w_{j+1}^{(r)}, a_j\right)$ depending on the $a_j$ being considered.
\end{enumerate}
% Using $(b, r)$ as a superscript for the $r$th sample obtained using posterior draw $\phi_\calR^{(b)}$, the reward can be approximated as:
% \begin{align*}
%     \calR_k(h_k, \underline{t}_{k+1}, \underline{d}_k, \phi_{\calR 0}) \approx \frac{1}{B} \sum_{b=1}^B \calR_k(h_k, \underline{t}_{k+1}, \underline{d}_k, \phi_{\calR}^{(b)}) \approx \frac{1}{B \cdot R} \sum_{b=1}^B \sum_{r = 1}^{R} \sum_{j=k+1}^{K+1} w_{j}^{(b, r)} \, .
% \end{align*}
The estimation of $\widehat{a}_j^\opt$ depends on $\calR_{k+1}(h_{j+1}, \underline{t}_{j+2}, \underline{d}_{j+1}^\opt, \phi_\calR)$ which itself can be estimated using the steps above and also depends on optimal future rewards as well. As such, a recursive argument is evoked implying a backwards induction procedure where each feasible treatment arm must be considered in order to take max/argmax.

We note that, for numerical stability, the logarithm can be used in step 2:
\begin{align*}
    \mathsf{U} & \leq \min\left(\frac{f_\calE(\uw^{'} \mid \overline{w}_k, \overline{t}_{k}, \overline{d}_{k-1}, \phi_\calR)}{f_\calE(\uw^{(r - 1)} \mid \overline{w}_k, \overline{t}_{k}, \overline{d}_{k-1}, \phi_\calR)}, 1\right)\\
    \log(\mathsf{U}) & \leq \log\left(\min\left(\frac{f_\calE(\uw^{'} \mid \overline{w}_k, \overline{t}_{k}, \overline{d}_{k-1}, \phi_\calR)}{f_\calE(\uw^{(r - 1)} \mid \overline{w}_k, \overline{t}_{k}, \overline{d}_{k-1}, \phi_\calR)}, 1\right)\right)\\
    &= \min\left(\log\left(\frac{f_\calE(\uw^{'} \mid \overline{w}_k, \overline{t}_{k}, \overline{d}_{k-1}, \phi_\calR)}{f_\calE(\uw^{(r - 1)} \mid \overline{w}_k, \overline{t}_{k}, \overline{d}_{k-1}, \phi_\calR)}\right), \log(1)\right)\\
    &= \min\Bigg(\sum_{j=1}^k \left\{\log\left(f_\calE(w_j \mid \overline{t}_j, \overline{d}_{j-1}, \overline{w}_{j-1}, \uw^{'}, \phiw \right) - \log \left(f_\calE(w_j \mid \overline{t}_j, \overline{d}_{j-1}, \overline{w}_{j-1}, \uw^{(r-1)}, \phiw \right)\right\}\\
    &\hspace{1.7cm} + \log\left(f_\calE(\uw^{'} \mid \phiu\right) - \log\left(f_\calE(\uw^{(r-1)} \mid \phiu\right), 0\Bigg) \, .
\end{align*}
Lastly, the exponential increase in computational costs with respect to the number of decision points in this optimization procedure can be mitigated with Numba due to its ability to optimize repeated computational tasks \citep{lam2015numba}. However, while Numba significantly accelerates computations, it cannot completely eliminate the inherent complexity of the problem.

\subsection{Monte Carlo Standard Error of Reward Estimates}
\label{appendix-1:mc-error-reward}

Let $\theta^{(b, r)}$ denote the $r$th draw from $f_\calO(\theta \mid \phi_\calR^{(b)})$, a distribution parametrized by $\phi_{\calR}^{(b)}$ with the $b$th posterior draw for parameter $\phi_{\calR}$. For our purposes, we are interested in Monte Carlo draws for the conditional reward given posterior draws. To calculate the Monte Carlo standard error of Bayesian estimator $\widehat{\theta} = \frac{1}{B\cdot R}\sum_{b, r} \theta^{(b, r)}$, we articulate two properties that we will make use of:
\begin{itemize}
    \item Rewards with superscript $(b)$ are conditionally independent of one and another:
    \begin{align*}
        \theta^{(b, r_1)} \indep \theta^{(b, r_2)} \mid \phi_\calR^{(b)} \quad \text{for any } 1 \leq r_1 < r_2 \leq R \, .
    \end{align*}
    \item Assuming that MCMC draws $\left\{\phi_\calR^{(b)}\right\}_{b=1}^B$ are themselves independent given sufficiently large $R$ and thinning process, we assume that $\Var(\theta^{(b_1, r)}) = \Var(\theta^{(b_2, r)})$ for any $1 \leq b_1, b_2 \leq B$. Using $\widehat{\theta}$ to denote the point estimate of $\theta$ across $B$ posterior predictive draws and $R$ samples in the Monte Carlo integration, combining with Monte Carlo draws being independent by construction yields:
    \begin{align*}
        \Var\left(\widehat{\theta}\right)
        &= \frac{1}{B^2} \sum_{b=1}^B \Var \left(\frac{1}{R}\sum_{r=1}^R \theta^{(b, r)} \right)\\
        &= \frac{1}{B^2} \sum_{b=1}^B \left\{\E\left[\Var \left(\frac{1}{R} \sum_{r=1}^R \theta^{(b, r)} \mid \phi_\calR^{(b)}\right)\right] + \Var\left[\E\left(\frac{1}{R} \sum_{r=1}^R \theta^{(b, r)} \mid \phi_\calR^{(b)}\right)\right]\right\}\\
        &= \frac{1}{B^2} \sum_{b=1}^B \left\{\E\left[\Var \left(\frac{1}{R} \sum_{r=1}^R \theta^{(b, r)} \mid \phi_\calR^{(b)}\right)\right] + \Var\left[\E\left(\frac{1}{R} \sum_{r=1}^R \theta^{(b, r)} \mid \phi_\calR^{(b)}\right)\right]\right\}\\
        &= \frac{1}{B^2} \sum_{b=1}^B \left\{\frac{1}{R^2} \sum_{r=1}^R \E\left[\Var \left(\theta^{(b, r)} \mid \phi_\calR^{(b)}\right)\right] + \Var\left[\E\left(\frac{1}{R} \sum_{r=1}^R \theta^{(b, r)} \mid \phi_\calR^{(b)}\right)\right]\right\}\\
        &= \frac{\sum_{b=1}^B \E\left[\Var \left(\theta^{(b, 1)} \mid \phi_\calR^{(b)}\right)\right]}{B^2 R} + \frac{\sum_{b=1}^B \Var\left[\E\left(\frac{1}{R} \sum_{r=1}^R \theta^{(b, r)} \mid \phi_\calR^{(b)}\right)\right]}{B^2} \, .
    \end{align*}
    Notice that, with $R \rightarrow \infty$, we are left with the variance due to the MCMC procedure:
    \begin{align*}
        \lim_{R \rightarrow \infty} \Var\left(\widehat{\theta}\right) = \frac{\sum_{b=1}^B \Var\left[\E\left(\frac{1}{R} \sum_{r=1}^R \theta^{(b, r)} \mid \phi_\calR^{(b)}\right)\right]}{B^2} \, .
    \end{align*}
    However, with $B \rightarrow \infty$, we are left with an asymptotically 0 variance.
\end{itemize}

\subsection{Monte Carlo Standard Error of Simulation Study Reward Estimates}
\label{appendix-1:mc-error-reward-simulation-study}

An additional layer to the reward estimate is that multiple replications of the entire simulation are run:
\begin{align*}
    \calR_k(h_k, \underline{t}_{k+1}, \underline{d}_{k+1}, \phi_{\calR 0}) &\approx \frac{1}{S} \sum_{s=1}^{S} \left(\frac{1}{B} \sum_{b=1}^B \calR_k(h_k, \underline{t}_{k+1}, \underline{d}_{k+1}, \phi_{\calR}^{(s, b)}) \right)\\
    &\approx \frac{1}{S} \sum_{s=1}^{S} \left( \frac{1}{B} \sum_{b=1}^B \left( \frac{1}{R} \sum_{r = 1}^{R} \left\{\sum_{j=k+1 }^{K+1} w_{j}^{(s, b, r)}\right\} \right)\right)
\end{align*}
where $\phi_{\calR}^{(s, b)}$ is the $b$th draw from the posterior distribution of $\phi_\calR$ conditional via MCMC on the $s$th generated dataset $\calF_n^{(s)}$. Here, we denote $\theta^{(s, b, r)}$ to be the $r$th sample in the Monte Carlo integration and $b$th MCMC draw for the $s$th simulation, for $s = 1, \dots, S$ with $S$ being the total number of simulation study replications. We also index MCMC draws of the reward parameters $\phi_\calR^{(s, b)}$ with the simulation iteration.
\begin{align*}
    \Var(\widehat{\theta}) &= \frac{1}{S^2} \sum_{s=1}^S \Var\left(\frac{1}{B}\sum_{b=1}^B \frac{1}{R}\sum_{r=1}^R \theta^{(s, b, r)}\right)\\
    &= \frac{1}{S^2} \sum_{s=1}^S \left\{\Var\left[\E\left(\frac{1}{B}\sum_{b=1}^B \frac{1}{R}\sum_{r=1}^R \theta^{(s, b, r)}\right) \mid \calF_n^{(s)}\right] + \E\left[\Var\left(\frac{1}{B}\sum_{b=1}^B \frac{1}{R}\sum_{r=1}^R \theta^{(s, b, r)}\right) \mid \calF_n^{(s)}\right]\right\}\\
    &= \frac{1}{S^2} \sum_{s=1}^S \Bigg\{\Var\left[\E\left(\frac{1}{B}\sum_{b=1}^B \frac{1}{R}\sum_{r=1}^R \theta^{(s, b, r)}\right) \mid \calF_n^{(s)}\right] + \frac{\sum_{b=1}^B \E\left[\Var \left(\theta^{(s, b, 1)} \mid \phi_\calR^{(s, b)}, \calF_n^{(s)}\right)\right]}{B^2 R} \\
    &\hspace{2cm} + \frac{\sum_{b=1}^B \Var\left[\E\left(\frac{1}{R} \sum_{r=1}^R \theta^{(s, b, r)} \mid \phi_\calR^{(s, b)}, \calF_n^{(s)}\right)\right]}{S^2 B^2}\Bigg\}\\
    &= \frac{\sum_{s} \Var\left[\E\left(\frac{1}{B}\sum_{b=1}^B \frac{1}{R}\sum_{r=1}^R \theta^{(s, b, r)}\right) \mid \calF_n^{(s)}\right]}{S^2} + \frac{\sum_{s, b} \Var\left[\E\left(\frac{1}{R} \sum_{r} \theta^{(s, b, r)} \mid \phi_\calR^{(s, b)}, \calF_n^{(s)}\right)\right]}{S^2 B^2}\\
    &\qquad + \frac{\sum_{s, b, r} \E\left[\Var \left(\theta^{(s, b, r)} \mid \phi_\calR^{(s, b)}, \calF_n^{(s)}\right)\right]}{S^2 B^2 R^2}\,.
\end{align*}
From the second step to the third, we make use of the result from Appendix~\ref{appendix-1:mc-error-reward}. Each term characterizes different quantities in the variability of simulations.
\begin{itemize}
    \item $\frac{\sum_{s} \Var\left[\E\left(\frac{1}{B}\sum_{b=1}^B \frac{1}{R}\sum_{r=1}^R \theta^{(s, b, r)}\right) \mid \calF_n^{(s)}\right]}{S^2}$: between-simulation variability;
    \item $\frac{\sum_{s, b} \Var\left[\E\left(\frac{1}{R} \sum_{r} \theta^{(s, b, r)} \mid \phi_\calR^{(s, b)}, \calF_n^{(s)}\right)\right]}{S^2 B^2}$: variability due to MCMC draws, or variability of the distribution over which integration is conducted;
    \item $\frac{\sum_{s, b} \E\left[\Var \left(\theta^{(s, b, 1)} \mid \phi_\calR^{(s, b)}, \calF_n^{(s)}\right)\right]}{S^2 B^2 R}$: variability due to Monte Carlo integration.
\end{itemize}

%% file: appendix/appendix-E-simulation-algorithm.tex
\section{Simulation Study Data Generation}
\label{appendix-1:simulation-study}

\RestyleAlgo{boxruled}
\LinesNumbered
\begin{algorithm}[!ht]
    \KwIn{Parameters $\phi_{\mathrm{A 0}}, \phi_{\mathrm{T 0}}, \phi_{\mathrm{X 0}}, \phi_{\mathrm{Y 0}}, \sigma^2_{\mathrm{X 0}}, \lambda_0, \alpha_0, \Sigma_0$}
    \caption{Data Generation Procedure for Simulation Study}
    \For{i in \{1, \dots, n\}}{
        $\log(\Uit), \Uiw, \Uia \sim \operatorname{MVN}\big(0, \Sigma_0\big)$\\
        $\zeta_{i} \sim \mathcal{U}(3.5, 4)$ \tcc{Independent censoring}
        \tcc{First Visit}
        $j \gets 1$\\
        $T_{i1} \gets 0$\\
        $X_{i1} \gets 0$\\
        $Y_{i1} \sim \operatorname{Bern}\left(\Phi^{-1}\left(H_{ij, \mathrm{Y}}\phi_{\mathrm{Y}0} + \Uiw\right)\right)$\\
        \While{$T_{ij} < \zeta_i$}{
            \tcc{Visit occurs, increment counting process}
            $j \gets j + 1$\\
            \tcc{Defining model-specific regressors}
            $H_{ij, \mathrm{A}} \gets \big[1, Y_{ij}, X_{ij}, A_{i(j-1)}\big]$\\
            $H_{ij, \mathrm{Y}} \gets \big[1, T_{ij}, X_{ij}, A_{i(j-1)}, A_{i(j-1)} X_{ij}, A_{i(j-1)} T_{ij}\big]$\\
            $H_{ij, \mathrm{X}} \gets H_{ij, \mathrm{T}} \gets \big[X_{i(j-1)}, A_{i(j-1)}\big]$\\
            \tcc{Sampling $X_{ij}, Y_{ij}, A_{ij}$}
            $A_{ij} \mid \overline{W}_{ij}, \overline{T}_{ij}, \overline{A}_{i(j-1)}, \Uia, \phia \sim \operatorname{Bern}\left(\Phi^{-1}\left(H_{ij, \mathrm{A}} \phia + \Uia\right)\right)$\\
            $X_{ij} \mid \overline{T}_{ij}, \overline{A}_{i(j-1)}, \overline{W}_{i(j-1)}, \Uix, \phix \sim \mathcal{N}\left(H_{ij, \mathrm{X}}\phix + \Uix, \Sigma_{\mathrm{W}}\right)$\\
            $Y_{ij} \mid \overline{X}_{ij}, \overline{T}_{ij}, \overline{A}_{i(j-1)}, \Uiy, \phiy \sim \operatorname{Bern}\left(\Phi^{-1}\left(H_{ij, \mathrm{Y}} \phiy + \Uiy\right)\right)$\\
            \tcc{Generate candidate $T_{ij}$}
            $T_{ij} - T_{i(j-1)} \mid H_{ij, \mathrm{T}}, \Uit, \phit \sim \operatorname{Weibull}\left(\lambda_0 \log(\Uit) \exp \left\{H_{ij, \mathrm{T}} \phit \right\}, \alpha\right)$\\
        }
    }
    \caption{Data Generation Procedure for Simulation Study}
    \label{alg:project-1}
\end{algorithm}
\FloatBarrier

%% file: appendix/appendix-F-simulation-results.tex
\section{More Simulation Results}
\label{appendix-1:more-simulation-results}

\subsection{Accuracy of Optimal Reward Estimates}
\label{appendix-1:accuracy-opt-reward}

\begin{table}[!b]
    \centering
    \renewcommand{\arraystretch}{1.5}
    \begin{tabular}{|c|c||c|c|c|c||c|c|c|c|}
        \hline
        & & \multicolumn{4}{c||}{Patient 1} & \multicolumn{4}{c|}{Patient 2} \\ \hline
        $n$ & Model & Bias & MC Error & SE MC & Avg. SE & Bias & MC Error & SE MC & Avg. SE \\ \hline
    \multirow{4}{*}{$300$} & $(Y, A, T)$ & 3.867 & 0.053 & 5.301 & 7.003 & 2.068 & 0.058 & 5.757 & 7.035 \\ \cline{2-10}
        & $(Y, A)$ & 8.347 & 0.062 & 6.239 & 7.273 & 6.378 & 0.068 & 6.837 & 7.409 \\ \cline{2-10}
        & $(Y, T)$ & 5.669 & 0.065 & 6.468 & 7.538 & 2.814 & 0.069 & 6.850 & 7.578 \\ \cline{2-10}
        & $(Y)$ & 11.901 & 0.069 & 6.901 & 7.682 & 9.440 & 0.076 & 7.640 & 7.921 \\ \hline
    \multirow{4}{*}{$100$} & $(Y, A, T)$ & 9.564 & 0.093 & 9.329 & 10.864 & 4.490 & 0.110 & 10.990 & 11.273 \\ \cline{2-10}
        & $(Y, A)$ & 12.973 & 0.096 & 9.607 & 11.066 & 8.088 & 0.113 & 11.316 & 11.572 \\ \cline{2-10}
        & $(Y, T)$ & 10.728 & 0.100 & 9.992 & 11.377 & 5.225 & 0.116 & 11.629 & 11.860 \\ \cline{2-10}
        & $(Y)$ & 14.939 & 0.102 & 10.15 & 11.535 & 9.796 & 0.118 & 11.768 & 12.095 \\ \hline
    \end{tabular}
    \caption{Simulation results showcasing bias, Monte Carlo error (MC Error), standard error of the Monte Carlo sampling distribution (SE MC) and average standard error (Avg. SE) in estimating the reward under optimal regime for both patient profiles $h_{11}$, $h_{21}$ as defined above. All quantities have been multiplied by 100 in the table.}
    \label{table:simulation-results-opt}
\end{table}
From Table~\ref{table:simulation-results-opt}, we observe that the bias in estimating the optimal reward value decreases with increasing sample size and $(Y, A, T)$ model yields the lowest bias estimate compared to the three other partial models. We note that, akin to~\cite{saarela2015predictive}, the estimates in the rewards under optimal regime do not appear to be completely unbiased due to the non-linearity induced in the definition of optimal rewards and non-smooth estimators can violate conventional asymptotic convergence properties. However, under \textit{fixed} regimes, e.g. $(d_2, d_3) = (0, 0)$ or $(d_2, d_3) = (1, 1)$, regimes which we refer as ``never treated'' and ``always treated'', we report that reward estimates are unbiased (see Appendix~\ref{appendix-1:more-simulation-results}) and our method performs as we expect.

\begin{table}[!ht]
    \centering  % This command properly centers the figure
    \renewcommand{\arraystretch}{1.5}
    \begin{tabular}{|c|c|c|c|c|c|}
        \hline
        & $|\,$Bias$\,|$ & $\overline{\text{MC Error}}$ & $\overline{\text{SD MC}}$ & $\overline{\text{Avg SE}}$& AR \\ \hline
        $(Y, A, T)$ & 3.049 & 0.517 & 5.035 & 13.406 & 68.4\%\\ \hline
        $(Y, A)$ & 7.300 & 0.522 & 5.087 & 13.010 & 73.2\%\\ \hline
        $(Y, T)$ & 4.183 & 0.532 & 5.181 & 13.243 & 70.3\%\\ \hline
        $(Y)$ & 10.431 & 0.525 & 5.112 & 12.767 & 62.9\%\\ \hline
    \end{tabular}
    \caption{Simulation results showcasing absolute bias, Monte Carlo error, standard deviation of the Monte Carlo sampling distribution (MC Error), average standard error (Avg SE) and agreement rate (AR), all averaged across 100 randomly generated patient profiles under $\calE$. All quantities except AR have been multiplied by 100 in the table.}
    \label{table:simulation-results-2}
\end{table}

\clearpage

\subsection{Accuracy of Reward Estimates Under Fixed Regime}
\label{appendix-1:accuracy-fixed-reward}

\begin{table}[!ht]
    \centering
    \renewcommand{\arraystretch}{1.46}
    \begin{tabular}{|c|c||c|c|c|c||c|c|c|c|}
        \hline
        & & \multicolumn{4}{c||}{Patient 1} & \multicolumn{4}{c|}{Patient 2} \\ \hline
        $n$ & Model & Bias & MC Error & SE MC & Avg. SE & Bias & MC Error & SE MC & Avg. SE \\ \hline
        \multirow{4}{*}{$100$} & $(Y, A, T)$ & 2.192 & 1.636 & 15.945 & 16.393 & --0.442 & 1.668 & 16.256 & 15.275 \\ \cline{2-10}
        & $(Y, A)$ & 8.672 & 1.568 & 15.287 & 15.888 & 6.725 & 1.589 & 15.491 & 14.999 \\ \cline{2-10}
        & $(Y, T)$ & 6.919 & 1.633 & 15.913 & 16.062 & 5.033 & 1.574 & 15.341 & 14.876 \\ \cline{2-10}
        & $(Y)$ & 13.132 & 1.546 & 15.07 & 15.462 & 12.159 & 1.479 & 14.415 & 14.344 \\ \hline
        \multirow{4}{*}{$300$} & $(Y, A, T)$ & 0.796 & 0.950 & 9.256 & 10.235 & --0.567 & 1.039 & 10.13 & 9.503 \\ \cline{2-10}
        & $(Y, A)$ & 8.628 & 0.891 & 8.684 & 9.813 & 8.446 & 0.984 & 9.594 & 9.326 \\ \cline{2-10}
        & $(Y, T)$ & 6.077 & 0.930 & 9.069 & 9.942 & 5.638 & 0.983 & 9.579 & 9.208 \\ \cline{2-10}
        & $(Y)$ & 13.884 & 0.864 & 8.426 & 9.417 & 14.927 & 0.904 & 8.810 & 8.754 \\ \hline
    \end{tabular}
    \caption{Simulation results showcasing bias, Monte Carlo error (MC Error), standard error of the Monte Carlo sampling distribution (SE MC) and average standard error (Avg. SE) for both patient profiles under an \textbf{always treated} regime. All quantities have been multiplied by 100 in the table.}
    \label{table:simulation-results-at}
\end{table}

\begin{table}[!ht]
    \centering
    \renewcommand{\arraystretch}{1.46}
    \begin{tabular}{|c|c||c|c|c|c||c|c|c|c|}
        \hline
        & & \multicolumn{4}{c||}{Patient 1} & \multicolumn{4}{c|}{Patient 2} \\ \hline
        $n$ & Model & Bias & MC Error & SE MC & Avg. SE & Bias & MC Error & SE MC & Avg. SE \\ \hline
        \multirow{4}{*}{$100$} & $(Y, A, T)$ & 1.883 & 1.037 & 10.105 & 10.914 & --2.633 & 0.991 & 9.656 & 10.326 \\ \cline{2-10}
        & $(Y, A)$ & 1.455 & 1.036 & 10.102 & 10.947 & --1.963 & 0.984 & 9.590 & 10.319 \\ \cline{2-10}
        & $(Y, T)$ & --1.712 & 1.010 & 9.848 & 10.592 & --5.412 & 0.997 & 9.713 & 10.247 \\ \cline{2-10}
        & $(Y)$ & --1.885 & 1.029 & 10.029 & 10.59 & --4.561 & 1.000 & 9.743 & 10.180 \\ \hline
        \multirow{4}{*}{$300$} & $(Y, A, T)$ & 0.062 & 0.581 & 5.664 & 7.068 & --1.131 & 0.558 & 5.443 & 6.776 \\ \cline{2-10}
        & $(Y, A)$ & --0.567 & 0.588 & 5.728 & 7.099 & --0.255 & 0.566 & 5.516 & 6.725 \\ \cline{2-10}
        & $(Y, T)$ & --4.175 & 0.560 & 5.456 & 6.867 & --4.459 & 0.562 & 5.480 & 6.643 \\ \cline{2-10}
        & $(Y)$ & --4.547 & 0.562 & 5.478 & 6.811 & --3.290 & 0.569 & 5.545 & 6.585 \\ \hline
    \end{tabular}
    \caption{Simulation results showcasing bias, Monte Carlo error (MC Error), standard error of the Monte Carlo sampling distribution (SE MC) and average standard error (Avg. SE) for both patient profiles under an \textbf{never treated} regime. All quantities have been multiplied by 100 in the table.}
    \label{table:simulation-results-nt}
\end{table}

\clearpage

\subsection{Using Data Generated With Partially Correlated and Uncorrelated Random Effects}
\label{appendix-1:simulation-results-under-partial-models}

\begin{figure}[!ht]
    \centering
    \begin{minipage}[b]{0.49\textwidth}
        \includegraphics[width=1.1\textwidth]{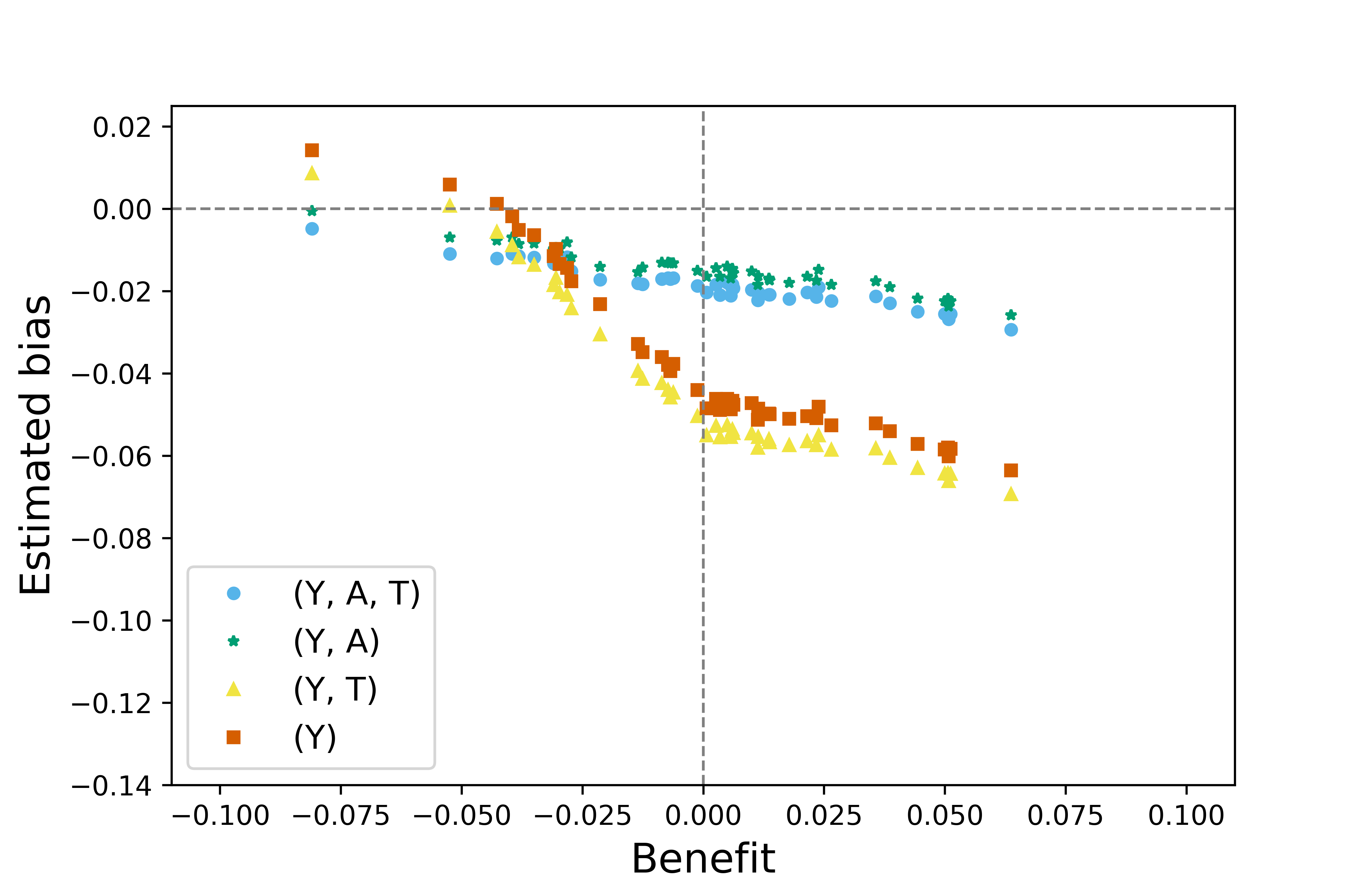}
    \end{minipage}
    \hfill
    \begin{minipage}[b]{0.49\textwidth}
        \includegraphics[width=1.1\textwidth]{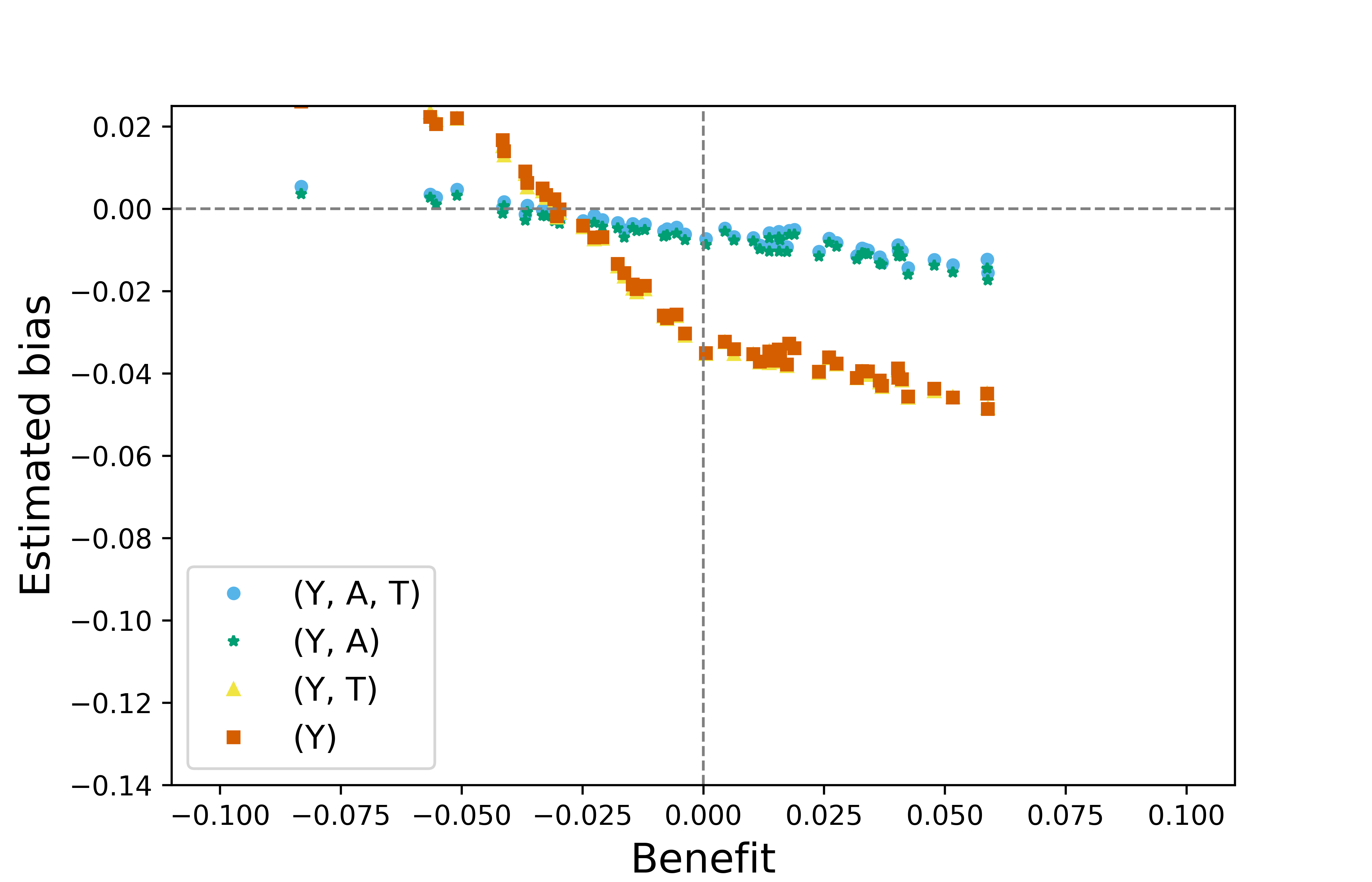}
    \end{minipage}
    \label{fig:benefit-vs-estimated-loss-ya}
    \caption{Benefit vs. estimated loss for all four model specifications given previous outcome -- $y_{i1} = 0$ (left) and $y_{i1} = 1$ (right) -- and data generated under $(\uiw, \uia) \indep \uit.$}

    \begin{minipage}[b]{0.49\textwidth}
        \includegraphics[width=1.1\textwidth]{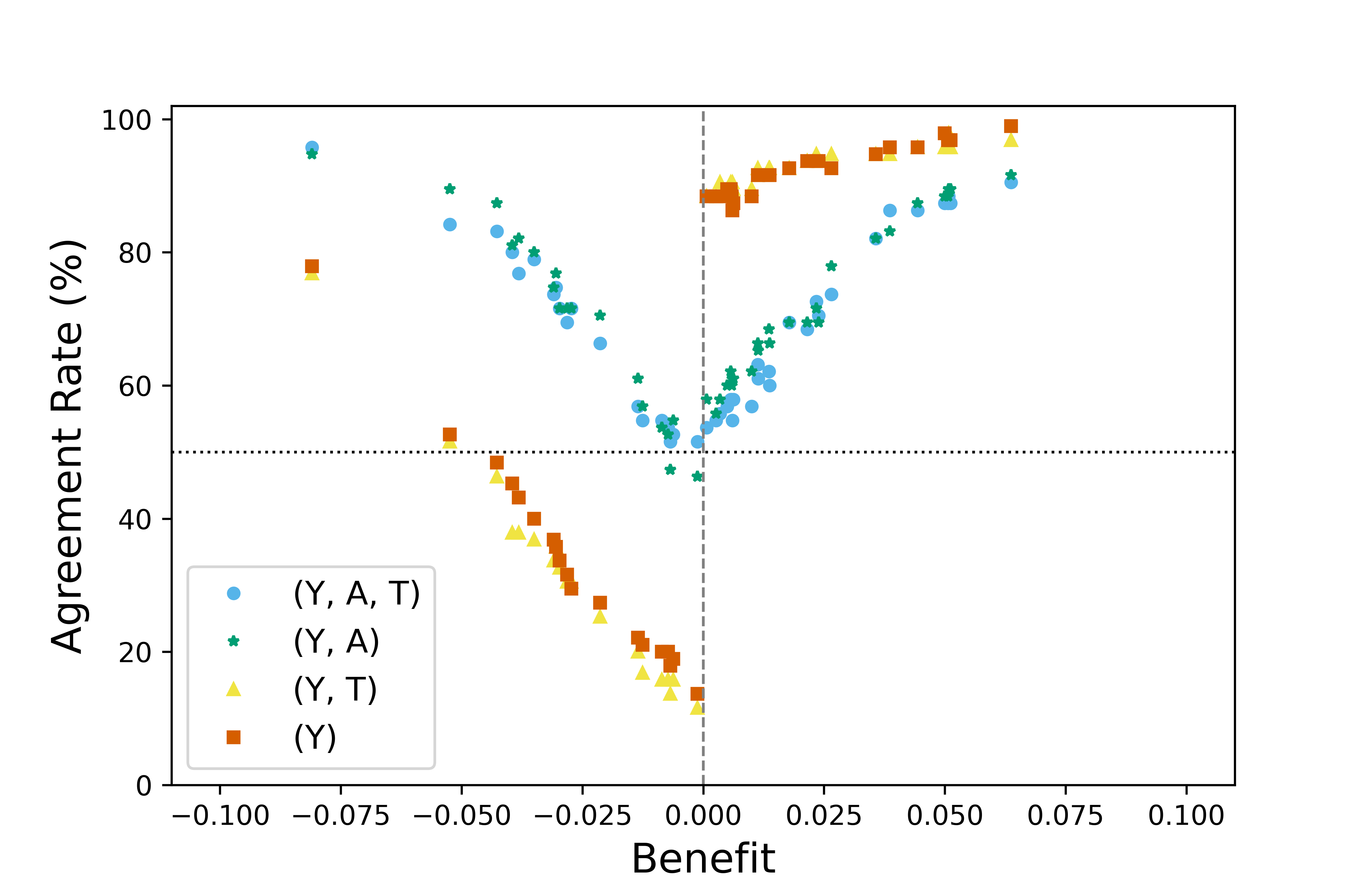}
    \end{minipage}
    \hfill
    \begin{minipage}[b]{0.49\textwidth}
        \includegraphics[width=1.1\textwidth]{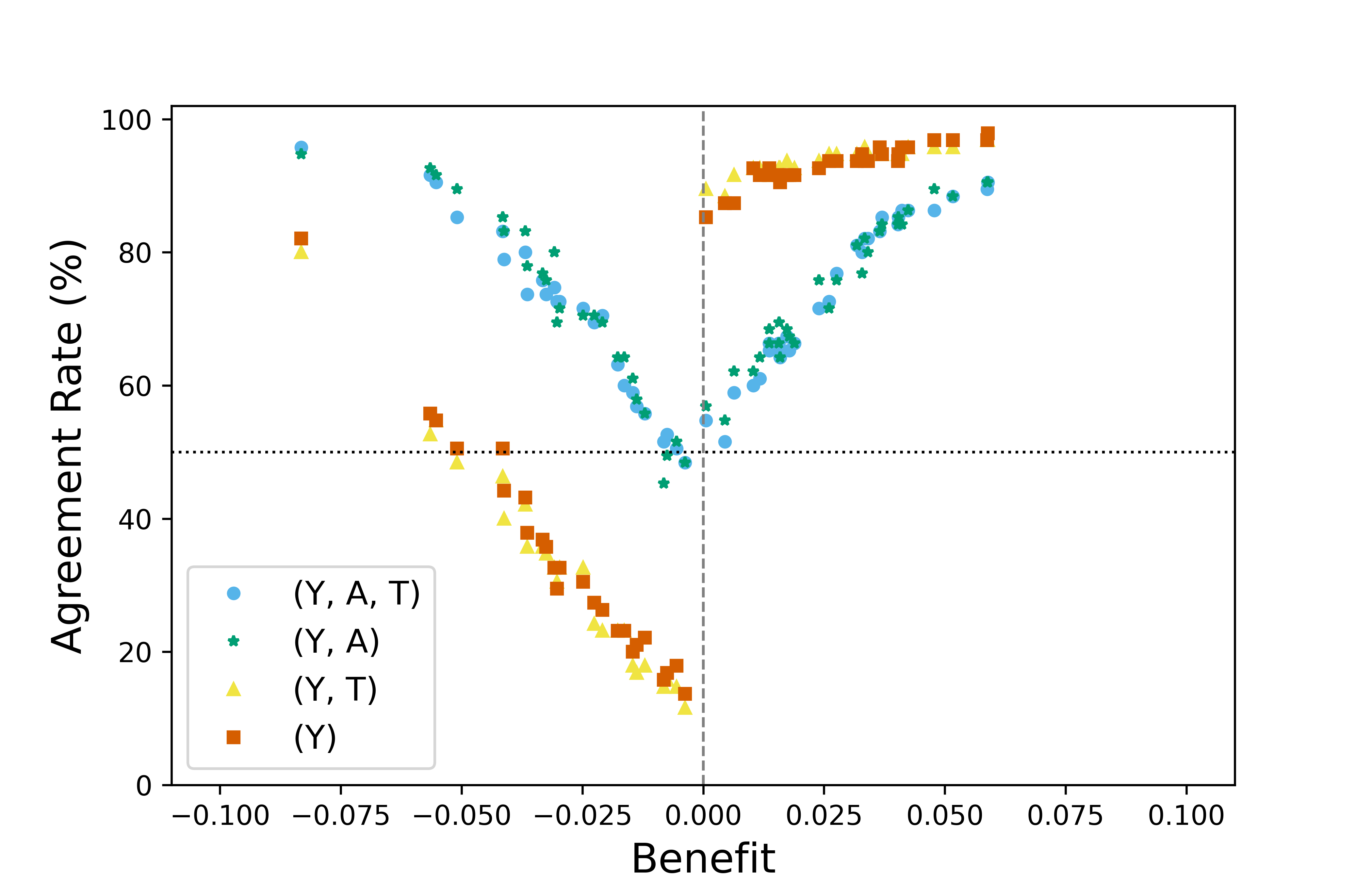}
    \end{minipage}
    \label{fig:benefit-vs-agreement-rate-ya}
    \caption{Benefit vs. agreement rate for all four model specifications given previous outcome -- $y_{i1} = 0$ (left) and $y_{i1} = 1$ (right) -- and data generated under $(\uiw, \uia) \indep \uit.$}
\end{figure}

\begin{figure}[!ht]
    \centering
    \begin{minipage}[b]{0.49\textwidth}
        \includegraphics[width=1.1\textwidth]{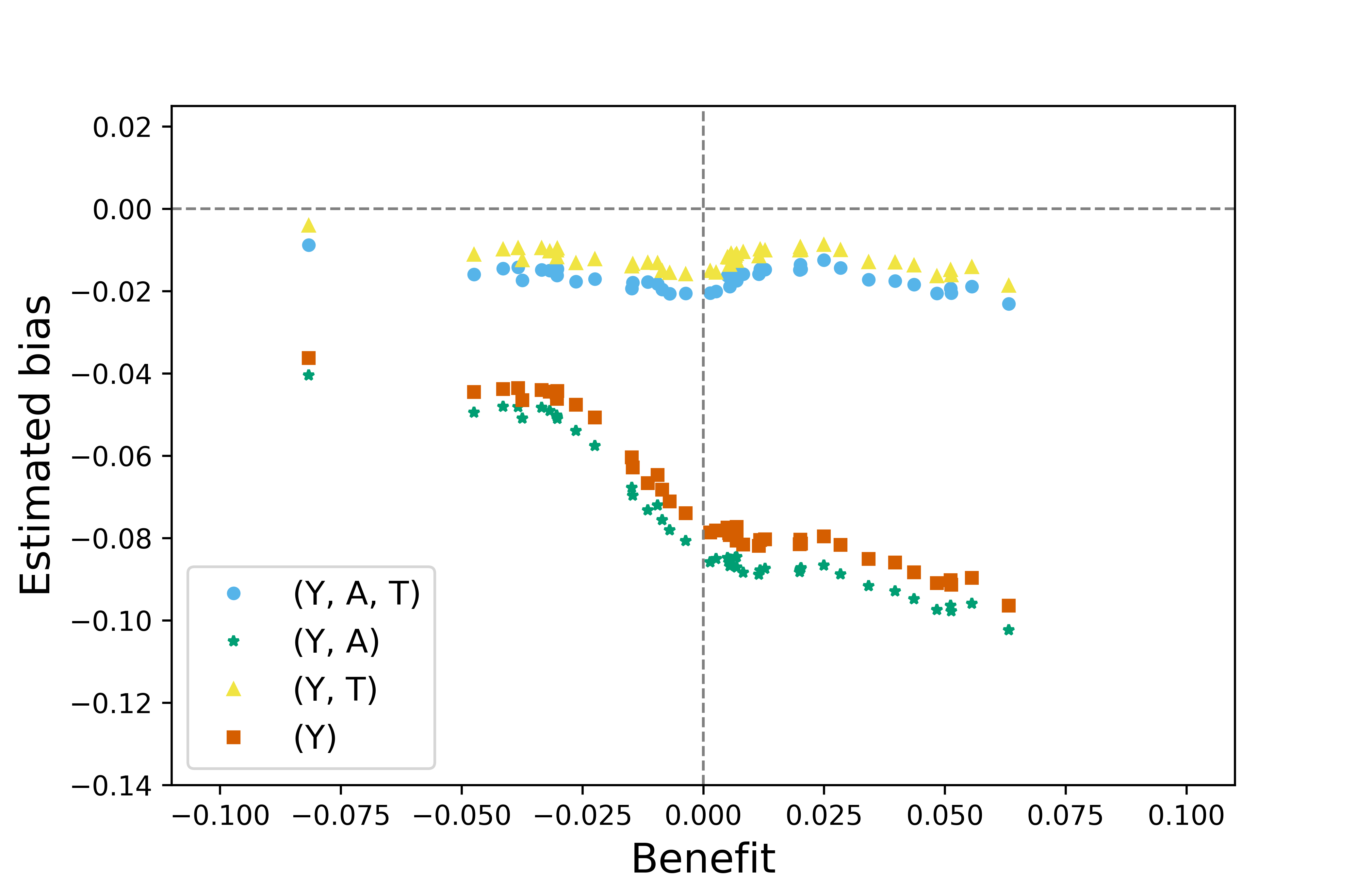}
    \end{minipage}
    \hfill
    \begin{minipage}[b]{0.49\textwidth}
        \includegraphics[width=1.1\textwidth]{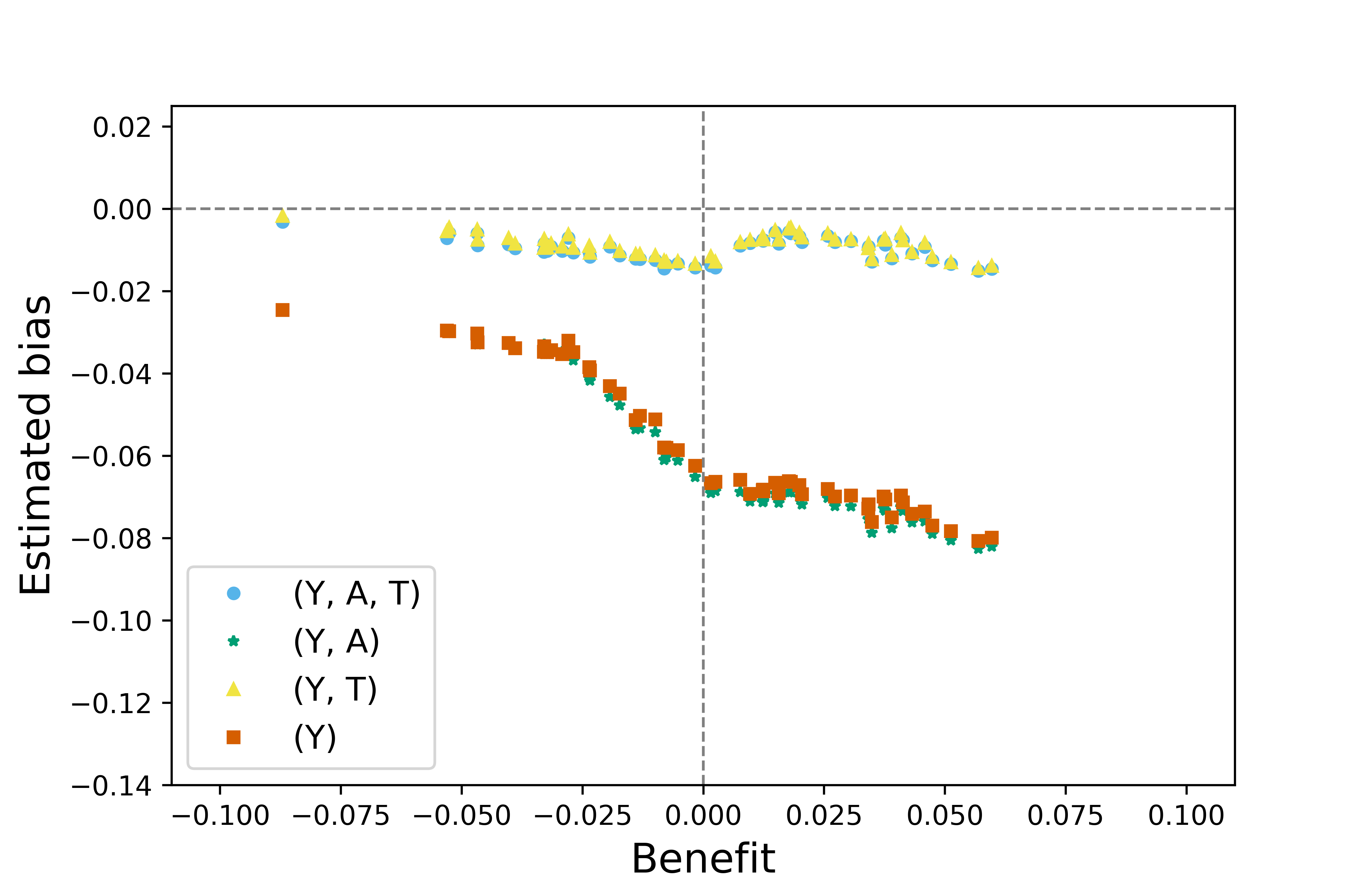}
    \end{minipage}
    \label{fig:benefit-vs-estimated-loss_y1_yt}
    \caption{Benefit vs. estimated loss for all four model specifications given previous outcome -- $y_{i1} = 0$ (left) and $y_{i1} = 1$ (right) -- and data generated under $(\uiw, \uit) \indep \uia.$}

    \begin{minipage}[b]{0.49\textwidth}
        \includegraphics[width=1.1\textwidth]{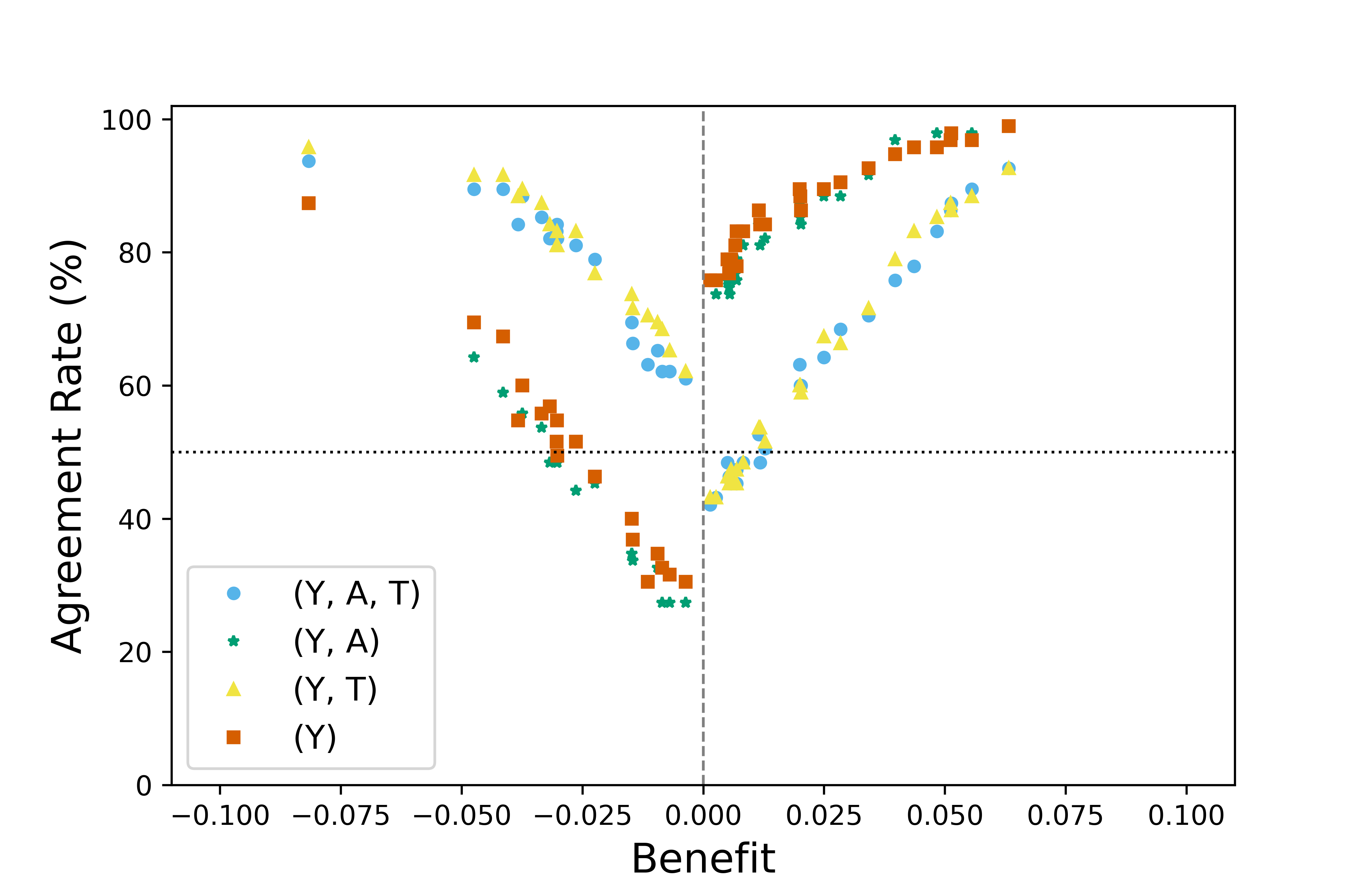}
    \end{minipage}
    \hfill
    \begin{minipage}[b]{0.49\textwidth}
        \includegraphics[width=1.1\textwidth]{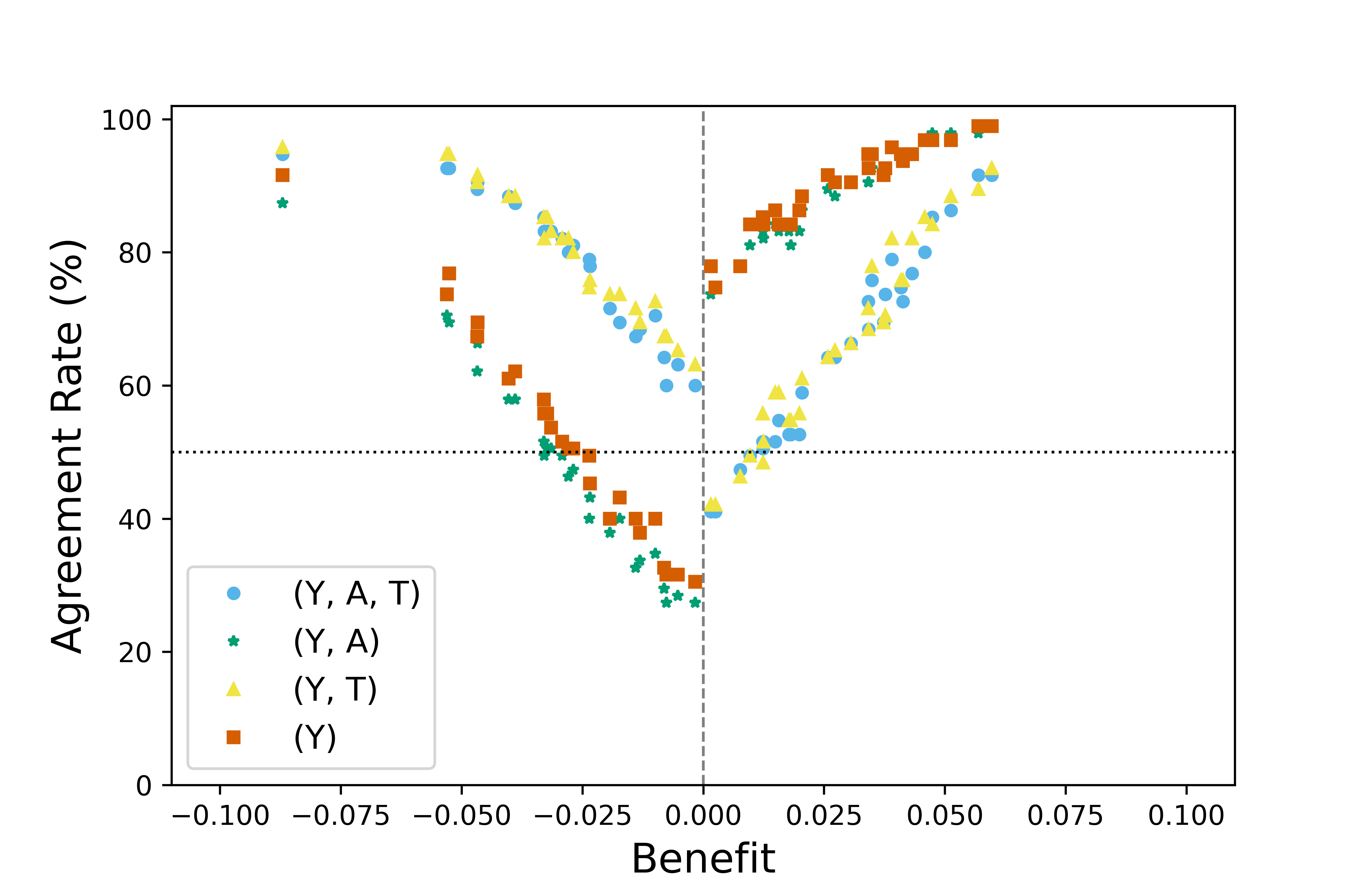}
    \end{minipage}
    \label{fig:benefit-vs-agreement-rate-yt}
    \caption{Benefit vs. agreement rate for all four model specifications given previous outcome -- $y_{i1} = 0$ (left) and $y_{i1} = 1$ (right) -- and data generated under $(\uiw, \uit) \indep \uia.$}
\end{figure}

\begin{figure}[!ht]
    \centering
    \begin{minipage}[b]{0.49\textwidth}
        \includegraphics[width=1.1\textwidth]{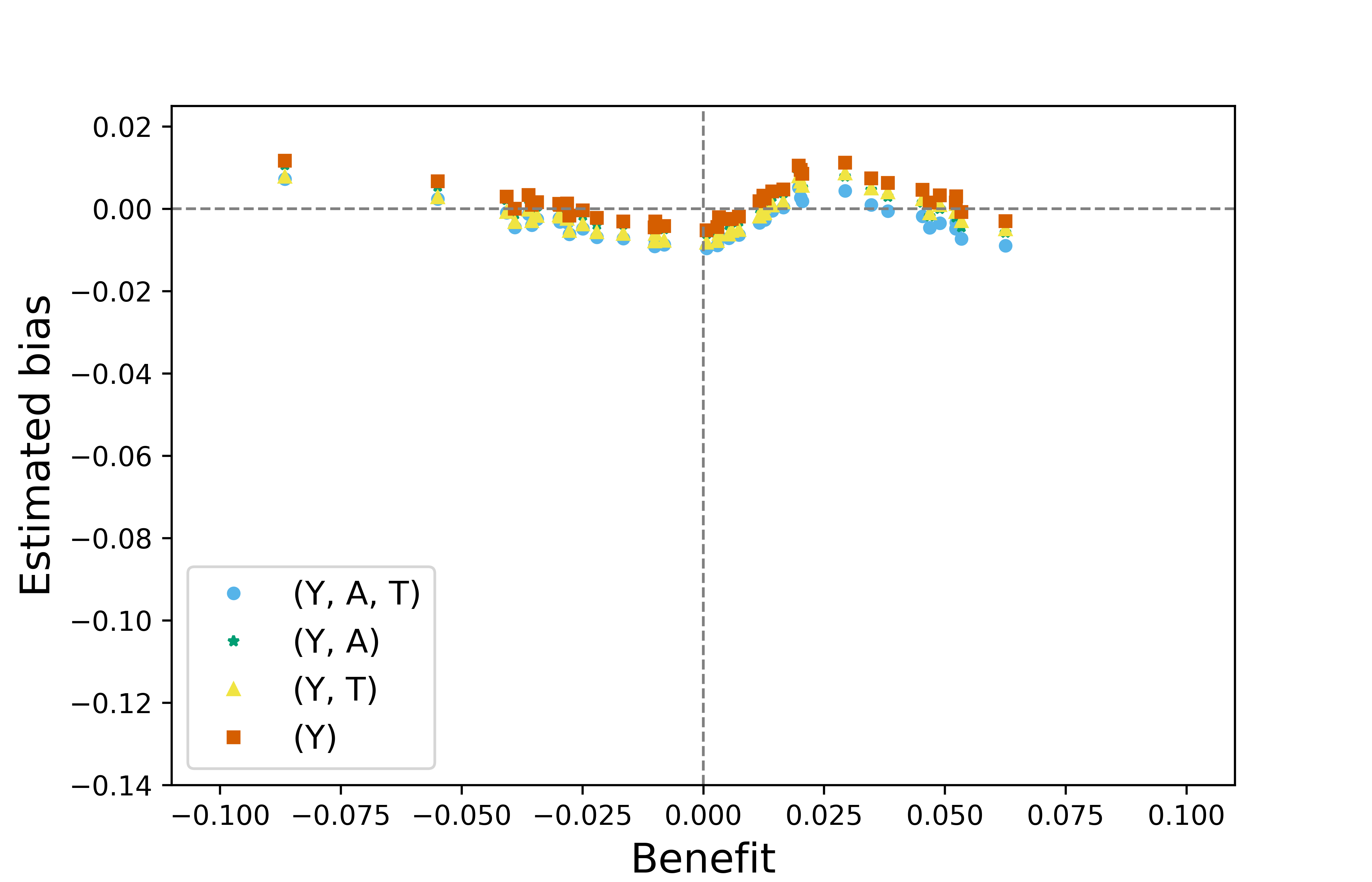}
    \end{minipage}
    \hfill
    \begin{minipage}[b]{0.49\textwidth}
        \includegraphics[width=1.1\textwidth]{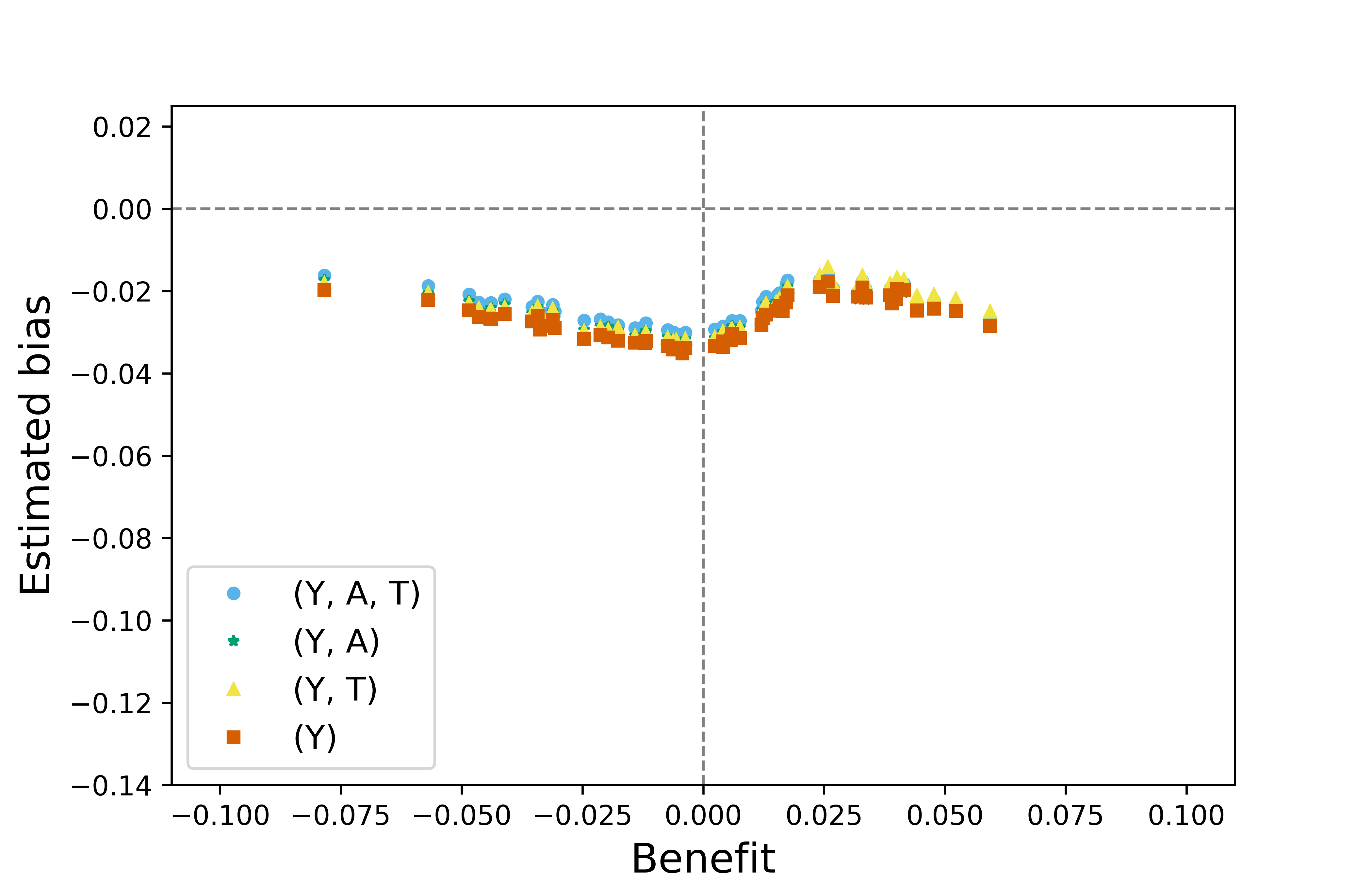}
    \end{minipage}
    \label{fig:benefit-vs-estimated-independent}
    \caption{Benefit vs. estimated loss for all four model specifications given previous outcome -- $y_{i1} = 0$ (left) and $y_{i1} = 1$ (right) -- and data generated under $\uiw \indep \uia \indep \uit$.}

    \begin{minipage}[b]{0.49\textwidth}
        \includegraphics[width=1.1\textwidth]{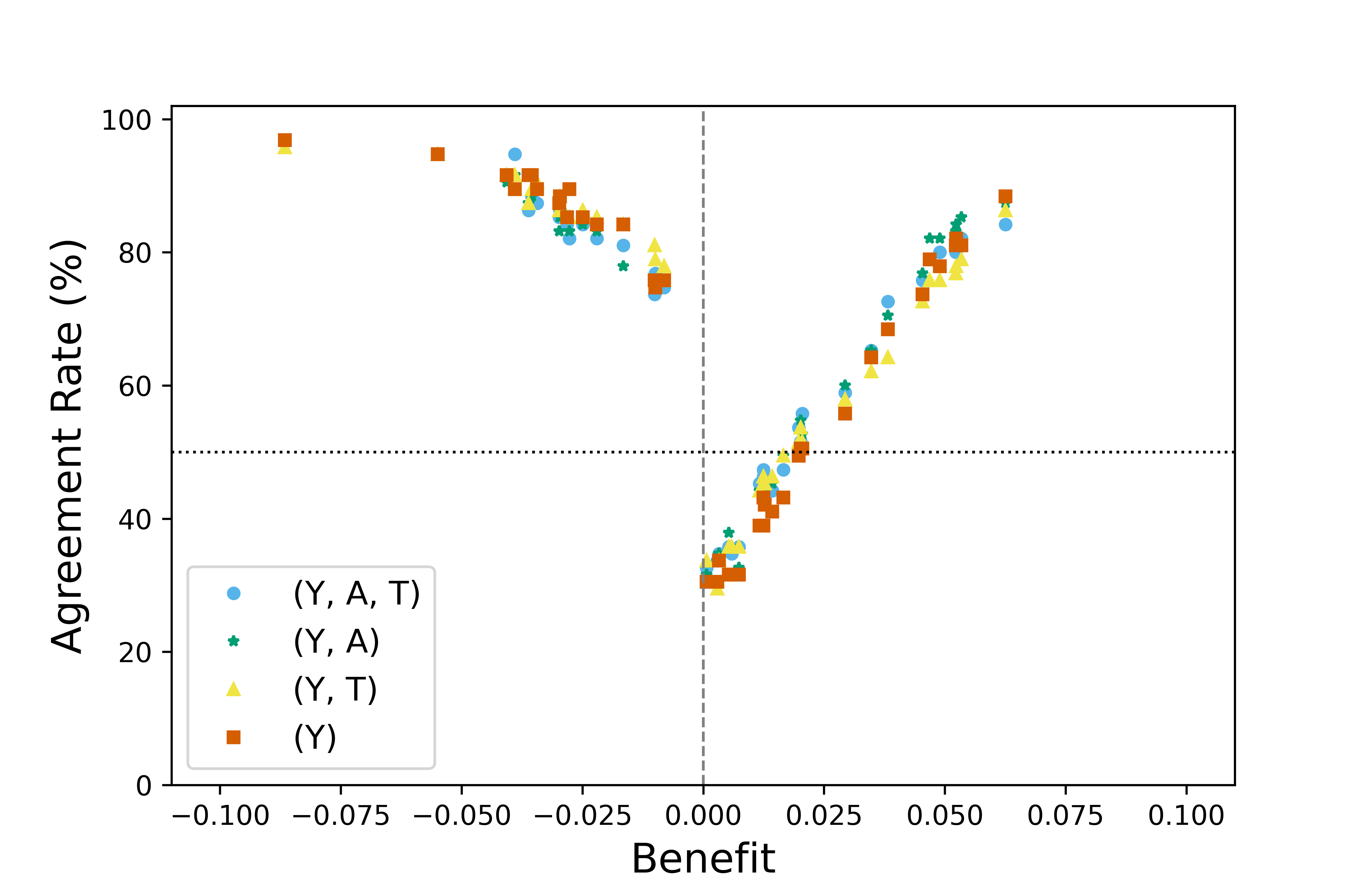}
    \end{minipage}
    \hfill
    \begin{minipage}[b]{0.49\textwidth}
        \includegraphics[width=1.1\textwidth]{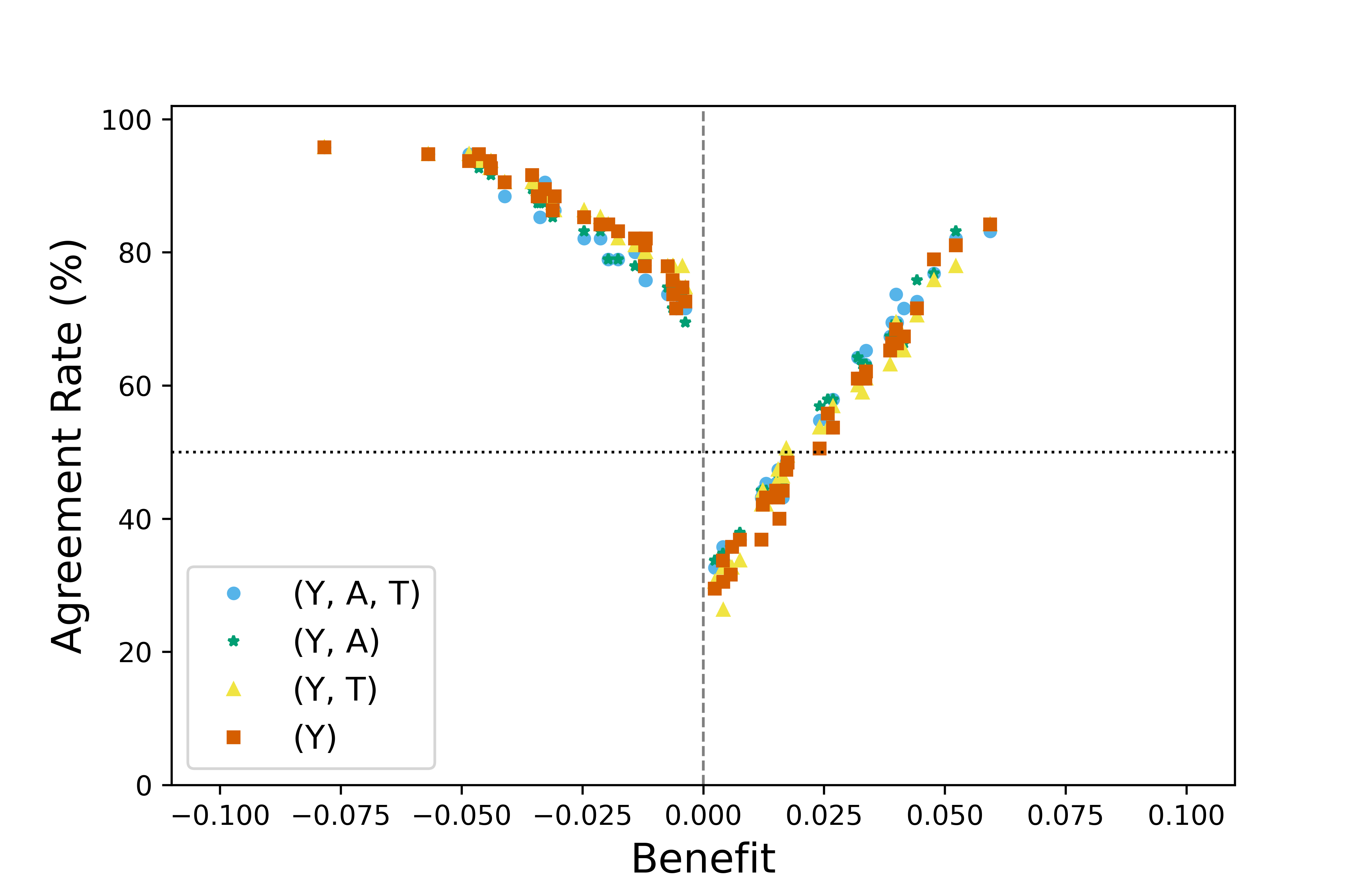}
    \end{minipage}
    \label{fig:benefit-vs-agreement-rate-independent}
    \caption{Benefit vs. agreement rate for all four model specifications given previous outcome -- $y_{i1} = 0$ (left) and $y_{i1} = 1$ (right) -- and data generated under $\uiw \indep \uia \indep \uit$.}
\end{figure}

%% file: appendix/appendix-G-data-analysis.tex
\section{Details on Data Analysis}

\subsection{Protocols of Original INSPIRE Studies and Target Trial}
\label{appendix-1:inspire-protocol}

\begin{table}[!hb]
    \begin{center}
        \renewcommand{\arraystretch}{1.5}
        \begin{tabular}{|p{2.3cm}|p{4cm}|p{4cm}|p{4cm}|}
            \hline
            \textbf{Protocol Component} & \textbf{INSPIRE 2 \& 3 Protocols} & \textbf{Observational Setting} & \textbf{Target Trial Specification} \\ \hline
            Eligibility Criteria & HIV-infected adults on HAART & \textit{Idem} & \textit{Idem} \\ \hline
            Treatment strategy & Three injections of IL-7 administered at 7-day intervals & One, two or three injections of IL-7 administered at 7-day intervals & Three injections of IL-7 administered at 7-day intervals \\ \hline
            Treatment assignment & Provide IL-7 doses if CD4 blood concentration $\leq$ 550 cells/$\mu$L with a maximum of 4 injection cycles over 21 months of follow-up & 0, 1, 2 or 3 injections were provided; when CD4 $\leq$ 550 cells/$\mu$L, at least one injection was often provided, but not always. & 0 or 3 injections randomly provided, at 90-day intervals, with a maximum of three cycles across follow-up \\ \hline
            Follow-up schedule & Quarterly intervals between assessments at which injections are decided to be administered or not & Some departure in protocol, both in the number of injections and scheduled visit time, can be observed. & Exactly 90-day intervals \\ \hline
            Outcome & post-injection CD4 concentration & post-injection CD4 concentration & indicator of CD4 $\geq$ 500 cells/$\mu$L of blood \\ \hline
            Target of inference & Time spent with CD4 blood concentration $\geq 500$ cells/$\mu$L & \textit{Idem} & Conditional rewards under optimal regime \\ \hline
            % Analysis plan & Intention-to-treat & Use of mechanistic models to estimate the CD4 trajectory & Use of Bayesian joint modeling to estimate treatment effects with irregularly observed data
            % and adapted G-computation (c.f.~\ref{theorem:adapted-g-computation}) to estimate conditional rewards and optimal regime\\ \hline
        \end{tabular}
        \caption{Component-wise specifications of the INSPIRE 2 \& 3 studies and the completely randomized SMART as $\calE$ that we wish to emulate. Deviations from the INSPIRE protocol may result in confounding and informative visit time biases, rendering the data observational.}
    \end{center}
    \label{table:inspire-study-design}
\end{table}

\newpage

\subsection{Frequency of Individualized Treatment Frequency}
\label{appendix-individualized-treatment-frequency.}

\begin{table}[!ht]
    \centering
    \renewcommand{\arraystretch}{1.5}
    \begin{tabular}{|c|c|}
        \hline
        Treatment Course & Recommendation Frequency\\ \hline
        $(0, 1, 0, 0) $ & 67.9\% \\
        $(1, 0, 0, 0)$ & 29.1\%\\
        $(0, 0, 1, 0)$ & 3.0\%\\
        \hline
    \end{tabular}
    \caption{Patient A individualized treatment recommendation.}
    \label{tab:patient-a-individualized treatments}
\end{table}

\begin{table}[!ht]
    \centering
    \renewcommand{\arraystretch}{1.5}
    \begin{tabular}{|c|c|}
        \hline
        Treatment course & Recommendation frequency\\ \hline
        $(1, 1, 0, 0)$ & 53.5\% \\
        $(1, 0, 1, 0)$ & 33.8\% \\
        $(0, 1, 1, 0)$ & 11.3\% \\
        $(0, 1, 0, 1)$ & 0.6\%\\
        $(1, 0, 0, 1)$ & 0.4\%\\
        $(0, 0, 1, 1)$ & 0.3\%\\
        \hline
    \end{tabular}
    \caption{Patient B individualized treatment recommendation.}
    \label{tab:patient-b-individualized-treatments}
\end{table}

\clearpage

\subsection{MCMC Chains}
\label{appendix-1:mcmc-traceplots}

Traceplots and empirical density plots for all four chains for reward parameters $\phiy$, $\phiu$, $\mu$ and $\tau$ for the full joint model $(Y, A, T)$ are respectively displayed in Figures~\ref{fig:phi-y-traceplot},~\ref{fig:phi-u-traceplot} and~\ref{fig:mu-tau-traceplot} below.

\begin{figure}[!b]
    \centering
    \includegraphics[width=\linewidth]{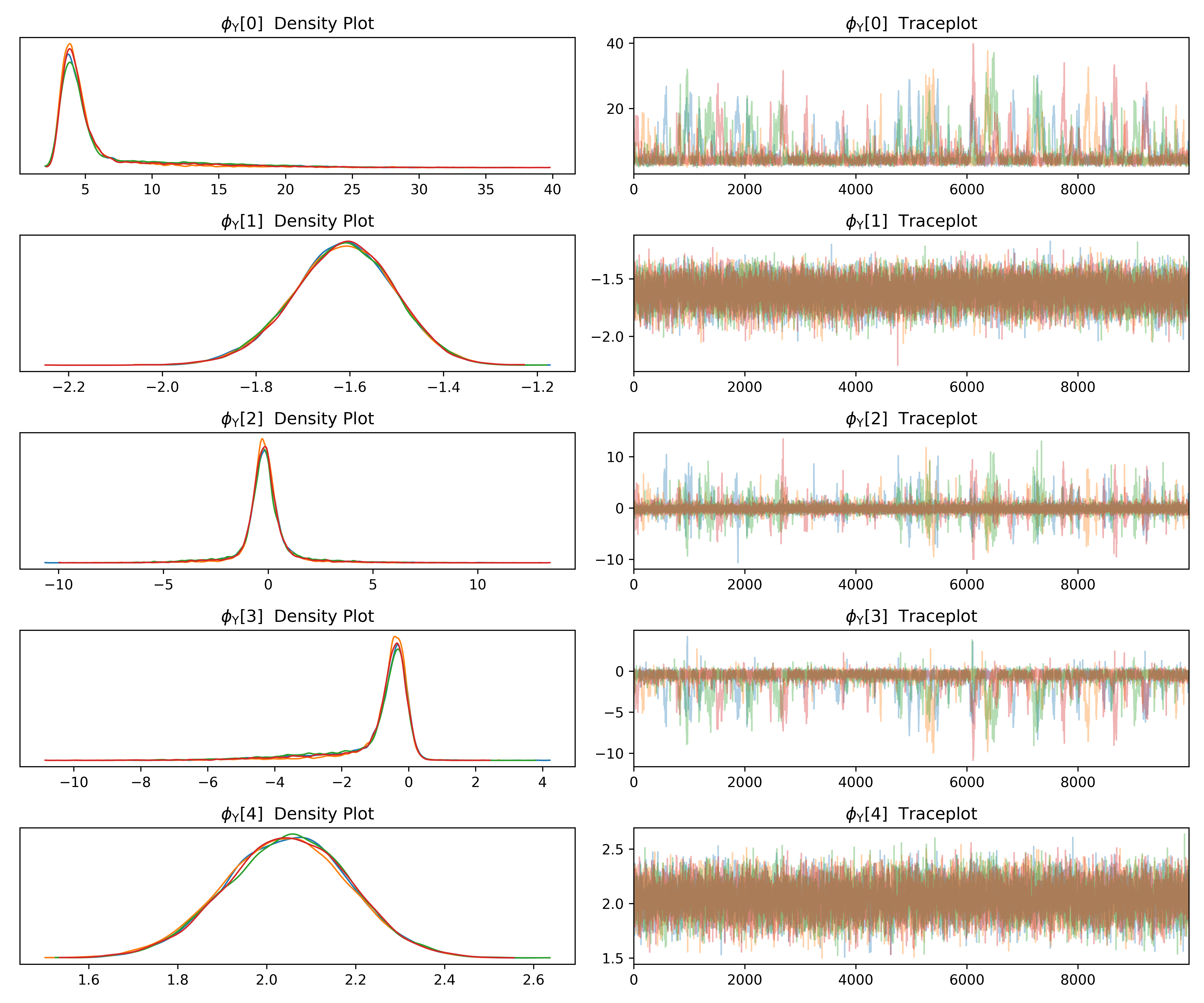}
    \caption{$\phiy$ traceplots}
    \label{fig:phi-y-traceplot}
\end{figure}

\begin{figure}
    \centering
    \includegraphics[width=0.8\linewidth]{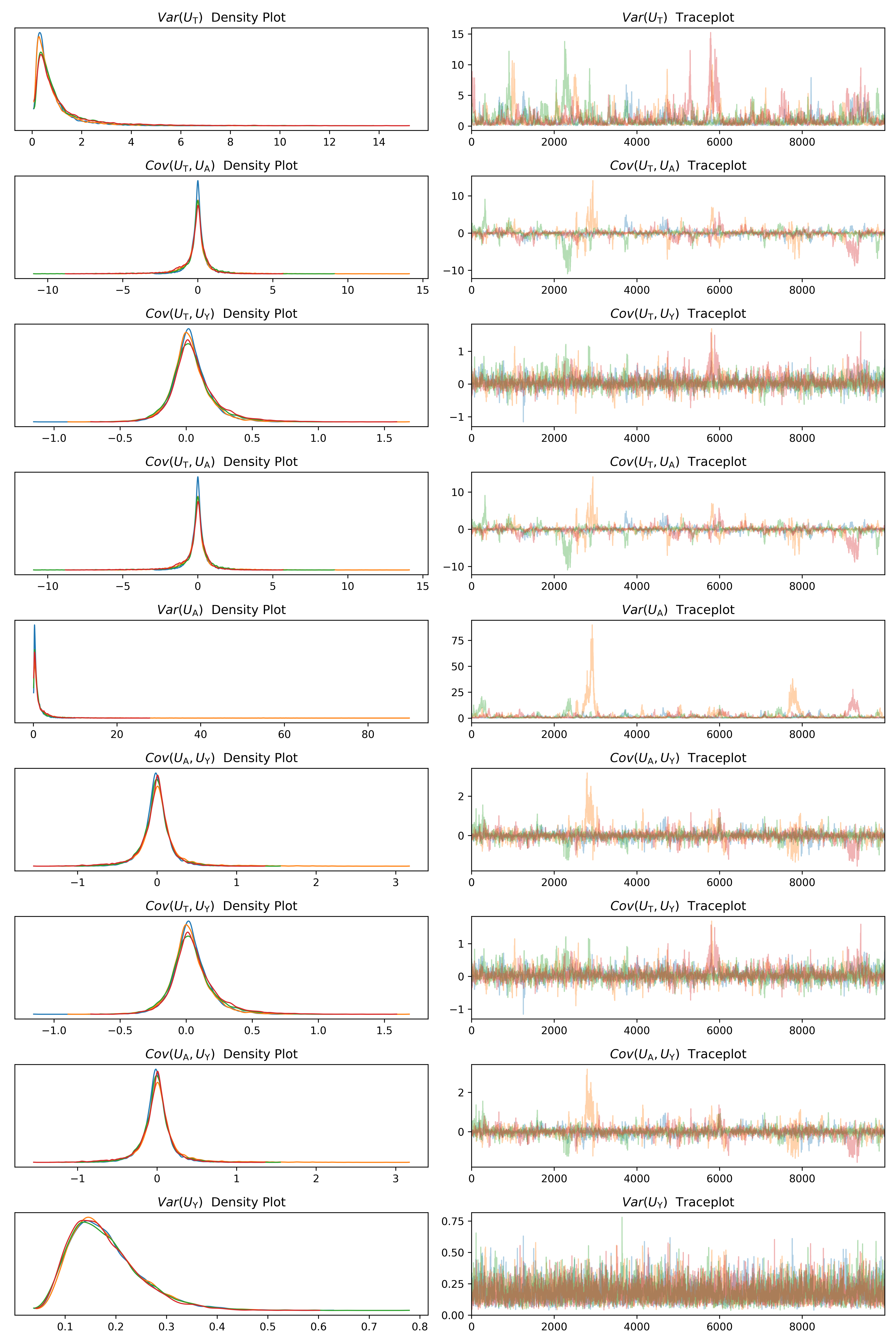}
    \caption{$\phiu$ traceplots}
    \label{fig:phi-u-traceplot}
\end{figure}

\begin{figure}
    \centering
    \includegraphics[width=0.8\linewidth]{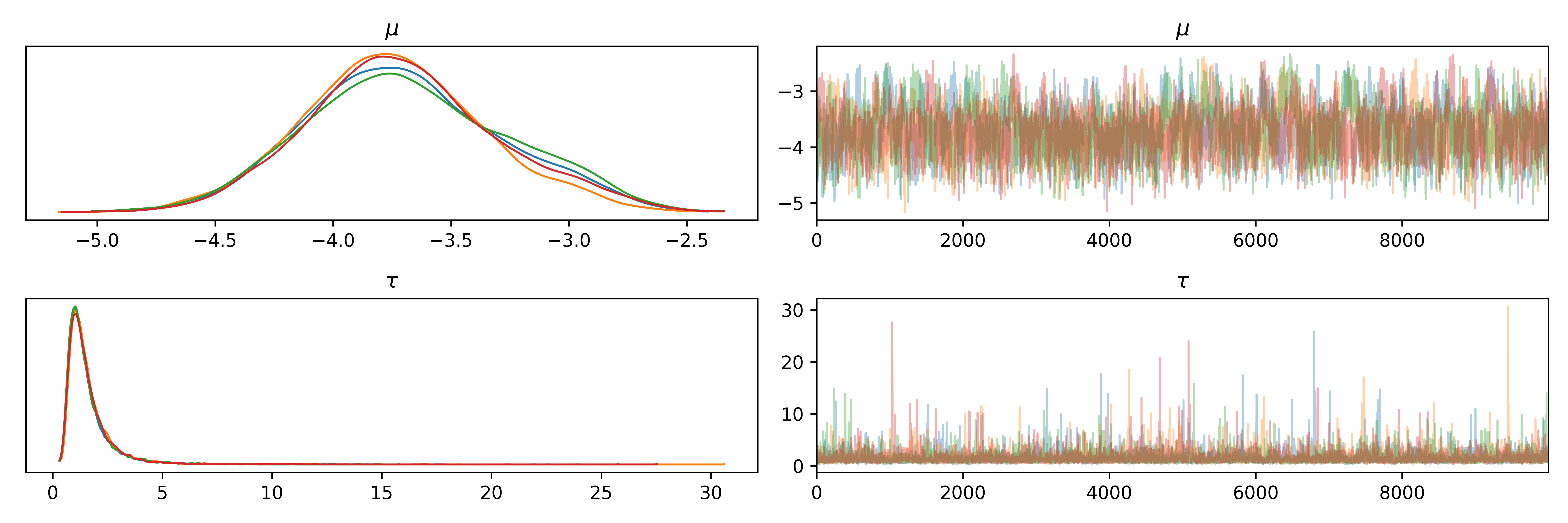}
    \caption{$\mu$ and $\tau^2$ traceplots}
    \label{fig:mu-tau-traceplot}
\end{figure}

\newpage

\subsection{Summary Table of Outcome Model Parameters}
\label{subsection:inspire-summary-table}

\begin{table}[!ht]
    \begin{center}
        \renewcommand{\arraystretch}{1.5}
        \begin{tabular}{|c|c|c|}
            \hline
            Coefficient & \multicolumn{1}{c|}{Estimate} & \multicolumn{1}{c|}{95\% Credible Interval}\\ \hline
            Intercept & --2.05 & (--2.48, --1.70) \\ 
            $A_{ij}^{\text{DR}}$ & \phantom{--}3.41 & (\phantom{--}2.83, \phantom{--}4.11) \\
            $A_{ij}^{\text{DR}} * \operatorname{Age}_i$ & --0.20 & (--0.60, \phantom{--}0.21) \\ 
            $A_{ij}^{\text{DR}} * \operatorname{BMI}_i$ & --0.18 & (--0.35, \phantom{--}0.16) \\ 
            $A_{ij}^{\text{DR}} * Y_{i(j-1)}$ & \phantom{--}1.48 & (\phantom{--}1.24, \phantom{--}1.73) \\
            $\mu$ & --5.24 & (--5.80, --4.68) \\
            $\tau^2$ & 0.82 & (\phantom{--}0.47, \phantom{--}1.64) \\
            \hline
        \end{tabular}
    \end{center}
    
    \caption{Summary of outcome model parameters from MCMC.}
    \label{table:inspire-summary-table}
\end{table}

\begin{table}[!ht]
    \begin{center}
        \renewcommand{\arraystretch}{1.5}
        \begin{tabular}{|c|c|c|c|}
            \hline
            & $\operatorname{Cov}(\cdot, \Ut)$ & $\operatorname{Cov}(\cdot, \Ua)$ & $\operatorname{Cov}(\cdot, \Uy)$\\
            \hline
            $\operatorname{Cov}(\Ut, \cdot)$ & \phantom{--}0.61\ (\phantom{--}0.14,\ \phantom{--}4.11) & --0.01\ (--2.43,\ \phantom{--}1.72) & \phantom{--}0.03\ (--0.25,\ \phantom{--}0.42) \\
            \hline
            $\operatorname{Cov}(\Ua, \cdot)$ & 
            --0.01\ (--2.43,\ \phantom{--}1.72) & \phantom{--}0.70\ (\phantom{--}0.14,\ \phantom{--}9.18) & --0.01\ (--0.45,\ \phantom{--}0.45) \\ \hline
            $\operatorname{Cov}(\Uy, \cdot)$ & \phantom{--}0.03\ (--0.25,\ \phantom{--}0.42) & --0.01\ (--0.45,\ \phantom{--}0.45) & \phantom{--}0.17\ (\phantom{--}0.08,\ \phantom{--}0.35) \\ \hline
        \end{tabular}
    \end{center}
    
    \caption{Summary of correlation matrix fit from MCMC.}
    \label{table:inspire-correlation-matrix}
\end{table}